\documentclass[hyperfootnotes=false,letterpaper]{article}
\usepackage[english]{babel}
\usepackage{cite}

\usepackage[utf8]{inputenc}
\usepackage{fullpage}
\usepackage{amsmath, amsfonts, amssymb, amsthm}
\usepackage[linesnumbered,noresetcount,vlined,ruled,procnumbered]{algorithm2e}

\usepackage{color}
\usepackage{xcolor}
\definecolor{darkgreen}{rgb}{0,0.5,0}
\definecolor{darkblue}{rgb}{0,0,0.6}
\definecolor{DarkRed}{rgb}{0.65,0,0}
\definecolor{Red}{rgb}{1,0,0}

\usepackage[linktocpage=true,backref=page,
pagebackref=true,colorlinks, linkcolor=darkblue,citecolor=darkgreen,urlcolor=blue,filecolor=magenta]{hyperref}

\usepackage[capitalize, nameinlink]{cleveref}

\usepackage{graphicx}
\usepackage{comment}
\usepackage{thm-restate}
\usepackage{bbm}
\usepackage{tikz}
\usetikzlibrary{shapes.geometric, arrows.meta, positioning, bending,calc}
\usepackage[linewidth=1pt]{mdframed}
\usepackage[most]{tcolorbox}
\definecolor{block-gray}{gray}{0.85}
\newtcolorbox{shadequote}{colback=block-gray,grow to right by=-2mm,grow to left by=-2mm,
boxrule=0pt,boxsep=0pt,breakable}
\usepackage{color}
\usepackage[inline]{enumitem}

\usepackage{caption}
\usepackage{cancel}
\usepackage{subcaption}
\usepackage{mathrsfs}
\usepackage{xspace}

\usepackage{mathtools} 
\usepackage{booktabs} 
\usepackage{makecell}
\usepackage{tabularx}

\SetKwComment{Comment}{$\quad\triangleright$\ }{}
\newcommand{\nosemic}{\renewcommand{\@endalgocfline}{\relax}}

\usepackage[normalem]{ulem}

\crefalias{AlgoLine}{line}

\makeatletter
\let\cref@old@stepcounter\stepcounter
\def\stepcounter#1{\cref@old@stepcounter{#1}\cref@constructprefix{#1}{\cref@result}\@ifundefined{cref@#1@alias}{\def\@tempa{#1}}{\def\@tempa{\csname cref@#1@alias\endcsname}}\protected@edef\cref@currentlabel{[\@tempa][\arabic{#1}][\cref@result]\csname p@#1\endcsname\csname the#1\endcsname}}
\makeatother

\usepackage{nicefrac}
\usepackage{pgf}

\newif\ifnotes
\notestrue
\ifnotes
    \newcommand{\keren}[1]{{\ifnotes \color{blue}{Keren: #1}
                \fi}}
    \newcommand{\tomer}[1]{{\ifnotes \color{red!65!black}{Tomer: #1}
                \fi}}
    \newcommand{\virgi}[1]{{\ifnotes \color{red}{Virgi: #1}
                \fi}}
\else
    \newcommand{\keren}[1]{}
    \newcommand{\tomer}[1]{}
    \newcommand{\virgi}[1]{}
\fi

\newtheorem{theorem}{Theorem}
\newtheorem{lemma}[theorem]{Lemma}
\newtheorem{remark}{Remark}

\newtheorem{proposition}[theorem]{Proposition}
\newtheorem{claim}[theorem]{Claim}
\newtheorem{corollary}[theorem]{Corollary}

\newtheorem{hypothesis}{Hypothesis}

\makeatletter
\renewenvironment{proof}[1][\proofname]{\par
    \pushQED{\qed}\normalfont \topsep6\p@\@plus6\p@\relax
    \trivlist
    \item\relax
    {\bfseries\boldmath
        #1\@addpunct{.}}\hspace\labelsep\ignorespaces
}{\popQED\endtrivlist\@endpefalse
}
\makeatother

\newcommand{\apm}{(1\pm\eps)}
\newcommand{\apmp}{(1\pm\eps')}

\newcommand{\MM}[1]{\mathsf{MM}\brak{#1}}

\newcommand{\brak}[1]{\left(#1\right)}
\newcommand{\brk}[1]{(#1)}
\newcommand{\sbrak}[1]{\left[#1\right]}
\newcommand{\Exp}[1]{\mathbb{E}\left[ #1 \right]}
\newcommand{\Var}[1]{\mathsf{Var}\left[ #1 \right]}
\renewcommand{\Pr}[1]{{\mathrm{Pr}}\left[ #1 \right]}
\newcommand{\set}[1]{\left\{ #1 \right\}}
\newcommand{\sset}[1]{\{#1\}}
\newcommand{\polylog}[1]{\mathrm{polylog}\brak{#1}}
\newcommand{\Bin}[2]{\mathsf{Bin}\brak{{#1},{#2}}}

\newcommand{\mo}{O}
\newcommand{\BO}[1]{\mo\brak{#1}}
\newcommand{\TO}[1]{\tilde{\mo}(#1)}
\newcommand{\Omc}[1][1]{\Omega\brak{#1}}
\newcommand{\zrn}[1]{\set{0,1,\ldots,#1}}
\newcommand{\zrnone}[1]{\set{1,2,\ldots,#1}}

\newcommand{\eps}{\varepsilon}
\newcommand{\Em}[1]{\mathcal{E}_{#1}}
\newcommand{\EE}[1]{\mathcal{E}(#1)}
\newcommand{\UU}{\mathcal{U}}
\newcommand{\VV}{\mathcal{V}}

\newcommand{\PP}{\mathcal{P}}

\newcommand{\XC}{\mathcal{X}}

\newcommand{\AB}{\textbf{A}}

\newcommand{\FF}{\mathcal{F}}

\newcommand{\fpr}[2]{1-\frac{{#1}}{n^{#2}}}

\DeclareMathSymbol{\mhyphen}{\mathord}{AMSa}{"39}

\makeatletter
\newcommand{\whp}{w.h.p.\@ifnextchar.{\@gobble}{\xspace}}
\makeatother

\DeclarePairedDelimiter\abs{\lvert}{\rvert}

\DeclarePairedDelimiter{\ceil}{\lceil}{\rceil}

\makeatletter
\let\oldabs\abs
\def\abs{\@ifstar{\oldabs}{\oldabs*}}
\DeclarePairedDelimiter{\floor}{\lfloor}{\rfloor}

\crefname{claim}{Claim}{Claims} \crefname{ineq}{inequality}{inequalities} \crefname{proof}{Proof}{Proofs} \crefname{line}{Line}{Lines}
\crefname{algorithm}{Algorithm}{Algorithms}
\crefname{black box}{Procedure}{Procedure}
\crefname{figure}{Figure}{Figures}

\newcommand{\DD}{\mathcal{D}}
\newcommand{\DDF}{\hat{\mathsf{D}}(\Lambda q)}

\newcommand{\haT}{\hat{m}}
\newcommand{\htt}{\hat{m}}

\newcommand{\neweps}{\frac{\min\set{1,\eps}}{4\log n}}

\newcommand{\rr}{400\log n}

\newcommand{\clog}{\log^3 n}

\newcommand{\Prod}{\mathsf{Product}_k}

\newcommand{\Prodd}{\mathsf{Product}}

\newcommand{\eo}{1-1/e}

\newcommand{\nemp}{\neq\emptyset}

\newcounter{blackbox}[section] \newenvironment{blackbox}[1][]{\refstepcounter{blackbox}\begin{mdframed}[
    frametitle=#1,
    frametitlealignment=\centering,
    backgroundcolor=gray!10,
    linewidth=1pt,
    innerleftmargin=8pt,
    innerrightmargin=8pt,
    innertopmargin=10pt,
    innerbottommargin=10pt,
    splittopskip=\topskip,
    skipabove=\topsep,
    skipbelow=\topsep,
    userdefinedwidth=0.98\textwidth
  ]}{\end{mdframed}}

\crefname{blackbox}{Procedure}{Procedure}

\newcommand{\poly}{\mathsf{poly}}

\newcommand{\ZZ}{A}

\newcommand{\Vin}{V^{\mathrm{input}}}

\renewcommand{\clog}{\log^{(k^2)}n}
\newcommand{\cclog}{(k\log (4n))^{k^2}}

\newcommand{\cl}{\mathscr{T}}

\newcommand{\llow}{\Lambda_{\mathrm{low}}}
\newcommand{\gs}{\mu}

\newcommand{\vl}{V_\Lambda}
\newcommand{\hmvl}{\htt_{\vl}}

\newcommand{\bbot}{``\W \text{ is too large}''}
\newcommand{\apmt}{(1\pm\eps/2)}

\newcommand{\ple}{O(\poly(\log n,\tfrac1\eps))}
\newcommand{\hdv}{\hat{d}_v}

\newcommand{\BB}{\text{Hyperedge-Oracle}\xspace}
\newcommand{\HO}{\BB}
\newcommand{\CL}{\texttt{Duplicatable}\xspace}
\newcommand{\W}{L}

\newcommand{\Qd}{Q}
\newcommand{\Qdval}{5k\cdot 2^k\cdot \log^3 n}
\newcommand{\K}{\gamma}

\newcommand{\algDeg}{\mathbf{Deg\mhyphen Approx}_{\eps}}

\newcommand{\algCore}[1]{\arec\brak{#1}}

\newcommand{\algWrapD}[1]{\crec\brak{#1}}
\newcommand{\algWrapDNP}{\crec}
\newcommand{\apx}{\mathbf{GuessApx}}

\newcommand{\algPrev}{\mathbf{BBGM\mhyphen Approx}}

\newcommand{\arec}{\mathbf{RecursiveApx}_{\bot}}
\newcommand{\brec}{\mathbf{RecursiveApx}}
\newcommand{\crec}{\mathbf{MedianApx}}

\newcommand{\algCorez}[1]{\arec\brak{#1}}
\newcommand{\algCorezNP}{\arec}
\newcommand{\algCorezB}[1]{\brec\brak{#1}}
\newcommand{\algCorezNPB}{\brec}
\newcommand{\algmain}{\mathbf{Hyperedge\mhyphen Approx}}

\newcommand{\Cialg}{\mathbf{Count} \mhyphen\mathbf{Heavy}}

\newcommand{\FindHalg}{\mathbf{Find} \mhyphen\mathbf{Heavy}}
\newcommand{\FindHalgB}{\mathbf{Find} \mhyphen\mathbf{Heavy}_{\bot}^{\zeta}}

\title{Output-sensitive approximate counting via a measure-bounded hyperedge oracle,\\ or: How asymmetry helps estimate $k$-clique counts faster}
\date{}

\author{Keren Censor-Hillel \thanks{Department of Computer Science, Technion. \texttt{ckeren@cs.technion.ac.il}. The research is supported in part by the Israel Science Foundation (grant 529/23).} 
\and
Tomer Even \thanks{Department of Computer Science, Technion. \texttt{tomer.even@campus.technion.ac.il}.} 
\and 
Virginia Vassilevska Williams \thanks{Massachusetts Institute of Technology, Cambridge, MA, USA. \texttt{virgi@mit.edu}. Supported by NSF Grant CCF-2330048, BSF Grant 2020356, and a Simons Investigator Award.}}

\date{}

\begin{document}
\maketitle

\begin{abstract}
    Dell, Lapinskas and Meeks [DLM SICOMP 2022] presented a general reduction from approximate counting to decision for a class of fine-grained problems that can be viewed as hyperedge counting or detection problems in an implicit hypergraph, thus obtaining tight equivalences between approximate counting and decision for many key problems such as $k$-clique, $k$-sum and more. Their result is a reduction from approximately counting the number of hyperedges in an implicit $k$-partite hypergraph to a polylogarithmic number of calls to a hyperedge  oracle that returns whether a given subhypergraph contains an edge.

    The main result of this paper is a generalization of the DLM result for {\em output-sensitive} approximate counting, where the running time of the desired counting algorithm is inversely proportional to the number of witnesses. Our theorem is a reduction from approximately counting the (unknown) number of hyperedges in an implicit $k$-partite hypergraph to a polylogarithmic number of calls to a hyperedge oracle called only on subhypergraphs with a small ``measure''. If a subhypergraph has $u_i$ nodes in the $i$th node partition of the $k$-partite hypergraph, then its measure is $\prod_i u_i$.

    Using the new general reduction and by efficiently implementing measure-bounded colorful independence oracles, we obtain new improved output-sensitive approximate counting algorithms for $k$-clique, $k$-dominating set and $k$-sum. In graphs with $n^t$ $k$-cliques, for instance, our algorithm $(1\pm \epsilon)$-approximates the $k$-clique count in time
    $$\tilde{O}_\epsilon(n^{\omega(\frac{k-t-1}{3},\frac{k-t}{3},\frac{k-t+2}{3}) }+n^2),$$
    where $\omega(a,b,c)$ is the exponent of $n^a\times n^b$ by $n^b\times n^c$ matrix multiplication.
    For large $k$ and $t>2$, this is a substantial improvement over prior work, even if $\omega=2$.
\end{abstract}
\pagebreak
\tableofcontents
\pagebreak

\renewcommand{\EE}{\mathcal{E}}
\section{Introduction}
For many computational problems, approximately counting the number of solutions is known to be reducible to detecting whether a solution exists. This is true in a large number of scenarios. For instance, Valiant and Vazirani \cite{ValiantV86} showed that if one can solve SAT in polynomial time, then one can also approximately solve $\#$SAT (or any problem in $\#$P) in polynomial time, within an arbitrary constant precision. Similar statements are true in the FPT setting \cite{muller} and the subexponential time setting \cite{DellL21}.

For the fine-grained setting, Dell, Lapinskas and Meeks \cite{DellLM22,DellLM24} (following \cite{DellL21}) presented a general tight reduction from $(1+\eps)$-approximate counting to  decision for any problem that can be viewed as determining whether an $n$-node $k$-uniform $k$-partite hypergraph contains an edge.

Let's take $k$-clique as an example. It is known that detecting, finding and counting $k$-cliques in an $n$-node graph is equivalent to the respective problem in $k$-partite graphs with $n$ nodes in each part. Then, given a $k$-partite graph $G$, one can mentally define a $k$-partite $k$-uniform hypergraph $H$ with $n$ nodes in each part where $(u_1,\ldots,u_k)$ is a hyperedge in $H$ if and only if $(u_1,\ldots,u_k)$ is a $k$-clique in $G$. Counting the number of $k$-cliques in $G$ approximately then corresponds to approximately counting the number of hyperedges in $H$.

Many other problems parameterized by their solution size can be viewed in this way: $k$-SUM, $k$-Dominating set, $k$-OV and so on.

Dell, Lapinskas and Meeks \cite{DellLM22} present a reduction that given an oracle $O$ that can determine whether a $k$-partite hypergraph has no hyperedges (a ``colorful independence oracle''), for any constant $k$, one can obtain with high probability a $(1\pm \eps)$-approximation to the number of hyperedges in $H$ in time $(n/\eps^2)\textrm{polylog}(n)$ and with $(1/\eps^2)\textrm{polylog}(n)$ queries to $O$:

\begin{theorem}[Dell-Lapinskas-Meeks'22]\label{thm:dlm}
    Let $G$ be a $k$-partite hypergraph with vertex set $\Vin=\Vin_1\sqcup \Vin_2\sqcup \ldots\sqcup \Vin_k$, where $\abs{\Vin}=n$,
    and a set $E$ of $m$ (unknown) hyperedges.
There exists a randomized algorithm Count$(G,\eps)$ that takes as input the vertex set $\Vin$ and has access to a {\em colorful independence oracle} $O$, which given $X_1\subseteq \Vin_1,\ldots,X_k\subseteq \Vin_k$ returns whether the subhypergraph of $G$ induced by $\cup_i X_i$ contains a hyperedge. Count outputs $\hat{m}$ such that $\Pr{\hat{m}=m\apm}\geq 1-1/n^4\;.$
Count runs in $O(nT)$ time and queries $O$ at most $T$ times, where
    \[
        T=(k\log n)^{O(k)}/\eps^{2}.
    \]
\end{theorem}

This generic reduction shows that for all problems that can be cast in the $k$-partite hypergraph framework, the running time for $(1\pm \eps)$-approximate counting of solutions is within a $(1/\eps^2)\textrm{polylog}(n)$ factor of the running time for detecting a ``colorful'' solution, i.e. one that has exactly one vertex in each node partition in a $k$-partite input. For a large number of problems such as subgraph pattern problems like $k$-clique and directed $k$-cycles detecting a colorful pattern is equivalent to detection (without the colors), so that this theorem gives a tight reduction from approximate counting to detection.

For many problems, when the witness count is large, detection and approximate counting can be done faster. For instance, consider the Triangle (i.e. $3$-clique) problem in $n$ node graphs. If $G$ contains at least $t$ triangles, then by sampling $\tilde{O_\eps}(n^3/t)$ triples of vertices and checking them for triangles, we can estimate $t$ within a $(1\pm \eps)$ factor. That is, the approximate counting running time is {\em inversely} proportional to the number of copies. Algorithms with running times that depend on the number of solutions are called {\em output-sensitive}.

For triangles, T{\v{e}}tek \cite{tvetek2022approximate} obtained an even better output-sensitive approximate counting running time, $\tilde{O_\eps}(n^\omega/t^{\omega-2})$ , where $\omega<2.37134$ \cite{moreasym} is the matrix multiplication exponent. Censor-Hillel, Even and Vassilevska W. \cite{CEW24} improved the bound further using rectangular matrix multiplication and also obtained similar algorithms for approximately counting $k$-cycles in directed graphs for any constant $k\geq 3$. They also showed that their algorithms are conditionally optimal, based on assumptions from fine-grained complexity.
These works implicitly reduce output-sensitive triangle and $k$-cycle approximate counting to detection. A very natural and important question is
\begin{center}{\em Can the general theorem of Dell-Lapinskas-Meeks be extended for output-sensitive approximate counting?}\end{center}

If this can be done, then what does the running time for output-sensitive approximate counting look like for various important problems in fine-grained complexity: for $k$-clique, $k$-SUM, $k$-dominating set?

\subsection{Our results.}

Our main theorem is an extension of Dell-Lapinskas-Meeks' theorem for output-sensitive approximate counting. The theorem statement involves a function $\mu$ on subsets $U:=U_1\cup U_2\cup \ldots \cup U_k$ where for each $i\in [k]$, $U_i\subseteq \Vin_i$, and where $\Vin_i$ is the $i$th node partition of the input $k$-partite hypergraph. We define $\mu(U)=\prod_{i=1}^k |U_i|$. An \HO is an oracle that takes as input a set $U$ and returns whether the induced subhypergraph $G[U]$ contains a hyperedge or not.

\begin{restatable}[Main Theorem]{theorem}{thmMain}\label{thm4:main}
    Let $G$ be a $k$-partite hypergraph with vertex set $\Vin=\Vin_1\sqcup \Vin_2\sqcup \ldots\sqcup \Vin_k$, where $\abs{\Vin}=n$,
    and a set $E$ of $m$ hyperedges.
There exists a randomized algorithm $\algmain$ that takes as input the vertex set $\Vin$ and has access to a \HO, which
    outputs $\hat{m}$ such that $\Pr{\hat{m}=m\apm}\geq 1-1/n^4\;.$
    The algorithm queries the oracle at most $T$ times, while only querying sets $U$ that satisfy $\gs(U)\leq R$, for
    \begin{align*}
        T=(k\log n)^{O(k^3)}/\eps^{2k}\;, &  & R=\frac{\gs(\Vin)}{m}\cdot \frac{(k\log n)^{O(k^3)}}{\eps^{2k}}\;.
    \end{align*}
\end{restatable}

Given our new general theorem, we turn to implementing \HO for several fundamental problems of interest: $k$-Clique, $k$-Dominating set and $k$-SUM.

\paragraph{A new approximate counting algorithm for $k$-Clique.}
The fastest algorithm for $k$-Clique detection and exact counting is by Ne\v{s}etril and Poljak \cite{nesetril} who reduced the problem to triangle counting in a larger graph.
The running time of the algorithm is  $$O(n^{\omega(\lfloor k/3 \rfloor,\lceil k/3 \rceil,k-\lfloor k/3 \rfloor-\lceil k/3\rceil)}),$$ where $\omega(a,b,c)$ is the exponent for $n^a\times n^b$ by $n^b\times n^c$ matrix multiplicaiton. (The best known bounds for these rectangular exponents can be found in \cite{moreasym}).

When an $n$-node graph has $m=n^{t}$ $k$-cliques, one can easily estimate $m$ in $\tilde{O}_\eps(n^k/m)$ time by sampling $k$-tuples of vertices and checking if they form a $k$-clique.
Since $\Omega(n^2)$ time is necessary when $m\leq n^{k-2}$ \cite{EdenRS20,EdenR18}\footnote{The lower bound is $\Omc[e]$ where $e$ is the number of edges in the graph, or $\Omc[\frac{n}{m^{1 / k}}+\min \set{\frac{e^{k / 2}}{m}, e}]$, which is $\Omc[n^2]$ for
    $e=\Omc[n^2]$ and $m\leq n^{k-2}$.} and the sampling algorithm gets at least $\Omega(n^2)$ then, the interesting regime is when $m\leq O(n^{k-2})$.
\begin{center}{\em What is the best running time for approximate $k$-clique counting when $m\in [\Omega(1),O(n^{k-2})]$?} \end{center}
As mentioned earlier, for $k=3$, i.e. approximate triangle counting, a faster algorithm is possible, running in time $\tilde{O}_\eps(n^{\omega(1-t,1,1)})$ \cite{CEW24,tvetek2022approximate}.
We could apply Ne\v{s}etril and Poljak's reduction from $k$-clique to triangle and then use the approximate triangle counting algorithm of \cite{CEW24}. How good is this? Say that $k$ is divisible by $3$. Then  Ne\v{s}etril and Poljak's reduction creates a graph on $O(n^{k/3})$ vertices that has a number of triangles that is exactly ${k\choose k/3}{2k/3\choose k/3}$ times the number $m=n^t$ of $k$ cliques in the original graph. Applying the algorithm of \cite{CEW24} on this graph we get a $(1\pm \eps)$-approximation to $m$ in $\tilde{O}_\eps(n^{\omega(k/3-t,k/3,k/3)})$ time.

    {\em We show that we can do better:}

\begin{restatable}[Clique (Simplified)]{theorem}{cliqueSimp}\label{thms:clique}
    Let $G$ be a given graph with $n$ vertices and let $k\geq 3$ be a fixed integer.
    There is a randomized algorithm that outputs an approximation $\htt$ for the number $m$ of $k$-cliques  in $G$ such that
    $\Pr{(1-\eps)m \leq \htt\leq (1+\eps)m}\geq 1-1/n^2$, for any constant $\eps>0$. Let $r$ be such that $n^r=n^k/m$.
    The running time is bounded by $\tilde{O_\eps}({n^\psi + n^2})$, where $$\psi=\frac{1}{3}\max\set{\omega(r-1,r-1,r+2),\omega(r-2,r+1,r+1)}\leq \frac{1}{3}\omega(r-1,r,r+2).$$
\end{restatable}

If $m=n^t$ and $r=k-t$, the running time of our $k$-clique approximate counting algorithm is never worse than $$\tilde{O}_\eps(n^{\omega(\frac{k-t-1}{3},\frac{k-t}{3},\frac{k-t+2}{3}) }+n^2).$$
Thus for large $k$, we are essentially distributing the $-t$ term almost equally among the three inputs to $\omega$, getting a running time closer to $n^{\omega(k-t)/3}$. For $t > 2$, our running time is substantially better than the $\tilde{O}_\eps(n^{\omega(k/3-t,k/3,k/3)})$ time that follows from composing \cite{CEW24} with the Ne\v{s}etril and Poljak reduction. The easiest way to see this is when $\omega=2$: then $\omega(k/3-t,k/3,k/3)=2k/3$, whereas $\omega(\frac{k-t-1}{3},\frac{k-t}{3},\frac{k-t+2}{3})=2k/3-2(t-2)/3$.
In other words, our algorithm dependency in $t$ substantially improves upon previous approach even if $\omega=2$.

We obtain our result using our \Cref{thm4:main} and by designing a specialized \HO for $k$-clique. This \HO is obtained by a modification of Ne\v{s}etril and Poljak's reduction aimed at optimizing the running time for $k$-clique detection in unbalanced $k$-partite graphs. See the overview section.

For the special case of $k=4$ we get the following running time, which is faster than the best known running time of $O(n^{\omega(2,1,1)})$ for $4$-clique detection, for every $m=n^{\Omc}$.
\begin{restatable}[$4$-clique]{theorem}{tKfour}\label{thms:k4}
    Let $G$ be a given graph with $n$ vertices.
    There is a randomized algorithm that outputs an approximation $\htt$ for the number $m=n^t$ of $4$-cliques  in $G$ such that
    $\Pr{(1-\eps)m \leq \htt\leq (1+\eps)m}\geq 1-1/n^2$, for any constant $\eps>0$. The running time is bounded by $\tilde{O_\eps}({n^{\psi_4} + n^2})$, where
    \begin{align*}
        \psi_4=\max\set{\omega(1, 1, 2-t),\left(1-\frac{t}{4}\right)\cdot\omega(1, 1, 2),\omega\left(1, 1-\frac{t}{3}, 2-\frac{2t}{3}\right)}\;.
    \end{align*}
\end{restatable}

We note that 
\[\psi_4\leq \max\set{t+(1-t/2)\omega(1, 1, 2),\left(1-\frac{t}{4}\right)\cdot\omega(1, 1, 2),(t/3)+(1-t/3)\omega\left(1, 1, 2\right)}\leq \omega(1,1,2)-t/2,\]
and so the running time of our $4$-clique approximate counting algorithm is at most $\tilde{O}(n^2+n^{\omega(1,1,2)}/\sqrt{m})$.

\begin{figure}[ht]
    \centering
    \begin{subfigure}{0.48\textwidth}
        \centering
        \includegraphics[width=\textwidth]{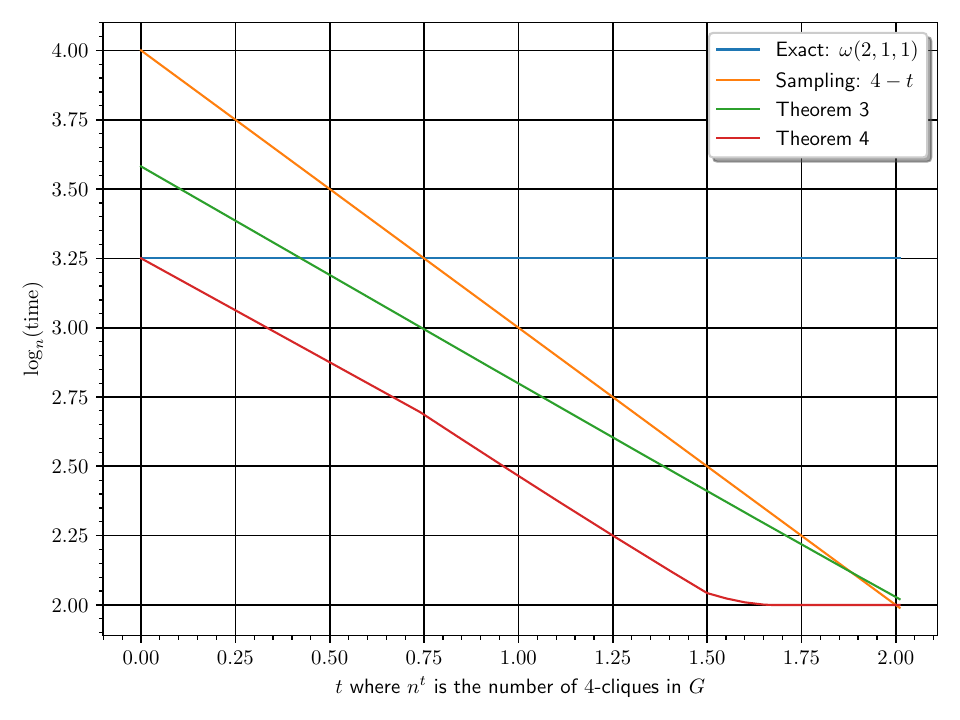}
        \caption{$k=4$}
        \label{fig:4clique}
    \end{subfigure}
    \hfill
    \begin{subfigure}{0.48\textwidth}
        \centering
        \includegraphics[width=\textwidth]{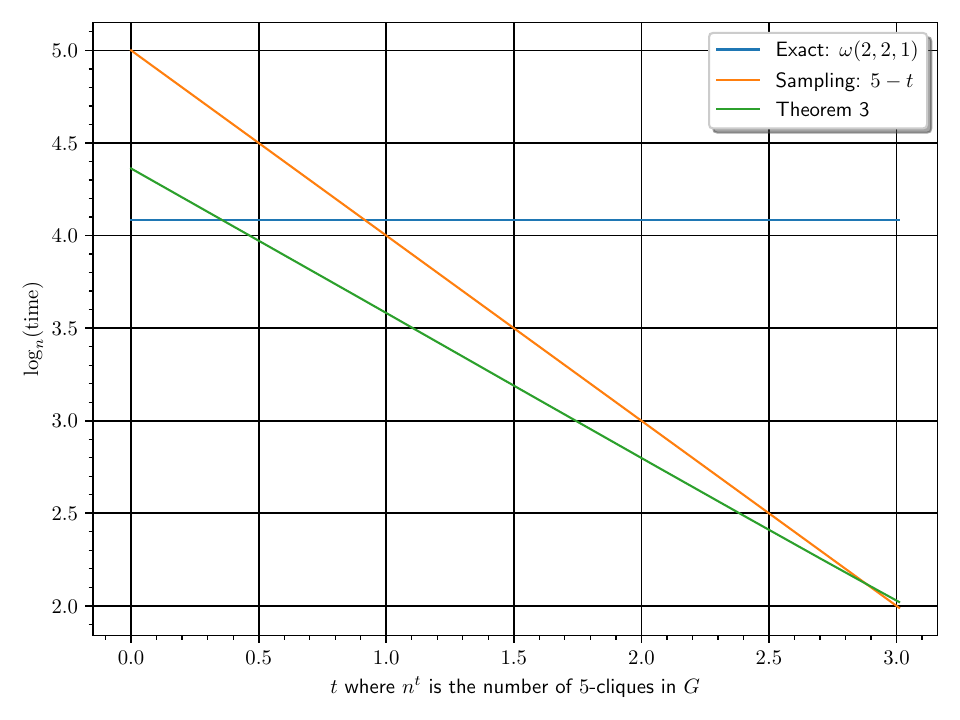}
        \caption{$k=5$}
        \label{fig:5clique}
    \end{subfigure}
    \caption{The above figures illustrate the comparison between our work and prior work for $k=4$ and $k=5$ in the $k$-clique problem. The blue line depicts the running time of \cite{nesetril}, which does not depend on $m$. The orange line depicts the running time of the sampling algorithm. 
    The green line depicts our running time specified in \Cref{thms:clique}, and the red line depicts our running time specified in \Cref{thms:k4}.
    \Cref{fig:clique0,fig:ds0} are plotted based on \cite{RenMM2021}.
    }
    \label{fig:clique0}
\end{figure}

\paragraph{New approximate counting algorithms for $k$-Dominating set.}

The fastest algorithm for $k$-dominating set detection and exact counting is by
Eisenbrand and Grandoni \cite{eisenbrand2004complexity} who reduced this problem to counting triangles in a tripartite graph, with vertex set of size $n,n^{\lfloor k/2 \rfloor}$ and $n^{\lceil k/2 \rceil}$.
The running time of the algorithm is
$$O(n^{\omega(1,\lfloor k/2 \rfloor,\lceil k/2 \rceil)})\;.$$
The problem of detecting whether a graph contains a $k$-dominating set is in the class $W[2]$, and also has a tight conditional lower bound of $n^{k-o(1)}$
based on SETH \cite{puatracscu2010possibility}.
On the other hand, the previously mentioned algorithm of \cite{eisenbrand2004complexity} obtains a running time of $n^{k+o(1)}$ for $k\geq 8$, and if $\omega=2$ then also for $k\geq 2$.
The question of deciding whether a graph contains a $k$-dominating set when the graph is sparse was also studied in \cite{FischerKR24}.

The simple sampling algorithm for approximate counting of $k$-dominating sets is to sample $\tilde{O_\eps}(n^k/m)$ sets of size $k$ and check if they form a $k$-dominating set, where testing if a set of $k$ vertices is a $k$-dominating set can be done in $O(kn)$ time, and therefore, for constant $k$ we get an upper bound of $\tilde{O_\eps}(n^{k+1}/m)$ time.
Other than that, to the best of our knowledge, there is no known output-sensitive algorithm for $k$-dominating set approximate counting or detection, and we provide the first such algorithm:
\begin{restatable}[$k$-DS (Simplified)]{theorem}{DSSimplified}\label{thms:DS}
    Let $G$ be a given graph with $n$ vertices and let $k\geq 3$ be a fixed integer.
    There is a randomized algorithm that outputs an approximation $\htt$ for the number $m=n^t$ of $k$-dominating sets in $G$ such that
    $\Pr{(1-\eps)m \leq \htt\leq (1+\eps)m}\geq 1-1/n^2$, for any constant $\eps>0$. The running time is bounded by
\begin{align*}
        \tilde{O_\eps}(n^{\omega(1,(k-t-1)/2,(k-t+1)/2)})\;.
    \end{align*}
    If $k-t\geq 1.64$ (see footnote\footnote{When $1\leq \alpha\cdot (k-t-1)/2$, where $\alpha\geq 0.321334$ is the largest known value s.t. $\omega(1,\alpha,1)=2$, we have that $\omega(1,(k-t-1)/2,(k-t+1)/2)=k-t$.}) or if $\omega=2$, the running time is $\tilde{O_\eps}(n^{k}/m)$.
\end{restatable}
Thus, we obtain an improvement over the simple $\tilde{O_\eps}(n^{k+1}/m)$ sampling algorithm.
\begin{restatable}[$3$-dominating set]{theorem}{thmDSthree}\label{thms:ds3}
    Let $G$ be a given graph with $n$ vertices.
There is a randomized algorithm that outputs an approximation $\htt$ for the number $m=n^t$ of $k$-dominating sets in $G$ such that
    $\Pr{(1-\eps)m \leq \htt\leq (1+\eps)m}\geq 1-1/n^2$, for any constant $\eps>0$. The running time is bounded by
\begin{align*}
        \tilde{O_\eps}(n^2 + n^{\psi})\;, &  & \text{where }\quad 
        \psi=\max\set{\omega(1,1-\frac{t}{3},\frac{2(t-1)}{3}),\omega(1,1,2-t)}\leq \omega(1,1,2(1-\frac{t}{3}))\;.
    \end{align*}
\end{restatable}

\begin{figure}[ht]
    \centering
    \begin{subfigure}{0.48\textwidth}
        \centering
        \includegraphics[width=\textwidth]{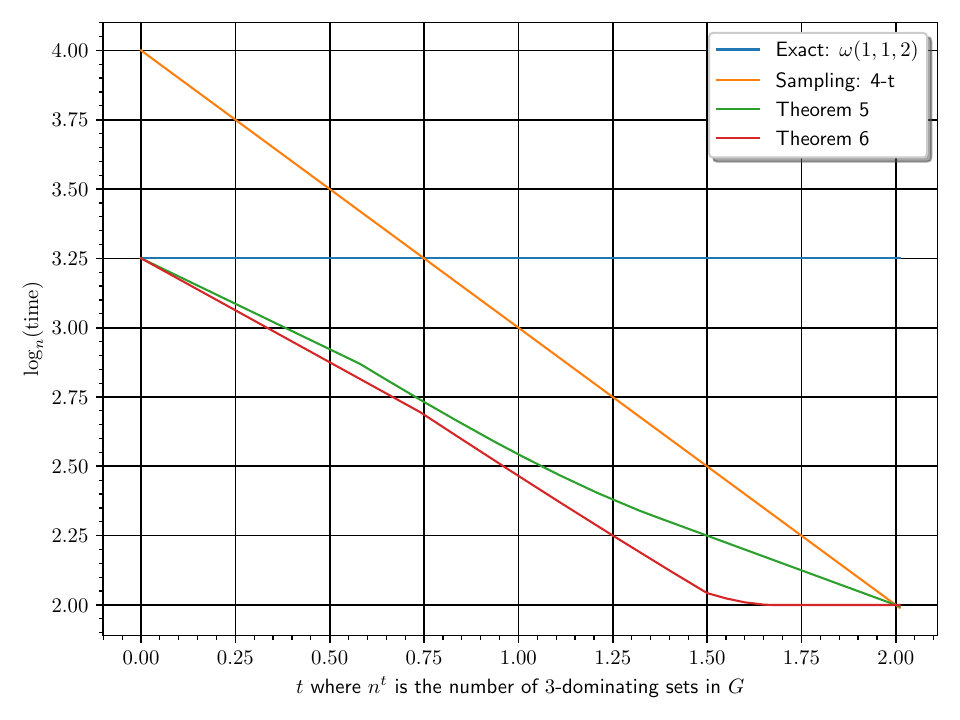}
        \caption{$k=3$}
        \label{fig2:DS4}
    \end{subfigure}
    \hfill
    \begin{subfigure}{0.48\textwidth}
        \centering
        \includegraphics[width=\textwidth]{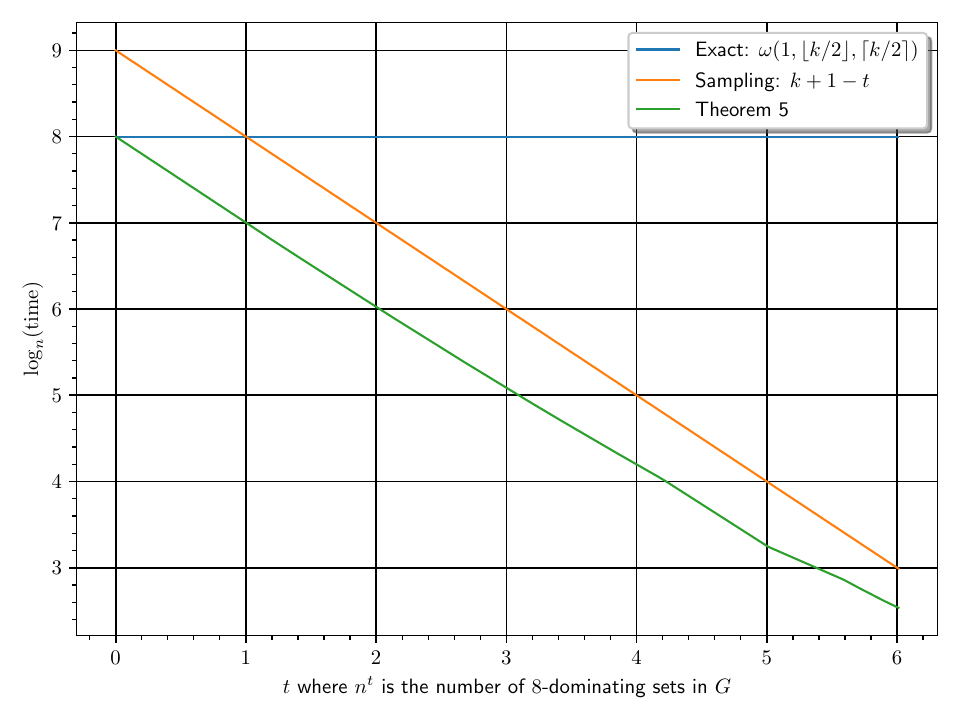}
        \caption{$k=8$}
        \label{fig2:DS10}
    \end{subfigure}
    \caption{The above figures illustrate the comparison between our work and prior work for $k=3$ and $k=8$ in the $k$-DS problem. The blue line depicts the running time of \cite{eisenbrand2004complexity}, which does not depend on $m$. The orange line depicts the running time of the sampling algorithm. 
    The green line depicts our running time specified in \Cref{thms:DS}, and the red line depicts our running time specified in \Cref{thms:ds3}. }
    \label{fig:ds0}
\end{figure}

\paragraph{New approximate counting algorithms for $k$-sum.}
In the $k$-sum problem, we are given a set of $n$ numbers and we are asked to decide if there exists a $k$-tuple of numbers that sum to zero.
This problem can be deterministically reduced to a ``colorful'' version, where we are given $k$  sets $A_1,\ldots,A_k$ of $n$ numbers each, and we are asked to decide if there exists a $k$-tuple of numbers, one from each set, that sum to zero.
Moreover, the reduction preserves the number of solutions, and if the original problem had $m$ solutions, then the colorful version has $m\cdot k!$ solutions.
The problem of $k$-sum detection is one of the most studied problems in fine-grained complexity, and it is conjectured that for every constant $\eps>0$
it does not have a $O(n^{\ceil{k/2}-\eps})$ time algorithm (see e.g. \cite{abboud2013exact,vassilevska2009finding,vsurvey}).

The $k$-sum problem also has a simple sampling algorithm that samples $\tilde{O_\eps}(n^k/m)$ $k$-tuples of numbers and checks if they sum to zero, but to the best of our knowledge, there is no known output-sensitive algorithm for $k$-sum approximate counting.

We provide the first output-sensitive algorithm for $k$-sum approximate counting, for $k\geq 3$.
\begin{restatable}[$k$-Sum (Simplified)]{theorem}{sumSimp}\label{thms:sum}
    Let $A$ be a set of $n$ integers, and let $k\geq 3$ be a fixed integer.
    There is a randomized algorithm that outputs an approximation $\htt$ for the number $m=n^t$ of $k$-sum tuples in $A$ such that
    $\Pr{(1-\eps)m \leq \htt\leq (1+\eps)m}\geq 1-1/n^2$, for any constant $\eps>0$. The running time is bounded by $\tilde{O_\eps}(n+\min\set{n^{(k-t+1)/2},n^{\ceil{k/2}}})$.
\end{restatable}
The next theorem provides a better bound for smaller value of $k$. 
\begin{restatable}[$3$-Sum]{theorem}{sumThree}\label{thms:sum3}
    Let $A$ be a set of $n$ integers.
There is a randomized algorithm that outputs an approximation $\htt$ for the number $m=n^t$ of $3$-sum tuples in $A$ such that
    $\Pr{(1-\eps)m \leq \htt\leq (1+\eps)m}\geq 1-1/n^2$, for any constant $\eps>0$. The running time is bounded by $\tilde{O_\eps}(n+n^{2(1-t/3)})$.
\end{restatable}

\begin{figure}[ht]
    \centering
    \begin{subfigure}{0.48\textwidth}
        \centering
        \includegraphics[width=\textwidth]{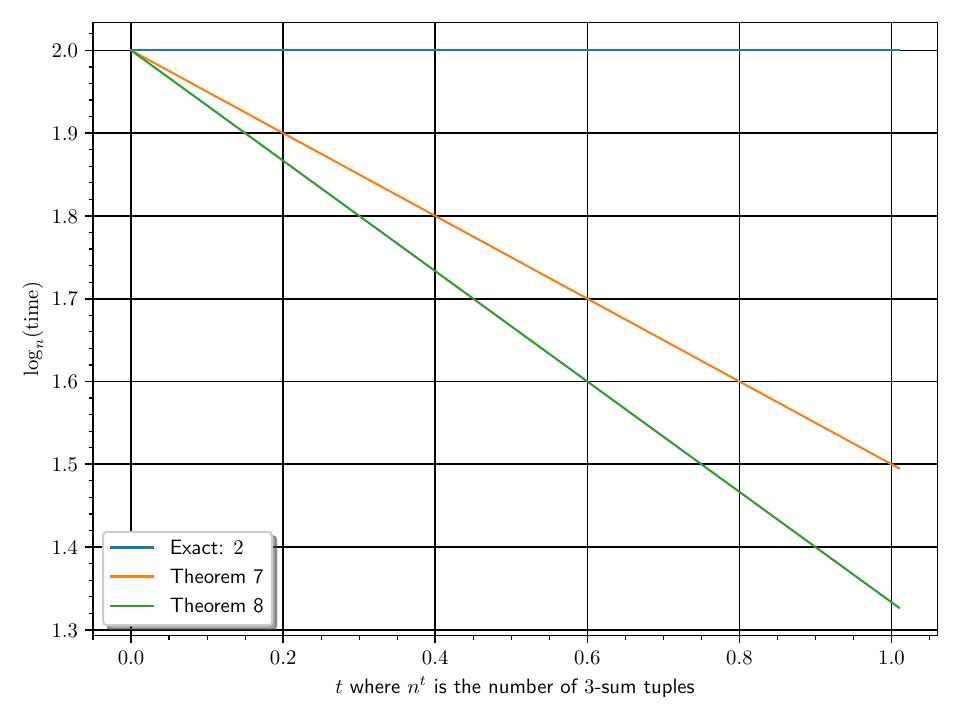}
        \caption{$k=3$}
        \label{fig2:sum3}
\end{subfigure}
    \hfill
    \begin{subfigure}{0.48\textwidth}
        \centering
        \includegraphics[width=\textwidth]{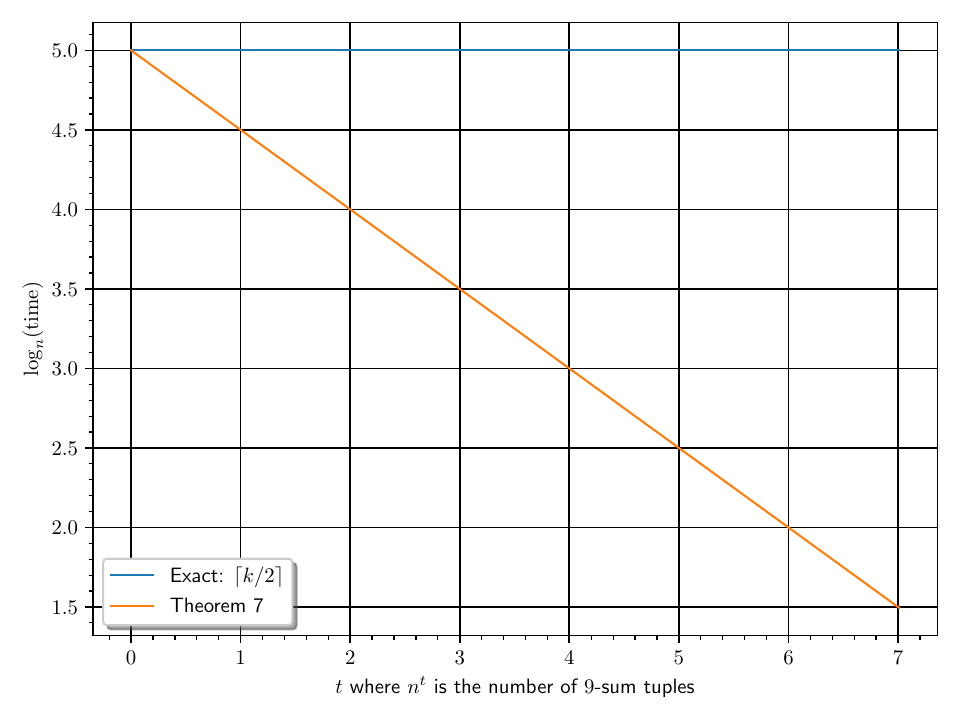}
        \caption{$k=9$}
        \label{fig2:sum9}
    \end{subfigure}
    \caption{The above figures illustrate the comparison between our work and prior work for $k=3$ and $k=9$ in the $k$-sum problem. The blue line depicts the folklore running time of $n^{\ceil{k/2}}$, which does not depend on $m$. 
    The orange line depicts the running time of the sampling algorithm.
    The green line depicts our running time specified in \Cref{thms:sum}, and the red line depicts our running time specified in \Cref{thms:sum3}.}
\end{figure}

\paragraph*{Additional Related Work.}

The problem of detecting or approximately counting subgraphs has also been explored within the framework of sublinear algorithms. In the standard sublinear model \cite{parnas2002testing,kaufman2004tight}, an algorithm may make three types of queries: \begin{enumerate*}[label=(\roman*)] \item degree queries, which return a vertex's degree, \item adjacency queries, which indicate whether two vertices are adjacent, and \item $i$th neighbor queries, which return the $i$th neighbor of a given vertex. \end{enumerate*} A stronger, augmented model \cite{aliakbarpour2018sublinear} allows a fourth type of query that yields a uniformly random edge.

In the general case, detecting a subgraph $H$ requires $\Omega(m)$ queries because the entire input must be read. Consequently, sublinear algorithms focus on settings with many copies of $H$ (denoted $t$), aiming to devise algorithms whose complexity decreases as $t$ increases.

For instance, for approximate counting the number of triangles \cite{eden2017approximately} establishes query time of 
$\Theta\bigl(\frac{n}{t^{1/3}} + \min\set{m,\frac{m^{3/2}}{t}}\bigr)$, with running time of $O(\frac{n}{t^{1/3}} + \frac{m^{3/2}}{t})$, 
and \cite{EdenRS20} 
provides an $O\bigl(\frac{n}{t^{1/k}} + \frac{m^{k/2}}{t}\bigr)$ time and query complexity algorithm for approximate counting of a $k$-clique. 
In the augmented model, there are results for general subgraphs $H$ \cite{assadi2018simple,Fichtenberger2020SamplingAS,Biswas2021TowardsAD}.
For approximate counting a general subgraph $H$ 
\cite{assadi2018simple} presents a sublinear-time algorithm in $\tilde{O}\bigl(\frac{m^{\rho(H)}}{t}\bigr)$ time, and $\tilde{O}\bigl(\min\set{m,\frac{m^{\rho(H)}}{t}}\bigr)$ query complexity, where $\rho(H)$ denotes the fractional edge-cover number of $H$, with a matching lower bound on the number of queries for odd length cycles.
For sampling a uniform copy of $H$ in the augmented model, \cite{Fichtenberger2020SamplingAS} presents a sublinear-time algorithm with the same time and query complexity as \cite{assadi2018simple}, and a matching lower bound for query complexity of sampling almost uniform $k$-cliques. 

~\\ Our work focuses on the regime where sublinear query complexity is unattainable, for example when $\frac{m^\rho(H)}{t}\geq m$. We provide algorithms that depend on $n$ the number of vertices, while being output-sensitive. The techniques we use are very different. 

\subsection{Technical Overview }
\paragraph{The Framework.}
Our algorithm is based on the following very general framework (consider the $k$-clique counting problem).
We sample (induced) subgraphs, check whether they contain at least one $k$-clique using a detection oracle, and use this information to approximate the number of $k$-cliques in the graph.
To analyze the complexity of the algorithm, we need to bound the number of queries to the oracle and the time it takes to answer each query.
The question of approximating the number of $k$-cliques in a graph, or more generally the number of hyperedges in a $k$-partite hypergraph, while minimizing the number of queries to the oracle, was previously studied in \cite{bhattacharya2024faster,DellLM22,DellLM24}, which roughly showed that the number of queries is $O_k(\polylog{n}/\eps^2)$.
However, the size of the sets that are queried is not taken into consideration in the above work.
Therefore, because there is no known $k$-detection oracle that is faster than an exact counting algorithm, the obtained algorithm for $k$-clique approximate counting is not faster than the exact counting version.
In a recent work \cite{DellLM24}, the authors provided a lower bound on the total running time of the algorithm for approximately counting the number of hyperedges in a $k$-partite hypergraph using a detection oracle, where every query on a set $U$ has cost that is proportional to $\abs{U}$.

We develop a new framework for approximately counting the number of hyperedges in a $k$-partite hypergraph, using a detection oracle which is sensitive to the number of hyperedges in the graph.
This allows us to provide a new algorithm for approximating the number of $k$-cliques in a graph $G$ with $n$ vertices, assuming the graph has $m$ $k$-cliques, which is faster than both the matrix multiplication approach and the na\"ive sampling algorithm, and also provides a negative dependency on the number of $k$-cliques.

To do so, we restrict the size of the sets that are queried with respect to some measure $\gs$, which we define as follows.

\paragraph*{The Measure $\gs$.}
For a $k$-partite hypergraph $G$ with vertex set $V=V_1\sqcup V_2\sqcup\ldots\sqcup V_k$, and a set of hyperedges $E$ containing $m$ hyperedges, we define $\gs(U)$ by $\gs(U)\triangleq \prod_{i=1}^{k}\abs{U\cap V_i}$.
To get some intuition on this measure, we look at the na\"ive sampling algorithm that samples $k$ vertices at random and queries the oracle on the sampled vertices to learn whether they form a hyperedge or not.
By querying the oracle $\TO{\frac{n^k }{m\eps^2}}$ times, we can approximate the number of hyperedges in $G$.
The number of queries is very large, and the measure $\gs$ of the queried set is constant.
On the other hand, in \cite{DellLM22,DellLM24,bhattacharya2024faster} they achieve an algorithm with a polylogarithmic number of queries, while querying sets with possibly maximum measure of $\gs(V)$.

In comparison, our approach gives a  polylogarithmic number of queries, while only querying sets $U$ that satisfy $\gs(U)\leq \gs(V)/m$,
which is comparable to the number of queries of the na\"ive sampling algorithm.
In other words, we switch from $\TO{n^k/m}$ many queries of constant measure to  a polylogarithmic number of queries and of measure at most $\TO{n^k/m}$.
Note that for the case of $k$-cliques, one could implement a detection oracle that answers queries on sets of $k$ vertices in constant time:
Given a set $U\subseteq V$ with $\gs(U)\leq M$, one could trivially implement the detection oracle in time $O(k^2\cdot M)$, by iterating over all possible $k$-tuple of vertex sets hyperedges in $G$ and counting the number of $k$-tuples that form a hyperedges.
However, to benefit from the new framework, we must provide a faster implementation of the detection oracle, which means it has to be \emph{sublinear} in the measure of the set.
To this end, we show that if the detection oracle takes time $T$ to answer the query on the entire vertex set $V$, then for every set $U\subseteq V$ with $\gs(U)\leq \gs(V)/N$, the detection oracle can be implemented on $U$ in time $O(T/N^{\Omega(1)})$.

~\\To conclude, we provide a new algorithm for approximating the number of $k$-cliques in a $k$-partite graph $G$ with $n$ vertices, assuming the graph has $m$ $k$-cliques, where $k$ is a constant, by reducing the problem to implementing a fast detection oracle on sets $U$ with $\gs(U)=\TO{\gs(V)/m}$, where the oracle is queried at most $\ple$ times.
We emphasize that we do not require that the queried sets have smaller cardinality, and in fact we query sets of size $\Omega(n)$.

In addition, we show that if our detection oracle is optimal, then so is the approximate counting algorithm we obtain.
In other words, assume there exists an algorithm $\AB$ that computes a $\apm$ approximation for the number of $k$-cliques in a graph $G$ that takes time $T$, assuming the graph has at least $m$ $k$-cliques.
Then, we can implement a detection oracle for $k$-cliques on sets $U$ with $\gs(U)=\TO{\gs(V)/m}$ in time $\TO{T}$.
This detection oracle can be used to obtain an approximate counting algorithm for $k$-cliques that takes time $\TO{T}$.
In other words, the framework we present is sufficiently strong to get an optimal approximate counting algorithm for the $k$-clique problem, and reducing the problem to implementing a detection oracle for $k$-cliques on sets $U$ with $\gs(U)=\TO{\gs(V)/m}$, is in some sense inherent.

\paragraph*{The Sampling Lemma.} Recall that our algorithm is based on the following very general framework, which samples (induced) subgraphs, checks whether they contain at least one $k$-clique, and use this information to approximate the number of $k$-cliques in the graph.
In what follows we explain according to which distributions we sample the subgraphs, and how these distributions are related to the measure $\mu$.
Let $G=(V,E)$ be a $k$-partite hypergraph with $m$ edges.
We define the following asymmetric sampling, which is tight in a sense we explain shortly.
For every vector $P\in[0,1]^k$, with $P=(p_1,p_2,\ldots,p_k)$, we define the weight of the vector $P$ by $w(P)\triangleq \prod_{i=1}^{k}p_i$.
We sample a subset of vertices $V[P]$ using the vector $P$ by adding to $V[P]$ every vertex $v\in V_i$ independently with probability $p_i$.
The sampling lemma states that if $G$ has at least $\Omc[m\cdot \polylog{n}]$ hyperedges, then there exists a vector $P$ with $w(P)\leq 1/m$, such that the probability that $G[V[P]]$ contains a hyperedge is at least $\Omc$.
This is the best that we can hope for, as the expected number of hyperedges in $G[V[P]]$ is $m\cdot w(P)=\TO{1}$.
In fact it provides us with something stronger. For every vertex $v$ with degree at least $\Lambda\cdot \polylog{n}$, there exists a vector $P$ with $w(P)\leq 1/\Lambda$, such that the probability that $v$ is non-isolated in $G[V[P]]$ is at least $(\eo)^k$. Moreover, for every vertex $v$ with degree less than  $\Lambda$, and any vector $P$ with $w(P)\leq 1/\Lambda$, we have that the probability that $v$ is non-isolated in $G[V[P]]$ is at most $1/4^k$.
This allows us to distinguish between vertices with high degree and vertices with low degree, which is a crucial step in our algorithm, on which we elaborate later.

While the space of possible vectors $P$ is infinite, we show that it suffices to consider only a finite subset of $\log^k (4n)$ vectors. In particular, it suffices to consider only the set of vectors $P=(p_1,p_2,\ldots,p_k)$ such that $p_i=2^{-j}$ for every $j\in\zrn{\log n + 1}$. We call such a vector a \emph{simple} vector. An important observation is the following connection between $\gs$ and the Sampling Lemma: For any sampling vector $P$, the set $V[P]$ has measure $\gs(V[P])=\gs(V)\cdot w(P)$.

A similar version of this sampling lemma was used in \cite{DellLM22} together with a detection oracle for hyperedge making only a polylogarithmic number of queries. The number of queries was later improved in \cite{DellLM24}.    However, both of these approaches don't lead to queries on subgraphs of reduced measure, and therefore are not applicable for our use.
Following the approach of \cite{CEW24} do lead to queries on subgraphs of reduced measure, but there a stronger detection oracle is used, which returns the degree of each vertex.

\paragraph*{Implementation of the Detection Oracle.}
An additional crucial step in our algorithm is how to implement the detection oracle on sets $U$ obtained by sampling subset of vertices from $V$ using a vector $P$ with $w(P)\leq 1/m$, and therefore with $\gs(U)=\TO{\gs(V)/m}$.
Moreover, how can we do so in time that is sublinear in the size of $\gs(U)$?
Intuitively, how can we benefit from the fact that the set $U$ is partitioned into $k$ parts of possibly very different sizes?

For example for $k$-cliques, to do so, we modify a well-known reduction from $k$-clique detection to triangle detection \cite{nesetril}.
The original reduction from $k$-clique detection starting from an $n$-node forms a new auxiliary graph whose vertices are those $k/3$-tuples that are also $k/3$ cliques (if $k/3$ is not an integer, there are three types of nodes: $\lceil k/3\rceil$, $\lfloor k/3\rfloor$ and $k-\lceil k/3\rceil - \lfloor k/3\rfloor$-cliques).
Then two $k/3$-cliques (i.e. nodes in the auxiliary graph) are connected by an edge if and only if together they form a $2k/3$ clique.

Our reduction is similar but with a twist. We start with a $k$-partite graph on parts $U_1,\ldots,U_k$ and instead of creating nodes corresponding to cliques of size $k/3$, we consider different sized cliques which have nodes from a pre-specified subset of the $U_i$s.
In other words, first a sampling vector is chosen, and then the sets $U_i$ are partitioned into $3$ parts, were the partition attempts to keep all parts of the same size, or at least minimize the difference between the sizes of the parts.

With a similar framework in mind, we provide implementations of a detection oracle also for $k$-dominating set, and $k$-sum, and show that the running time of the detection oracle on sets $U$ with $\gs(U)=\TO{\gs(V)/m}$ is sublinear in $\gs(U)$.

This simple idea can also be applied to the $k$-dominating set problem, and $k$-sum problem, because the optimal detection oracle for these problems follows the same structure as the detection oracle for $k$-cliques.
In other words, both problems have an algorithm that partitions the set $[k]$ into $2$ sets $I_1$, and $I_2$, and constructs the sets $W_1$ and $W_2$ where $W_1=\bigtimes_{j\in I_1}U_j$, and $W_2=\bigtimes_{j\in I_2}U_j$.
In the language of $k$-cliques, we partitioned the set $[k]$ into $3$ sets, and constructed the sets $W_1$, $W_2$, and $W_3$, and removed from $W_j$ any element that is not a $\abs{I_j}$-clique.
We conclude that any problem with an algorithm with such a structure can benefit from our framework.
Another example of a problem that can benefit from our framework is the $k$-cycle problem, as seen in \cite{CEW24}.

We elaborate on how to get a fast detection oracle on sets $U$ with $\gs(U)=\TO{\gs(V)/m}$, where the running time is sublinear in the size of $\gs(U)$, in \Cref{sec:reduction}.

\paragraph{Approximate Counting Hyperedges Using Oracle Detection.}
The second part of this paper is dedicated to proving \Cref{thm4:main} -- How to approximate the number of hyperedges in a $k$-partite hypergraph $G=(V,E)$ with $n$ vertices and $m$ hyperedges using a detection oracle, where only sets $U$ with $\gs(U)\leq \TO{n^k/m}$ are queried, and the oracle is queried at most $\ple$ times.
We highlight the techniques and tools we use in what follows.

\paragraph{The Recursive Algorithm.} We now explain how these detection oracles are used.
The recursive algorithm is an algorithm that takes as input a set of vertices $V$, of a $k$-partite hypergraph $G$, and computes a $\apm$ approximation for the number of hyperedges in $G$, using a detection oracle.
Given an algorithm that produces an approximation $\htt$ that falls within the interval $m\cdot \apm$ with probability of at least $2/3$, we can get an algorithm that produces an approximation $\htt$ such that $\Pr{\htt=m\apm}\geq 1-1/n^4$, by taking the median of $10\log n$ independent executions of the algorithm.
Therefore, we focus on getting an algorithm that succeeds with probability of at least $2/3$.
We also assume that the algorithm is given a coarse approximation
for $m$ the number of hyperedges in $G$, and the algorithm is required to output an approximation $\htt$ such that $\htt=m\cdot\apm$ with probability of at least $2/3$.
This assumption can be removed, and we explain how to do so later on.

The main goal of the recursive algorithm is to compute a $\apm$ approximation for the number of hyperedges in a $k$-partite hypergraph $G=(V,E)$ with $n$ vertices and $m$ hyperedges, by querying the oracle only $\TO{1}$ times, while only querying sets $U$ that satisfy $\gs(U)\leq \TO{n^k/m}$.
Recall that we say that a hypergraph is $k$-partite if the vertex set $V$ is partitioned into $k$ disjoint sets $V_1,V_2,\ldots,V_k$, and every hyperedge $e\in E$ intersects every vertex set $V_i$ in exactly one vertex.

A tempting (but unsuccessful) approach for obtaining a $\apm$ approximation for $m$ is to sample a subset of vertices $U$ by sampling each vertex $v\in V$ independently with probability $p$, and to compute an approximation $\htt$ for $m$ based on $m_U$ the number of hyperedges in the subhypergraph $G[U]$.
This approach fails because the probability that a hyperedge $e_1$ is in $G[U]$ is not independent of the probability that another hyperedge $e_2$ is in $G[U]$, when $e_1$ and $e_2$ share a vertex.
Moreover, the variance of the random variable $m_U$ is of order of $O(mp\cdot \Delta)$, where $\Delta$ is the maximum degree of a vertex in $G$.
In other words, $\Var{m_U}=O(\Exp{m_U}\cdot \Delta)$.
Using this observation, we can explain the main goal of our recursive algorithm.
We say that a vertex $v$ is $\Lambda$-heavy if its degree is at least $\Lambda$, and we say that it is $\Lambda$-light otherwise.
We say that a hyperedge is $\Lambda$-heavy if it contains at least one $\Lambda$-heavy vertex, and it is $\Lambda$-light otherwise.
The recursive algorithm approximates the heavy hyperedges and the light hyperedges separately, and outputs the sum of the approximations.
To approximate the number of $\Lambda$-heavy hyperedges, it finds a superset of the $\Lambda$-heavy vertices $\vl$ and approximates the number of $\Lambda$-heavy hyperedges, which are the hyperedges that contain at least one vertex from $\vl$.
The procedure for finding the heavy vertices is called $\FindHalgB$, and the procedure for approximating the number of heavy hyperedges is called $\Cialg$. Both procedures heavily relay on the detection oracle, as they do not have any  other way to access the hyperedges of the graph.

After finding the heavy vertices and approximating the number of heavy hyperedges, the algorithm removes the set $\vl$ from the graph, and approximates the number of hyperedges in the subhypergraph $G[V\setminus \vl]$, which contains only $\Lambda$-light vertices.
To approximate the number of light hyperedges, it samples a subset of vertices $U$ by sampling each vertex $v\in V\setminus \vl$ independently with probability $p=1/2$, and approximates the number of hyperedges in $G[U]$.
To approximate the number of hyperedges in the new sampled induced hypergraph $G[U]$, the algorithm recursively approximates the number of $\Lambda'$-heavy and light hyperedges in $G[U]$, with $\Lambda'=\Lambda\cdot p^k$.

These two steps are repeated recursively until the graph becomes sufficiently small and has only a polylogarithmic number of hyperedges. 
At this point, the maximum degree of the graph is at most polylogarithmic in $n$, and the dependencies between the vertices are limited. 
Let $W$ denote the set of vertices of the graph $G$ at this point.
We switch to the algorithm presented in \cite{DellLM22} that we refer to as $\algPrev(W,\eps)$ (detailed in \Cref{thm:dlm}), which approximates the number of hyperedges in the graph $G[W]$ by querying the oracle at most $O(\log^{3k+5} n/\eps^2)$ times, while querying sets $U$ with $\gs(U)=\gs(W)$.
We show that at this point, the number of remaining vertices in $W$ is sufficiently small, so that $\gs(W)\leq \gs(V)/m$.
The algorithm switches to this algorithm as soon as the heaviness threshold $\Lambda$ is smaller or equal to $1$.
We also note that using a faster algorithm for this case would not result in an improved running time, and therefore is not considered.

~\\The key observation of this recursive framework is that the two steps have the opposite effect on the variance of the random variable $m_U$.
Removing the heavy vertices in each step before sampling a new graph ensures that the variance of the random variable $m_U$ is not too large, and that we get a good estimation with probability $1-1/\polylog{n}$ at each step.
Sampling a new graph increases the variance of the random variable $m_U$, but allows us to work with a smaller graph.
The reason that we need to work with a smaller graph is that we gradually decrease the heaviness threshold, and finding the $\Lambda$-heavy vertices becomes more expensive as $\Lambda$ decreases.
To make sure that the cost of every level in the recursion is at most the cost of the previous level, we need to sample a subhypergraph with fewer vertices.

~\\ On a more intuitive level, the recursive algorithm tries to split the vertices into buckets by their degree, and approximates the sum of the degrees of each bucket separately.
Splitting into buckets here is helpful because given a set of vertices $S$ that have degree within the interval $[d,2d]$, one can approximate the sum of the degrees of vertices in $S$ by approximating the degree of only $O(\polylog{n})$ many vertices from $S$. This is done by approximating the average degree of the vertices in $S$, and multiplying it by $\abs{S}$.
While the sum of the degrees and the number of hyperedges in the graph, which include at least one vertex from $S$, are not identical, they are related. By applying a few additional techniques, we can also approximate the number of these hyperedges using a similar method.
So if we could partition the vertices into buckets, then we could approximate the degree of only $O(\polylog{n})$ vertices from each bucket, and get an approximation for the number of hyperedges in the graph.
However, constructing the buckets is not trivial, and could be too expensive.
Therefore, we avoid constructing these buckets explicitly, and we instead build them ``on the fly'' as follows.
We first construct the heaviest bucket which contains the set of all $\Lambda_0=m\cdot \frac{\eps}{\log n}$-heavy vertices. 
\newcommand{\vli}[1]{V_{\Lambda_{#1}}}
Let $\vli{0}$ be the vertices that are in the first bucket.
We then compute an approximation $\htt_0$ for the number of hyperedges that intersect this bucket, which means that they contain at least one vertex from this set, and remove the set $\vli{0}$ from the $V$.
Computing the approximation is done by computing an approximation for the average degree of the vertices in $\vli{0}$, and multiplying it by $\abs{\vli{0}}$.
We then sample a new vertex set $U_1$ that contains every vertex $v\in V-\vli{0}$ independently with probability $p=1/2$.
In the following iteration, we construct the next bucket that contains the set of all $\Lambda_1=\Lambda_0\cdot 2^{-k}$-heavy vertices, and use it to
compute an approximation $\htt_i$ for the number of hyperedges that intersect this bucket. As before, we do so by approximating the average degree of the vertices in this set, and multiplying it by $\abs{\vli{i}}$.
We continue in this manner for $\log n$ steps, and output the approximation $\sum_{i=0}^{\log n}\htt_i\cdot 2^{i}$.
The key observation here is that instead of constructing the buckets explicitly and then sampling $O(\polylog{n})$ vertices from each bucket,
we can construct the first bucket, and then sample half of the vertices and then construct the next bucket, and so on.

The formal property that we prove is that if a bucket is ``important'', i.e., initially at least $\Omega(m\cdot \eps/\polylog{n})$ hyperedges intersect it,
then the approximation $\htt_i$ is a good approximation for the number of hyperedges that intersect this bucket, with probability of at least $1-1/\polylog{n}$, which means we can use a union bound over all (important) buckets to show that the sum of the approximations is a good approximation for the number of hyperedges in the graph.

~\\To the best of our knowledge, this framework was used first in \cite{DellL21}
for the significantly simpler case of bipartite graphs (i.e. $k=2$), and in  
\cite{tvetek2022approximate} for approximately counting the number of triangles in a graph ($k=3$).   
In \cite{CEW24} the running time for approximating the number of triangles was improved, where the paper also extends this framework for approximate counting of $k$-cycles in a graph.
\paragraph{Challenges (Compared to the Previous Work).}
There are several key differences between our work and the previous work of \cite{CEW24}.
First, as previously discussed, we use the sampling lemma in a different way.
Moreover, in \cite{CEW24} the paper uses an access to a special stronger type of oracle, that takes a subset of vertices $U$ as input and outputs the degree of \emph{every vertex} in $U$ over the subhypergraph $G[U]$.
Moreover, this oracle can exactly compute the degree of every set of $n^\alpha$ vertices, over the entire hypergraph $G$, or any subhypergraph $G[U]$, in time $O(n^{2 + o(1)})$. Recall that $\alpha\geq 0.321334$ \cite{moreasym} is the largest known value s.t. $\omega(1,\alpha,1)=2$.
In contrast, we only assume we have access to a detection oracle that takes a subset of vertices $U$ as input and informs us whether $G[U]$ contains a hyperedge or not.

\sloppy{This raises two new challenges in the implementations of the two procedures $\FindHalg$ and $\Cialg$ which lie at the heart of the recursive algorithm.}
The $\FindHalg$ procedure takes as input a subset of vertices $U$ and a heaviness threshold $\Lambda$, and with probability of at least $1-1/n^4$ outputs a superset $\vl$ of all vertices of degree at least $\llow=\tfrac{\Lambda}{4^k\cdot \cclog}$ in $G[U]$, that contains all vertices of degree at least $\Lambda$ in $G[U]$.
In \cite{CEW24}, this procedure involved sampling a subset of vertices $U_1,U_2,\ldots,U_r$, for $r=\polylog{n}$, according to some distribution, and adding to $\vl$ every vertex $v$ that was non isolated (had a positive degree) at least $\tau$ times in $G[U_i]$, for every $i\in[r]$.
The challenge here is that this method relied on a stronger oracle that could retrieve the entire set of vertices of non-zero degree in each sample. However, in our work, we can find the non-zero degree vertices only when the number of such vertices is very small.
That is, we develop a deterministic tool to find the set $S$ of all non-zero degree vertices over $G[U]$ using $O(\abs{S}\cdot k\log n)$ queries to the oracle. Therefore, we must ensure that the distribution we use to sample the vertices $U_i$ is such that the number of non-zero degree vertices in $G[U_i]$ is very small in sufficiently many samples.

The $\Cialg$ procedure takes two subsets of vertices $U$ and $\vl$, together with a heaviness threshold $\Lambda$ and a precision parameter $\eps$, and with probability of at least $1-1/n^4$, outputs an approximation for the number of hyperedges in $G[U]$ that contain at least one vertex from $\vl$. 
We denote this set of edges by $E_U(\vl)$ and its cardinality by $m_U(\vl)$.
Computing an $\apm$ approximation for $m_U(\vl)$ using the stronger oracle, can be used by sampling a subset $S$ of $r=\polylog{n}$ from $\vl$ uniformly at random, and computes the exact degree of each vertex in $S$ over $G[U]$ using the stronger oracle, which gives us an estimation $\hat{d}_S$ for the average degree of vertices in $S$ over $G[U]$. Using $\hat{d}_S$ one can get a coarse approximation for $m_U(\vl)$ by computing $\hat{d}_S\cdot \abs{\vl}$.
This coarse approximation can be refined using additional tools.

In our case, with access only to the weaker oracle, even computing the approximate degree of a single vertex is not trivial, and requires new tools.
To approximate the degree of a vertex $v$ in a $k$-partite hypergraph $G$, we use the following recursive approach. Define a $(k-1)$ partite hypergraph $G'$, where the vertex set is $V'=V\setminus V_\ell$, where $V_\ell$ is such that $v\in V_\ell$, and the hyperedges of $G'$ are defined as follows.
For every hyperedge $e$ in $G$, if $e$ contains $v$, then $e\setminus\set{v}$ is added to $G'$.
We then apply our algorithm to approximate the number of hyperedges in $G'$
using a detection oracle. We can simulate a query $U$ on $G'$ by querying the oracle on $U\cup\set{v}$ on $G$.
The problem with this approach is that while $G$ had many hyperedges, and therefore we could get a $\apm$ approximation fast, the graph $G'$ might have very few hyperedges, and then computing the number of hyperedges in $G'$ would be too slow.
For that reason we slightly change our approach, and develop a new approximation algorithm called $\apx$ which takes as input a set of vertices $U$ and a heaviness threshold $\W$.
The algorithm either outputs a value $\htt$ which serves as an approximation for the number of hyperedges $m$ in $G[U]$ or outputs a message $\bbot$ indicating that the number of hyperedges in $G[U]$ is too small to complete the computation of a $\apm$ approximation for $m_U$.
In other words, this algorithm has the guarantee that its running time is bounded by the running time on graph with at least $\Omc[\W]$ hyperedges.
Moreover, if $\W\leq m_V/2$, then the algorithm outputs a value $\htt$ which is a $\apm$ approximation for $m_U$ with probability of at least $1-1/n^6$.

\section{Preliminaries}

\subsection{Notation}
\renewcommand{\EE}{E}
For the entire paper, $G$ will denote a $k$-partite hypergraph with vertex set $\Vin=(\Vin_1,\Vin_2,\ldots,\Vin_k)$, and hyperedges $\EE$. Each hyperedge $e\in\EE$ intersects every vertex set $\Vin_i$ in exactly one vertex.
For subset of vertices $A\subseteq \Vin$ we define the hyperedge set $E_A$ by
$E_A\triangleq \sset{e\in\EE\mid v\in e\implies v\in A}\;.$
In words, $E_A$ is a subset of hyperedges from $\EE$, where every $e\in\EE$ is also in $E_A$ if and only if all vertices that are in $e$ are also in $A$.
For two vertex sets $A,B$ we use $E_A(B)$ to denote the subset of hyperedges in $E_A$, where every $e\in E_A$ is also in $E_A(B)$ if and only if at least one vertex in $e$ is also in $B$.
We also $m_A(B)$ to denote the number of hyperedges in $E_A(B)$, where $m\triangleq m_{\Vin}(\Vin)$.
For a single vertex $v$, we use $d_U(v)$ to denote the number of hyperedges in $E_{U}(v)$, and $d(v)$ to denote $d_{\Vin}(v)$.
For a vertex set $U$, we define the induced subhypergraph $G[U]$ as a hypergraph with vertex set $U$ and hyperedge set $E_U$.
The set of non-isolated vertices over $G[U]$ is denoted by $\ZZ_U$.
That is, $\ZZ_U$ is the set of vertices in $U$ that have non-zero degree over $G[U]$.
We say that a vertex $v$ is $\Lambda$-heavy if $d(v)\geq \Lambda$, and $\Lambda$-light otherwise. We say that an hyperedge $e$ is $\Lambda$-heavy if it contains at least one $\Lambda$-heavy vertex, and $\Lambda$-light otherwise.
Note that a $k$-partite hypergraph $G$ with $m$ hyperedges contains at most $km/\Lambda$ vertices that are $\Lambda$-heavy, where this is sometimes referred to as the handshaking lemma.

We define a measure $\gs$ from every $U\subseteq \Vin$, by $\gs(U)\triangleq \prod_{i=1}^{k}|U\cap \Vin_i|$. We refer to $\gs(U)$ as the \emph{query measure} of $U$.
We say that an algorithm with oracle access has query measure at most $R$ if it only queries the oracle on sets $U$ for which $\gs(U)\leq R$.

Throughout the paper, $\eps$ denotes a fixed precision parameter, where we assume that $\eps\leq 1/4$. If this is not the case, we can always replace $\eps$ with $\min\set{\eps,1/4}$.
We also use $\eps'$ defined as $\eps'=\neweps$.

We use standard notation for the matrix multiplication exponents $\omega(a,b,c)$ to denote the smallest real number $w$ so that $n^a\times n^b$ by $n^b\times n^c$ matrices can be multiplied in $O(n^{w+\delta})$ time for all $\delta>0$. As is standard in the literature, for simplicity, we omit $\delta$ in the running time exponents of our algorithms and instead just write $\omega(a,b,c)$. 
Equivalently, we also write $\MM{n^a,n^b,n^c}$ for the running time of the matrix multiplication algorithm that multiplies $n^a\times n^b$ by $n^b\times n^c$ matrices. That is $n^{\omega(a,b,c)}$ and $\MM{n^a,n^b,n^c}$ denote the same running time.

\subsection{Tools}

~\\We also distinguish between $\Lambda$-heavy vertices, which are vertices with degree at least $\Lambda$, and $\Lambda$-light vertices, which have degree at most $\Lambda$.
We say that a hyperedge is $\Lambda$-heavy if it contains at least one $\Lambda$-heavy vertex, and $\Lambda$-light otherwise.
Throughout the paper we use the following observation.
A $k$-partite graph $G$ with $m$ edges, contains at most $km/\Lambda$ vertices that are $\Lambda$-heavy.

To compute the number of hyperedges in an induced subhypergraph $G[W]$ 
that has only $\polylog{n}$ hyperedges, we use the algorithm of \cite{DellLM22,bhattacharya2024faster}, which we refer to as $\algPrev(W,\eps)$.
This algorithm outputs a $\apm$ approximation for the number of hyperedges in the graph $G[W]$, while minimizing the number of queries to the oracle. However, it queries the oracle on set with maximum measure $\gs(W)$, and therefore we use this algorithm only when $\gs(W)\leq \gs(V)/m$.
\begin{theorem}[{\cite{DellLM22,bhattacharya2024faster}}]\label{thm:detect to count}
    Let $G=(V,E)$ be a $k$-partite hypergraph with $n$ vertices.
    There exists a randomized algorithm, denoted by $\algPrev(V,\eps)$.
    $\algPrev$ outputs $\hat{m}$ such that $\Pr{\hat{m}=m_V\apm}\geq 1-1/n^4\;.$
The algorithm makes $\frac{(k\log n)^{O(k)}}{\eps^2}$ queries to a CID-query oracle.
\end{theorem}
\paragraph*{Probabilistic Bounds.}
Let $V=V_1\sqcup\ldots\sqcup V_k$ be a set.
Let $V_p$ denote a random subset of $V$ where each element in $V$ is included independently with probability $p$.
We prove a concentration bound on the size of $V_p\cap V_i$, and on the size of $\gs(V_p)$, which as we previously mentioned is defined by $\gs(V_p)=\prod_{i=1}^k \abs{V_p\cap V_i}$.

\begin{claim}\label{claim:c2}
    \begin{align*}
        \Pr{\gs(V_p)\geq \gs(V)\cdot (p\cdot 6\log n)^k}\leq k/n^6\;.
    \end{align*}
\end{claim}
\begin{proof}
    We first prove that for every $i\in[k]$, it holds that
    \begin{align*}
        \Pr{\abs{V_p\cap V_i}\geq \abs{V_i}p\cdot 6\log n}\leq 1/n^6\;.
    \end{align*}
    This follows by Chernoff's inequality, as $\abs{V_p\cap V_i}$ is the sum of $n$ independent Bernoulli random variables with success probability $p$.
    By union bound over $i\in[k]$, we get that
    \begin{align*}
        \Pr{\gs(V_p)\geq \gs(V)\cdot (p\cdot 6\log n)^k}\leq k/n^6\;,
    \end{align*}
    which completes the proof.
\end{proof}

\newcommand{\clogA}{C}
\newcommand{\calc}[1]{\pgfmathparse{\clogA#1}\pgfmathprintnumber{\pgfmathresult}}
\newcommand{\clss}[1]{\log^{C{#1}} n}
\newcommand{\cgb}{\log^{(h^2)}n}
\newcommand{\cls}{\log^C n}
\newcommand{\clsm}{\log^5 n}
\newcommand{\clsa}{\log^4 n}
\newcommand{\Emh}{\Em{\mathrm{high}}}
\newcommand{\BEmh}{\overline{\Emh}}
\newcommand{\Emi}{\Em{\mathrm{ind}}}
\newcommand{\Emia}{\Em{\mathrm{ind}}(V,\Lambda)}
\newcommand{\Emib}{\Em{\mathrm{ind}}(U,\Lambda\cdot q)}
\newcommand{\Ems}{\Em{\mathrm{sample}}}
\newcommand{\Ebot}{\Em{\mathrm{\bot}}(V,\Lambda)}
\newcommand{\BEmia}{\overline{\mathcal{E}_{\mathrm{ind}}}(V,\Lambda)}
\newcommand{\BEmib}{\overline{\mathcal{E}_{\mathrm{ind}}}(U,\Lambda\cdot q)}
\newcommand{\BEms}{\overline{\mathcal{E}_{\mathrm{sample}}}}
\newcommand{\rPz}{\frac{4k\cdot \DD}{q\cdot  \K}}
\newcommand{\rPzB}{\frac{4k\cdot2^k \cdot \log^2 n}{\K}}
\newcommand{\rPzC}{\frac{5k\cdot2^k \cdot \log^2 n}{\Qd}}
\newcommand{\rPzk}{\frac{5k\cdot2^k \cdot \log^2 n}{\K}}
\newcommand{\II}{\mathcal{I}}

\section{The Recursive Algorithm}\label{sec:rec}
\paragraph{Organization.}
This section is dedicate to proving the following theorem, which captures the correctness guarantee that we use in \Cref{sec:complexity} for proving \Cref{thm4:main}. Informally, it shows that we can approximate the number of hyperedges given a good guess $L$ that bounds it.
\begin{restatable}[Approximating the number of hyperedges given a guess]{theorem}{LemTrunk}\label{lemma2:alg apx}
    For a vertex set $V$, a value $\W$, and a parameter $\eps$, there exists an algorithm $\apx(V,\W,\eps)$
    that produces either $\bbot$ or a value $\htt_V$ such that $\htt_V=m_V\apm$ with probability at least $1-1/n^6$.
    If $\W\leq m_V/2$, then the probability that the algorithm outputs $\bbot$ is at most $1/n^6$.
\end{restatable}

This section starts with an overview of our algorithm $\arec$, which is based on the recursive framework of \cite{tvetek2022approximate,CEW24}. See \Cref{lemma:simp} for a formal statement on the output of the algorithm.
We then explain how to implement the recursive algorithm using two randomized procedures, $\FindHalg(V,\Lambda)$ and $\Cialg(V,\vl,\Lambda)$, which find the set of $\Lambda$-heavy vertices $\vl$ and approximate $m_V(\vl)$, respectively.
In \Cref{ssec:bot} we show that the recursive algorithm $\arec$ outputs $\bot$ with probability at most $1/\log n$.
In \Cref{ssec:amplification} we prove \Cref{lemma2:alg apx}, where the algorithm $\apx$ specified in \Cref{lemma2:alg apx} can be viewed as an amplified version of the recursive algorithm $\arec$, that also employs a doubling technique to search for the right heaviness threshold.

~\\As promised, we start with an overview of the recursive algorithm  $\arec$.
We aim to get an algorithm for approximating the number of hyperedges in a $k$-partite ($k$-uniform) hypergraph, using a \HO, which is an oracle that detects whether a given set of vertices contains a hyperedge or not.

More formally, the input hypergraph $G$ has $k$ disjoint vertex sets $\Vin=\Vin_1,\ldots,\Vin_k$ and
a set of hyperedges $\EE$, where each hyperedge contains exactly $k$ vertices and intersects $\Vin_i$ in exactly one vertex, for each $i\in[k]$.
We denote the number of vertices by $n$ and the number of hyperedges by $m$.
The \HO takes as input a set of vertices $S\subseteq \Vin$. It outputs $1$ if there exists an hyperedge $e\in \EE$ such that $e\subseteq S$, and it outputs $0$ otherwise.
The algorithm holds the set of vertices $\Vin$ and has access to the \HO, but it can not access the hyperedges $\EE$ directly, apart from using this oracle.

We provide a randomized algorithm for outputting a value $\htt$ that is a $\apm$ approximation for $m$ \whp.
An important note that will be useful for the application of this approximation algorithm is that it only makes $\polylog{n}$ calls to the \HO,
while only querying the \HO with sets of vertices of $U$ for which $\prod_{i=1}^{k}\abs{U\cap \Vin_i}\leq \TO{n^k/m}$.
We define $\gs(U)\triangleq \prod_{i=1}^{k}\abs{U\cap \Vin_i}$, to denote this quantity.
We also note that instead of computing a $\apm$ approximation for $m$, we compute a $\apmp$ approximation for a smaller precision parameter  $\eps'\triangleq \frac{\min\set{\eps,1}}{4\log n}$. The value of $\eps'$ is set such that the interval $(1\pm \eps')^{2\log n}$ is contained in the interval $(1\pm \eps/2)$.

The algorithm $\arec$ takes as input two parameters: a set of vertices $V$ and a heaviness threshold $\Lambda$.
We assume we have some coarse estimate $\Lambda_0= \Theta(m/\polylog{n})$ for the number of hyperedges in the graph.
Later, we explain how to obtain such an estimate using doubling.
Recall that we say that a vertex is $\Lambda_0$-heavy if it has degree at least $\Lambda_0$, and otherwise we say that the vertex is $\Lambda_0$-light.
We say that a hyperedge is $\Lambda_0$-heavy if it contains at least one $\Lambda_0$-heavy vertex, and we say that the edge is $\Lambda_0$-light otherwise.
The algorithm approximates the number of $\Lambda_0$-heavy hyperedges and the number of $\Lambda_0$-light hyperedges separately.

Let $m_{\Lambda_0}$ denote the number of $\Lambda_0$-heavy hyperedges in the graph.
To approximate $m_{\Lambda_0}$, the algorithm finds a set of $\Lambda_0$-heavy vertices and uses it to compute an approximation $\htt_{\Lambda_0}$ for $m_{\Lambda_0}$.
To approximate the number of $\Lambda_0$-light hyperedges, the algorithm samples a subset of vertices $U$ from $V$, by choosing each vertex independently with probability $p=1/2$. Then it computes a $\apmp$ approximation for the number of hyperedges in the random induced subhypergraph $F=G[U]$ recursively, by invoking a recursive call $\arec(U,\Lambda_0\cdot p^k)$. Let $m_U$ denote the number of hyperedges in $F$, and let
$\htt_U$ denote the output of the recursive call.
The algorithm then outputs $\htt=\htt_{\Lambda_0}+\htt_U/p^k$ as the approximation for $m$ the number of hyperedges in $G$.

Our algorithm builds upon two procedures.
The first one, called $\FindHalg$, is given a set of vertices $V$ and a heaviness threshold $\Lambda$, and
finds a superset $V_\Lambda$ of the $\Lambda$-heavy vertices without any $\llow$-light vertices over $G$, for some $\llow= O(\Lambda/\polylog{n})$ 
The second one, called $\Cialg$, takes as input the previously computed set $V_\Lambda$ and approximates the number of hyperedges $m_V(\vl)$ 
which is the number of hyperedges in $G$ that contain at least one vertex in $V_\Lambda$.

\pagebreak
\noindent In \Cref{sec:count heavy,sec:find heavy}, we provide details on how to implement these procedures using the \HO.
Formally, the procedures that we use are the following.
The $\FindHalg$ procedure is:
\begin{restatable}[$\FindHalg{(V,\Lambda)}$]{blackbox}{reFHClean}\label{algo:find heavy bb1}
    Define $\llow\triangleq \Lambda/(4^k\cdot \cclog)$.
    \begin{description}
        \item[Input:] A subset of vertices $V$ and a heaviness threshold $\Lambda$.
        \item[Output:] A
            subset of vertices $V_\Lambda\subseteq V$, such that with probability at least $1-\frac{1}{n^4}$, for every $v\in V$:
            \begin{enumerate}
                \item If $d_V(v)\geq \Lambda$, then $v\in V_\Lambda$.
                \item If $d_V(v)< \llow$, then $v\notin V_\Lambda$.
            \end{enumerate}
    \end{description}
\end{restatable}

If the set of vertices $V_\Lambda$ outputted by the procedure $\FindHalg{(V,\Lambda)}$ indeed does not contain any $\llow$-light vertices, then it contains at most $k\cdot m_V/\llow$ vertices.
This follows from the fact that for any subset of vertices $V$ and any threshold $X$, the number of $X$-heavy vertices in $V$ is at most $k\cdot m_V/X$.
The procedure makes at most $O(\psi\cdot k\cdot m_V/\llow)$ queries to the \HO, where $\psi\triangleq 4^k\cdot k^2\log(k)\cdot\log^{k+2} n$ is a fixed parameter that depends on $k$ and $n$.
We also define a parameter $\zeta$ that will be set later.

We will use an additional procedure $\FindHalgB{(V,\Lambda)}$, which is a variant of $\FindHalg{(V,\Lambda)}$.
Informally, this procedure simulates the procedure $\FindHalg{(V,\Lambda)}$ while only allowing it to make $\BO{\psi\cdot \zeta}$ queries to the \HO.
If the simulation does not complete within $\BO{\psi\cdot \zeta}$ steps, or it does complete but the computed set $V_\Lambda$ contains more than $\zeta$ vertices, then the procedure $\FindHalgB(V, \Lambda)$ outputs $\bot$.
Otherwise, it outputs the set $V_\Lambda$ computed by the procedure $\FindHalg(V, \Lambda)$. These guarantees are summarize below.
\begin{restatable}[$\FindHalgB{(V,\Lambda)}$]{blackbox}{reFHBounded}\label{algo:find heavy bb bounded}
    Define $\llow\triangleq \Lambda/(4^k\cdot \cclog)$.
    \begin{description}
        \item[Input:] A subset of vertices $V$, a heaviness threshold $\Lambda$, and a threshold $\psi$.
        \item[Output:] Either $\bot$ or a
            subset of vertices $V_\Lambda\subseteq V$ containing at most $\zeta$ vertices,
            such that with probability at least $1-\frac{1}{n^4}$, for every $v\in V$:
            \begin{enumerate}
                \item If $d_V(v)\geq \Lambda$, then $v\in V_\Lambda$.
                \item If $d_V(v)< \llow$, then $v\notin V_\Lambda$.
            \end{enumerate}
    \end{description}
If $k\cdot m_V/\Lambda\leq \zeta$, then the probability of outputting $\bot$ is at most $1/n^4$.
\end{restatable}
Without an upper bound on the probability of outputting $\bot$, this procedure is useless, as it might output $\bot$ on every input.
To avoid this, we require that if $km_V/\llow \leq \zeta$, then the probability of outputting $\bot$ is at most $1/n^4$. Note that a $k$-partite hypergraph $G[V]$ with $m_V$ edges has no more than $k\cdot m_V/X$ vertices which are $X$-heavy, for every positive $X$. Thus, the condition $km_V/\llow \leq \zeta$ implies that the number of $\llow$-heavy vertices in $V$ is at most $\zeta$, and therefore that the number of $\Lambda$-heavy vertices in $V$ is at most $\zeta \cdot \llow/\Lambda$, which is also bounded by $\zeta$. Thus, the condition applies to smaller than $\zeta$ values of the size of the set of $\Lambda$-heavy vertices, and hence does not contradict the requirement that if the set of $\Lambda$-heavy vertices contains at least $\zeta$ vertices, then the procedure should output $\bot$.

~\\Next, the $\Cialg$ procedure is:
\begin{blackbox}[${\Cialg(V,V_{\Lambda},\llow)}$]\label{algo:count heavy}
    \begin{description}
        \item[Input:] A vertex set $V$, a vertex set $V_{\Lambda}\subseteq V$, a heaviness threshold $\llow$, and a precision parameter $\eps'$.
        \item[Output:]
            A value $\htt_{\Lambda}$, such that
            if every vertex $v\in V$ satisfies $d_V(v)\geq \llow$, then
            \begin{center}
                $\Pr{\htt_{\Lambda}=m_{V}(V_\Lambda)\cdot \apmp}\geq 1-\frac{1}{n^4}$.
            \end{center}
    \end{description}
\end{blackbox}

We now provide two implementations for the recursive algorithm, denoted by $\algCore{V,\Lambda}$ and $\algCorezB{V,\Lambda}$.
The reason we use two versions is that the first implementation uses the efficient procedure $\FindHalgB{(V,\Lambda)}$ to find a superset of the $\Lambda$-heavy vertices in $V$, while the second implementation is easier to analyze. We claim that if the first algorithm does not output $\bot$ and both algorithms use the same randomness, then they produce the same output.
\paragraph{Choosing the Parameters.}
We set the parameters as follows.
Let $\eps\in(0,1/2]$ be the precision parameter and let $\eps'=\neweps$.
We set $p=1/2$, $\Qd=\Qdval$,
$\zeta=k\cdot \cclog\cdot \frac{\Qd}{{\eps'}^2}\cdot \log^4 n$,
and $\psi=4^k\cdot k^2\log(k)\cdot\log^{k+2} n$.
We keep the specified symbols instead of their assigned value for readability.

~\\The algorithm $\algCorezNP$ is as follows.

\begin{algorithm}[H]
    \caption{$\algCorez{V,\Lambda}$}\label[algorithm]{alg:template recA}
    \setcounter{AlgoLine}{0}
    \KwIn{A vertex set $V$, a heaviness threshold $\Lambda$, and a precision parameter $\eps'$.}

    \medskip
    $V_\Lambda\gets \FindHalgB(V,\Lambda)$ \Comment*{\parbox[t]{3.5in}{$V_\Lambda$ is a superset of $\Lambda$-heavy vertices in $G[V]$ with no $\llow$-light vertices \whp.
}}

    \lIf{$V_\Lambda= \bot$}{\Return{$\bot$}}
\medskip

    \DontPrintSemicolon  \If{$\Lambda\leq 1$}{
        \Return $\algPrev(V,\eps')$; \hspace*{0.18in}$\triangleright\;$ \parbox[t]{3.5in}{ \texttt{Approximate the number of hyperedges in $G[V]$ using \Cref{thm:detect to count}.}}
    }
    \medskip

    $\htt_\Lambda\gets\Cialg(V,V_\Lambda,\tfrac{\Lambda}{2})$
    \Comment*{\parbox[t]{3.5in}{$\htt_\Lambda$ is a $\apmp$ approximation for $m_V(V_\Lambda)$ (the number of hyperedges intersecting $V_\Lambda$)}}
$U\gets V[p]-V_\Lambda$ \Comment*{\parbox[t]{3.5in}{$U$ is a random subset of $V-V_\Lambda$}}\label{line:5}
    \medskip

    $\htt_{U}\gets \algCore{U ,\Lambda\cdot p^k}$ \\
    \smallskip
    \Return{$\htt_{\Lambda}+\htt_{U}/p^k$}
\end{algorithm}

We emphasize that if one recursive call of the algorithm $\algCorezNP$ outputs $\bot$ then every preceding call, as well as the first call, also outputs $\bot$.

~\\The algorithm $\algCorezNPB$ is as follows.

\begin{algorithm}[H]
    \caption{$\algCorezB{V,\Lambda}$}\label[algorithm]{alg:template recB}
    \setcounter{AlgoLine}{0}
    \KwIn{A vertex set $V$, a heaviness threshold $\Lambda$, and a precision parameter $\eps'$.}

    \medskip

    $V_\Lambda\gets \FindHalg(V,\Lambda)$ \Comment*{\parbox[t]{3.5in}{$V_\Lambda$ is a superset of $\Lambda$-heavy vertices in $G[U]$ with no $\llow$-light vertices \whp}}
\DontPrintSemicolon  \If{$\Lambda\leq 1$}{
        \Return $\algPrev(V,\eps')$; \hspace*{0.18in}$\triangleright\;$ \parbox[t]{3.5in}{ \texttt{Approximate the number of hyperedges in $G[V]$ using \Cref{thm:detect to count}.}}
    }
    \medskip

    $\htt_\Lambda\gets\Cialg(V,V_\Lambda,\tfrac{\Lambda}{2})$
    \Comment*{\parbox[t]{3.5in}{$\htt_\Lambda$ is a $\apmp$ approximation for $m_V(V_\Lambda)$ (the number of hyperedges intersecting $V_\Lambda$)}}
$U\gets V[p]-V_\Lambda$ \Comment*{\parbox[t]{3.5in}{$U$ is a random subset of $V-V_\Lambda$}}\label{line2:5}
    \medskip

    $\htt_{U}\gets \algCorezB{U ,\Lambda\cdot p^k}$ \\
    \smallskip
    \Return{$\haT_{\Lambda}+\htt_{U}/p^k$}
\end{algorithm}

Informally, the guarantees of the algorithm $\algCorez{V,\Lambda}$ are as follows.
It outputs either $\bot$ or an approximation $\htt$ that falls inside the interval $m_V\apm\pm O_k(\log^3 n\cdot \Lambda/\eps')$ with probability at least $1-1/\log n$.
Moreover, if $\frac{k\cdot m_V}{\Lambda}\leq \frac{\zeta}{\log^4 n}$, then the algorithm outputs $\bot$ with probability at most $1/\log^3 n$.
The formal guarantees of the algorithm are specified in the following lemma.
\begin{restatable}[Guarantees for the $\arec$ Algorithm]{lemma}{LemCoreSimp}\label{lemma:simp}
    \sloppy{
    For every $\eps\in(0,1/4]$,
    the algorithm $\arec(V,\Lambda)$ outputs either $\bot$ or a value $\htt_V$ such that}
    \begin{align*}
         & \Pr{\arec(V,\Lambda)=m_V\cdot (1\pm \eps/2) \pm \Qd\cdot 2\Lambda/\eps'}\geq 1-\frac{1}{\log n}\;,
    \end{align*}
    where $\eps'=\neweps$, and $\Qd=\Qdval$.
    If $\frac{k\cdot m_V}{\Lambda}\leq \frac{\zeta}{\log^4 n}$, then the probability that the algorithm outputs $\bot$ is at most $2/\log^3 n$.
\end{restatable}
We explain how \Cref{lemma:simp} is used to prove \Cref{lemma2:alg apx}.
The first tool we need is a method to amplify the probability that the recursive algorithm $\arec$, specified in \Cref{lemma:simp}, outputs a value $\htt$ that falls within the interval $I_\Lambda \triangleq m_V\cdot (1\pm \eps/2) \pm \Qd\cdot 2\Lambda/\eps'$.
For that, we use the median of means technique, a standard technique for amplifying the success probability of an algorithm.
This includes executing the algorithm $\arec$ for $\log n$ times and taking the median of the outputs.
If too many executions yield $\bot$, then the algorithm outputs $\bot$.
Let $\crec$ denote this amplified algorithm, which we formally define in \Cref{lemma:alg trunk}.
The output of $\crec$ is either $\bot$ or falls inside the interval $I_\Lambda$ for some $\Lambda$ with probability at least $1-1/n^6$.

The second tool we need is doubling, which is used to find the right heaviness threshold $\Lambda$.
We start with the maximal possible heaviness threshold $\Lambda= \gs(V)=n^k$, and invoke a call to $\crec$ with output $\hat{s}$.
If $\hat{s}=\bot$, then we set $\Lambda\gets \Lambda/2$ and invoke another call to $\crec$ with the new heaviness threshold.
We continue in this manner until we find a heaviness threshold $\Lambda$ such that the output of $\crec$, denoted by $\hat{s}$, is not $\bot$.
We then check if $\hat{s}\geq \rho \cdot \Lambda$, where $\rho$ is a parameter we set later. If this condition is met, then we output $\hat{s}$ as our final output.
Otherwise, we set $\Lambda\gets \Lambda/2$ and again invoke a call to $\crec$.

This completes the description of how \Cref{lemma2:alg apx} is derived using \Cref{lemma:simp}. In \Cref{ssec:bot} we prove \Cref{lemma:simp} and in \Cref{ssec:amplification} we prove \Cref{lemma2:alg apx}. \renewcommand{\VV}{U}
\subsection{Analyzing the Recursive Algorithm $\arec$}\label{ssec:bot}
To prove \Cref{lemma:simp}, we split the proof into two parts. First, we show that the probability of outputting $\bot$ is small (\Cref{prop2:no bot}).
Second, we analyze a slightly different algorithm, denoted by $\algCorezNPB$, which is detailed in \Cref{alg:template recB}.
Recall that $\brec$ operates like $\arec$ but uses the procedure $\FindHalg{(V,\Lambda)}$ instead of $\FindHalgB{(V,\Lambda)}$, and therefore never outputs $\bot$ (\Cref{lemma:simp B}). 
Note that if the algorithms $\arec$ and $\brec$ use the same randomness and the former does not output $\bot$, then the two algorithms produce the same output.
Therefore, to bound the probability that $\arec$ outputs a good approximation, we analyze probability that $\brec$ outputs a good approximation and the probability that $\arec$ does not output $\bot$.

In this subsection we prove that the probability that the algorithm $\algCorez{V,\Lambda}$ outputs $\bot$ is bounded by $O(1/\polylog{n})$ as formally stated in the following lemma.
\begin{lemma}\label{prop2:no bot}
    \sloppy{
        Recall that $\zeta= k\cdot \cclog\cdot \frac{\Qd}{{\eps'}^2}\cdot \log^4 n$.
        Define the event
        $\Ebot\triangleq\set {\algCore{V,\Lambda}=\bot}$.}
Then, for every $\Lambda$ such that
    \begin{align}
        \frac{k\cdot m_V}{\Lambda} \leq \frac{\zeta}{\log^4 n}\;,
    \end{align}
    we have that $\Pr{\Ebot}\leq \frac{2}{\log^3 n}$.
\end{lemma}

We need the following definition of the depth of the recursive calls.
\renewcommand{\DD}[1][\Lambda]{\mathrm{Depth}\brak{#1}}
\newcommand{\Rmax}{R_{\max}}

\begin{restatable}[{$\DD$}]{definition}{defDepth}\label{def:Depth}
    \sloppy{We define the depth of the call $\algCore{V,\Lambda}$ by
        $\DD[\Lambda]=\max\set{\ceil*{\log_{1/p^k}(\Lambda)},0}$, where $p=1/2$.}
\end{restatable}
We can assume without the lost of generality that $\Lambda\leq n^{k-1}$,
as a $k$-partite hypergraph with $n$ vertices has at most $n^{k}$
hyperedges and maximal degree of $n^{k-1}$.
Therefore, there are no $(2n^{k-1})$-heavy vertices
and we get that
\begin{align*}
    \DD[\Lambda] + 1\leq (k-1)\log n/ \log(1/p^k) + 1\leq \log n +1-\log n/k\leq \log n\;.
\end{align*}

~\\ Consider the recursive calls $\sset{\algCorez{\VV_i,\Lambda_i}}_{i=0}^{\DD[\Lambda_0]}$ made by the recursive algorithm, where $\VV_0$ and $\Lambda_0$ are the values given to the initial invocation of $\arec$.
For every $i\in \zrn{\DD[\Lambda_0]}$ we use $\VV_i$ to denote the set of sampled vertices in the $i$-th recursive call, and we use $\Lambda_i$ as the heaviness threshold in the $i$-th recursive call, where $\Lambda_i=\Lambda_i\cdot p^{ki}$.
We use $m_i$ to denote the exact number of hyperedges in $G[\VV_i]$, we use $R_i$ to denote the ratio $k\cdot m_i/\Lambda_i$, and we also use $\Rmax$ to denote the maximum value of $R_i$ over $i\in\sset{0,1,\ldots, \DD[\Lambda_0]}$.
Before proving \Cref{prop2:no bot}, we prove the following claim that upper-bounds the ratio $R_i$ in each recursive call.

\begin{claim}\label{claim:ratio}
    For every $i\in\sset{0,1,\ldots, \DD[\Lambda_0]}$ and every positive $X$, we have that $\Pr{R_i\geq R_0 \cdot X}\leq 1/X$.
\end{claim}
\newcommand{\mi}{m_{\VV_i}}
\begin{proof}[Proof of \Cref{claim:ratio}]
    We first prove that for every $i\in \zrn{\DD[\Lambda_0]}$ we have that $\Exp{m_i}\leq m_0 \cdot p^{ik}$.
    Since $R_i=k\cdot m_i/\Lambda_i$, we get that $\Exp{R_i}=k\cdot \Exp{m_i}/\Lambda_i\leq k\cdot m_0/\Lambda_0 = R_0$.
    We then use Markov's inequality to get that $\Pr{R_i\geq R_0 \cdot X}\leq 1/X$. To complete the proof, we therefore need to show that $\Exp{m_i}\leq m_0 \cdot p^{ik}$.

    For that, we define a random variable $s_i$ as follows
Define $C_0\triangleq V$, and $C_{i+1}=C_i[p]$.
    That is, $C_{i+1}$ is obtained by sampling each vertex from $C_i$ independently with probability $p$.
    Recall that the set $\VV_{i+1}$ is also obtained by sampling every vertex from a \emph{subset} of $\VV_i$ independently with probability $p$,
    where this subset is defined as the set of vertices in $\VV_i$ that were not identified as $\Lambda_i$-heavy in the $i$-th recursive call.
Let $s_i$ denote a random variable that is equal to the number of hyperedges in $G[C_i]$.
    We show that $\Exp{s_i}=m_0\cdot p^{ki}$ and that if we use the same randomness then $m_i$ is smaller than $s_i$, which means that $\Exp{m_i}\leq \Exp{s_i}$, which completes the proof.

    To see why $\Exp{s_i}=m_{0}\cdot p^{ki}$, note that $\Exp{s_{i+1}\mid s_i=y}=y\cdot p^k$, and therefore $\Exp{s_{i+1}}=\Exp{s_i}\cdot p^k$, which means that $\Exp{s_i}=s_0\cdot p^{ki}=m_{0}\cdot p^{ki}$.

    Next, we show that $\Exp{m_i}\leq \Exp{s_i}$ and therefore conclude that $\Exp{m_i}\leq m_0\cdot p^{ki}$.
    As we previously explained, if we use the same randomness to sample the sets $\VV_i$ and $C_i$, then it must be the case that $\VV_i\subseteq C_{i}$, and therefore the number of hyperedges in $G[\VV_i]$ is at most the number of hyperedges in $G[C_i]$.
    In other words, we showed that for every value $z\geq 0$, it holds that $\Pr{m_i\geq z}\leq \Pr{s_i\geq z}$, which implies that $\Exp{m_i}\leq \Exp{s_i}$, and therefore completes the proof.
    
\end{proof}
We are now ready to prove \Cref{prop2:no bot}, which bounds the probability that $\arec$ outputs $\bot$.
\begin{proof}[Proof of \Cref{prop2:no bot}]
    \newcommand{\eva}{\mathcal{F}_{\bot}}
    \newcommand{\evb}{\mathcal{F}_{{\mathrm{ratio}}}}
    We define two sequences of events.
    For $i\in \zrn{\DD[\Lambda_0]}$,
    let $\evb(i)$ denote the event that $\sset{R_i\leq \zeta}$, and
    let $\eva(i)$ denote the event that the $i$-th recursive call to $\FindHalgB$ does not output $\bot$. We also define the events
    \begin{align*}
        \eva\triangleq\set{\bigcap_{i=0}^{\DD[\Lambda_0]} \eva(i)}, &  &
        \evb\triangleq\set{\bigcap_{i=0}^{\DD[\Lambda_0]} \evb(i)}\;.
    \end{align*}
Note that if the event $\eva$ occurs then no call to $\FindHalgB$ outputs $\bot$, and therefore the algorithm does not output $\bot$. This means that the events $\Ebot$ and $\eva$ are complements, i.e., $\Pr{\Ebot} =1-\Pr{\eva}$.
    Using the definition of conditional probability, we have that
    \begin{align*}
        \Pr{\eva}\geq \Pr{\eva\mid \evb}\cdot\Pr{\evb}\;.
    \end{align*}
    We analyze the two terms on the right-hand side of the inequality.
    Fix some $i\in \zrn{\DD[\Lambda_0]}$.
    Note that by the guarantee of the $\FindHalgB$ procedure, we have that
    if $R_i\leq \zeta$ then the $i$-th call $\FindHalgB(\VV_i,\Lambda_i,\zeta)$ does not output $\bot$ with probability at least $1-1/n^4$.
    In other words, if $\evb(i)$ occurs $\eva(i)$ occurs with probability at least $1-1/n^4$.
    We therefore get that $\Pr{\eva(i)\mid \evb(i)}\geq 1-1/n^4$.
    We also get $\Pr{\eva\mid \evb}\geq 1-\log n/n^4$ by a union bound over at most $\DD[\Lambda_0]< \log n$ events.

    We are left with proving that $\Pr{\evb}\geq 1-1/\log^3 n$, which we show under the assumption that $R_0\leq \zeta/\log^4 n$.
    We have the following inequalities:
    \begin{align*}
        \Pr{\evb(i)}=\Pr{R_i> \zeta}
        \leq \Pr{R_i> R_0\cdot \log^4 n}
        \leq 1/\log^4 n\;,
    \end{align*}
    where the first inequality follows by the assumption that $R_0 = k\cdot m_{\VV_0}/\Lambda_0\leq \zeta/\log^4 n$, and the second inequality follows by \Cref{claim:ratio}.
    Therefore, we get that $\Pr{\evb}\geq 1-1/\log^3 n$
by a union bound over at most $\DD[\Lambda_0]< \log n$ events.
    We plug in the computed values of $\Pr{\eva\mid \evb}$ and $\Pr{\evb}$, and get that
    \begin{align*}
        \Pr{\eva}\geq \Pr{\eva\mid \evb}\cdot\Pr{\evb}
        \geq \brak{1-\frac{\log n}{n^4}}\cdot \brak{1-\frac{1}{\log^3 n}}
        \geq 1-\frac{2}{\log^3 n}\;.
    \end{align*}
    As previously mentioned, we have $\Pr{\Ebot} =1-\Pr{\eva}$. Therefore, we get that $\Pr{\Ebot}\leq 2/\log^3 n$, which completes the proof.
\end{proof}

Recall that if we run the algorithms $\arec(V,\Lambda)$ and $\brec(V,\Lambda)$ with the same randomness and $\arec(V,\Lambda)$ does not output $\bot$, then the two algorithms produce the same output.
Therefore, to bound the probability that the algorithm $\arec(V,\Lambda)$ outputs a correct approximation it suffices to bound the probability that the algorithm $\brec(V,\Lambda)$ outputs a correct approximation and the probability that the algorithm $\arec(V,\Lambda)$ does not output $\bot$.
We already obtained an upper bound on the probability of the latter event in \Cref{prop2:no bot}.
The next lemma provides us with an upper bound on the probability that the algorithm $\brec(V,\Lambda)$ outputs a correct approximation.

\begin{restatable}[Guarantees for the $\algCorezNPB$ Algorithm]{lemma}{LemCoreSimpB}\label{lemma:simp B}
    For every $\eps\in(0,1/2]$ and every $\K\geq 1$, we have that
\begin{align*}
         & \Pr{\algCorezB{V,\Lambda}=m_V\cdot (1\pm \eps/2) \pm \K\cdot 2\Lambda/\eps'}\geq 1-\rPzB\;,
    \end{align*}
    where $\eps'=\neweps$.
    Moreover, if $\Lambda\leq m_V\cdot \frac{{\eps'}^2}{\K}$, we get that
    \begin{align*}
        \Pr{\algCorezB{V,\Lambda}=m_V\apm}\geq 1-\rPzB\;.
    \end{align*}
\end{restatable}
The proof that the algorithm $\algCorezB{V,\Lambda}$ outputs a correct approximation follows the same line as in \cite{CEW24}, and we therefore include it in the appendix (\Cref{sec:proof of B}) for completeness.

\begin{proof}[Proof of \Cref{lemma:simp} using \Cref{lemma:simp B}]
    \newcommand{\EbotC}{\mathcal{E}_{\bot}^c(V,\Lambda)}
    \newcommand{\EA}{\mathcal{E}_{A}}
    \newcommand{\EB}{\mathcal{E}_{B}}
    \Cref{prop2:no bot} guarantees that the algorithm $\arec(V,\Lambda)$ outputs $\bot$ with probability at most $2/\log^3 n$ assuming that
    $\frac{k\cdot m_V}{\Lambda}\leq \frac{\zeta}{\log^4 n}$.
    We are left with proving that if the output of $\arec(V,\Lambda)$ is not $\bot$, then the output falls inside the interval $m_V\cdot (1\pm \eps/2) \pm \Qd\cdot 2\Lambda/\eps'$ with probability at least $1-1/\log n$.

    Note that assuming that the output of $\arec(V,\Lambda)$ is not $\bot$, the output of $\arec(V,\Lambda)$ and $\brec(V,\Lambda)$ is the same.
    Moreover, the probability that the output of $\brec(V,\Lambda)$ falls inside the interval $m_V\cdot (1\pm \eps/2) \pm \K\cdot 2\Lambda/\eps'$ is at least $1-\rPzB$ by \Cref{lemma:simp B}.
    Note that $1-\rPzB$ is at least $1-\frac{4}{5\log n}$, by plugging in $\gamma=\Qd=\Qdval$, we get that if the output of $\arec(V,\Lambda)$ is not $\bot$, then the probability that the output of falls inside the interval $m_V\cdot (1\pm \eps/2) \pm \Qd\cdot 2\Lambda/\eps'$ is at least $1-\frac{1}{\log n}$, which completes the proof.
    
\end{proof}

\subsection{Amplifying the Recursive Algorithm $\arec$}\label{ssec:amplification}
We now show how to amplify the algorithm $\arec$ to get a $\apm$ approximation for the number of hyperedges in a $k$-partite hypergraph.
The main result in this subsection is as follows.

Informally, the algorithm either outputs an approximation $\htt_V$ for
$m_V$ or outputs $\bbot$.
If $m_V\geq \W$, then the algorithm outputs a $\apm$ approximation for $m_V$ with high probability.
In case that $m_V< \W$ the algorithm could output the correct approximation, but will probably output $\bbot$.
To implement the algorithm $\apx(V,\W,\eps)$ we need another algorithm ,$\algWrapD{V,\W,\eps}$, where $\apx$ is a wrapper around $\algWrapD{V,\W,\eps}$.
We state the guarantees of the algorithm $\algWrapD{V,\W,\eps}$ in the following lemma.
\begin{lemma}\label{lemma:alg trunk}
    Let $g(k)\triangleq \cclog$.
    The algorithm $\algWrapD{V,\W,\eps}$ either outputs $\bot$ or a value $\htt_V$ such that
    \begin{align*}
        \Pr{\htt_V=m_V\cdot (1\pm \eps/2) \pm \W\eps/8}\geq 1-1/n^6\;.
    \end{align*}
If $\W\geq m_V/g(k)$, then the probability that the algorithm outputs $\bot$ is at most $1/n^7$.
\end{lemma}

We explain how to prove \Cref{lemma2:alg apx} using \Cref{lemma:alg trunk}.
\begin{proof}[Proof of \Cref{lemma2:alg apx} using \Cref{lemma:alg trunk}]
    \newcommand{\httt}{\hat{t}}
    We start with a description of the algorithm $\apx(V,\W,\eps)$.
    \begin{mdframed}[frametitle={$\apx(V,\W,\eps)$}, frametitlealignment=\centering, backgroundcolor=gray!10]
        For $i=0$ to $i=k\log n$:
        \begin{enumerate}
            \item Let $\htt_i\gets \algWrapD{V,\W\cdot 2^i,\eps}$.
            \item If $\htt_i=\bot$, then continue to the next iteration.
            \item If $\htt_i\geq \W\cdot 2^i$, then output $\htt_i$.
            \item Otherwise, if $\htt_i< \W\cdot 2^i$ output $\bbot$.
        \end{enumerate}
        Output $\bbot$.
    \end{mdframed}
Let $\htt_i$ denote the output of the $i$-th execution of the algorithm $\algWrapDNP$.
    Denote by $s$ the number of iterations that an execution of the algorithm does, and note that $s\leq k\log n + 1$.
We denote the final output of the algorithm by $\hat{t}$, which is either $\bbot$, or $\htt_s$.
We say that an approximation $\htt_i$ is \emph{good} if \begin{align*}
        \htt_i
        =m_V\cdot (1\pm \eps/2) \pm (\W\cdot 2^i)\cdot \eps/8\;.
    \end{align*}
    For every $i\leq s$ for which $\W\cdot 2^i\geq m_V/g(k)$, let $\FF(i)$ denote the event that $\htt_i$ is good.
For every other $i\leq s$, that is, for which $\W\cdot 2^i< m_V/g(k)$, let $\FF(i)$ denote the event that $\htt_i$ is good or that $\htt_i=\bot$.
    We make this distinction because if $i$ satisfies $\W\cdot 2^i< m_V/g(k)$,
    then we can use the guarantee of \Cref{lemma:alg trunk}, which states that the algorithm $\crec$ does not output $\bot$ with high probability.
Let $\FF$ denote the intersection of the events $\FF(i)$ for every $i\leq s$.
    We have that $\Pr{\FF}\geq 1-s/n^7$ by \Cref{lemma:alg trunk} and a union bound.
For the rest of this proof we assume that $\FF$ occurs.
    The probability that $\FF$ does not occur is at most $s/n^7\leq 1/n^6$, and therefore we can ignore this event.

We prove the output $\hat{t}$ is either $\bbot$ or $m_V\apm$, and that the algorithm completes its execution before iteration $k\log n$.
First, we show that there exists an index $i$ for which $\htt_i\neq \bot$.
For any $i$ such that $\W\cdot 2^i\geq m_V/g(k)$, we have that if $\FF(i)$ occurs, then by definition $\htt_i\neq \bot$.
    If $\htt_i\neq \bot$, then the algorithm stops before iteration $(i+1)$ and outputs either $\bbot$ or an approximation for $m_V$.
If the algorithm does not reach iteration $i$, then it outputs either $\bbot$ or an approximation for $m_V$ in some iteration $j$ for $j< i$.
We conclude that the algorithm stops before iteration $k\log n$ as for $i=k\log n/2$ we have that $\W\cdot 2^i\geq m_V/g(k)$.

    Next, we show that if $\W\geq m_V/2$, then $\hat{t}\neq \bot$\
If $\W\geq m_V/2$, and since we conditioned on $\FF$ then $\FF(1)$ occurs, then the output $\htt_1$ is good and therefore
    \begin{align}
        \htt_1
        =         & m_V\cdot (1\pm \eps/2) \pm \W\cdot \eps/8  \\
        \subseteq & m_V\cdot (1\pm \eps/2) \pm m_V\cdot \eps/8 \\
        \subseteq & m_V\apm \;. \label{eq22:h1}
    \end{align}
    Moreover, we get that $\htt_1\geq m_V(1-\eps)\geq m_V/2\geq \W$, which means that $\htt_1\geq \W$, and therefore $\httt$ the output of the algorithm is equal to $\htt_1$.
    Moreover, by \Cref{eq22:h1} we get that $\htt_1=m_V\apm$, and since $\httt=\htt_1$, we get that $\httt=m_V\apm$.
This completes the proof that if $\W\geq m_V/2$, then $\hat{t}\neq \bot$.

~\\Finally, we show that if $\httt=\htt_s$, then $\httt=m_V\apm$.
    The case that $\httt=\htt_s$ occurs only if $\htt_s\geq \W\cdot 2^s$, so it suffices to show that if $\htt_s$ is the output of $\crec(V,\W\cdot 2^s,\eps)$ and $\htt_s\geq \W\cdot 2^s$, then $\htt_s=m_V\apm$ with probability at least $1-1/n^6$. If $\FF$ occurs, then $\htt_s=m_V\apm$.

For this, we prove the following properties under the assumption that $\htt_s$ is good, i.e., that $\htt_s = m_V\cdot (1\pm \eps/2) \pm (\W\cdot 2^i)\cdot \eps/8$.
    \begin{enumerate}
        \item If $2^s\cdot \W\geq 4m_V$, then $\htt_s< 2^s\cdot \W$.
        \item If $2^s\cdot \W\leq 4m_V$, then $\htt_s=m_V\apm$.
    \end{enumerate}
    Before proving these properties, we explain why they imply that $\httt=m_V\apm$.
    By Property (2), we get that if $2^s\cdot \W\leq 4m_V$ then $\htt_s=m_V\apm$, and therefore $\httt=m_V\apm$.
    We want to show that indeed $2^s\cdot \W\leq 4m_V$.
    To this end, we assume towards a contradiction that $2^s\cdot \W>4m_V$. However, by Property (1) this would give that $\htt_s< 2^s\cdot \W$, which is a contradiction to the assumption that $\httt=\htt_s$, because the latter can only occur when $\htt_s\geq 2^s\cdot \W$.
    Thus, if $\httt=\htt_s$, then $\httt=m_V\apm$. It remains to prove the two properties.
    \paragraph{Proof of Property (1).}
    If $2^s\W\geq 4m_V$ then
    $m_V\apmt\leq 2^s\W (1+\eps)/4$. Further, we condition on $\htt_s$ being good, which means that
    \begin{align*}
        \htt_s & =m_V\apmt \pm\W\eps/8                 \\
               & \leq 2^s\W(1+\eps)/4 +2^s\W\eps/8     \\
               & \leq 2^s\W\cdot ((1+\eps)/4 + \eps/8) \\
               & \leq 2^s\W\cdot 3/4 < 2^s\W\;,
\end{align*}
where the penultimate inequality follows as $\eps\leq 1/4$. This completes the proof of Property (1).
    \paragraph{Proof of Property (2).}
    If $2^s\W\leq 4m_V$ then, since we condition on $\htt_s$ being good, we have
    \begin{align*}
        \htt_s=m_V\apmt \pm2^s\W\eps/8
        \subseteq m_V\apmt \pm m_V\eps/2
        \subseteq m_V\apm\;.
    \end{align*}
This completes the proof of Property (2), and therefore the proof that if $\httt=\htt_s$, then $\httt=m_V\apm$.

    To conclude, we have that the algorithm outputs either $\bbot$ or a value $\httt$. The value $\httt$ is equal to $m_V\apm$ with probability at least $1-1/n^6$.
    If $\W\geq m_V/2$, then the probability that the algorithm outputs $\bbot$ is at most $1/n^6$.
    
\end{proof}
It remains to prove \Cref{lemma:alg trunk}.
\begin{proof}[Proof of \Cref{lemma:alg trunk}]
    We start with a description of the algorithm $\algWrapD{V,\W,\eps}$.
\begin{mdframed}[frametitle={$\algWrapD{V,\W,\eps}$}, frametitlealignment=\centering, backgroundcolor=gray!10]
        \begin{enumerate}
            \item Set $\eps'\gets\neweps\;,\; \Lambda\gets \W \cdot{\eps'}^2/\Qd\;,\; r\gets \rr$.
            \item Make $r$ independent calls to the algorithm $\algCorez{V,\Lambda}$, and denote the result of the $i$-th call by $\htt_i$.
            \item Let $b$ denote the number of calls that output $\bot$.
                  If $b\geq r/10$, then output $\bot$.
            \item Let $\htt$ denote the median of the $r-b$ calls that do not output $\bot$, and output $\htt$.
        \end{enumerate}
    \end{mdframed}
We first prove that if $\W\geq m_V/g(k)$ then the probability that the algorithm outputs $\bot$ is at most $1/n^7$.
    We show that if $\W\geq m_V/g(k)$ then
    \begin{align}
        \frac{k\cdot m_V}{\Lambda} \leq \frac{\zeta}{\log^4 n}\;. \label{eq:kml}
    \end{align}
    Recall that by \Cref{prop2:no bot} if \Cref{eq:kml} holds, then the probability that the algorithm $\arec(V,\Lambda)$ outputs $\bot$ is at most $2/\log^3 n$.
    Recall that $g(k)= \cclog$,
    $\zeta= k\cdot \cclog\cdot \frac{\Qd}{{\eps'}^2}\cdot \log^4 n$, and therefore
    $\zeta= k\cdot \frac{\Qd}{{\eps'}^2}\cdot g(k)\cdot \log^4 n$.
    We therefore get that
\begin{align*}
        \frac{k\cdot m_V}{\Lambda} & \leq \frac{\zeta}{\log^4 n}                       &      & \iff \\
        \frac{k\cdot m_V}{\Lambda} & \leq \frac{k\cdot \frac{\Qd}{{\eps'}^2}\cdot g(k)
        \cdot \log^4 n}{\log^4 n}  &                                                   & \iff        \\
        \frac{m_V}{\Lambda}        & \leq \frac{\Qd}{{\eps'}^2}\cdot g(k)              &      & \iff \\
        \frac{m_V}{g(k)}           & \leq  \frac{\Qd}{{\eps'}^2}\cdot \Lambda\;.
\end{align*}
    Note that this completes the proof that if $\W\geq m_V/g(k)$ then \Cref{eq:kml} holds, because $\frac{\Qd}{{\eps'}^2}\cdot \Lambda=\W$.
    We conclude that the probability that each of the $r$ calls outputs $\bot$ is at most $2/\log^3$.
    This means that $b$ is a binomial random variable with $r$ trails and success probability at most $2/\log^3 n$. Therefore, $\Exp{b}\leq 2r/\log^3 n$, and by Chernoff's inequality, we get that $\Pr{b\geq r/10}\leq 2^{-r}$ which is at most $1/n^9$, as $r=\rr$.

We are left with proving that if the algorithm outputs a value $\htt$, then
    \begin{align*}
        \Pr{\htt=m_V\cdot (1\pm \eps/2) \pm \W\eps/8}\geq 1-\frac{1}{n^6}\;.
    \end{align*}
    In other words, we need to show that if if $b<r/10$, then the median of all the values $\htt_i$ for which $\htt_i\neq \bot$ falls inside the interval $m_V\cdot (1\pm \eps/2) \pm \W\eps/8$ with probability at least $1-1/n^6$.
To show that, it suffices to prove that for every $i\in [r]$ for which $\htt_i\neq \bot$, it holds that
    \begin{align*}
        \Pr{\htt_i=m_V\cdot (1\pm \eps/2) \pm \W\eps/8}\geq 1-\frac{1}{n^6}\;,
    \end{align*}
    and then to use the median of means trick, detailed in \Cref{claim:med trick0} in \Cref{app:median}, to prove that the median $\htt$ falls inside this interval with probability at least $1-1/n^7$. Note that $\htt$ is the median of at least $9r/10$ values.
    By \Cref{lemma:simp}, we have that if $\htt_i\neq \bot$, then
    \begin{align*}
        \Pr{\htt_i=m_V\cdot (1\pm \eps/2) \pm \Qd\cdot 2\Lambda/\eps'}\geq 1-\frac1{\log n}\;.
    \end{align*}
    We show that the interval $m_V\cdot (1\pm \eps/2) \pm \Qd\cdot 2\Lambda/\eps'$ is contained in the interval $m_V\cdot (1\pm \eps/2) \pm \W\eps/8$, which follows as $\Qd\cdot 2\Lambda/\eps'\leq \W\eps/8$.
    To see that, we plug in $\Lambda=\W\cdot {\eps'}^2/\Qd$ and get that
    \begin{align*}
        \Qd\cdot 2\Lambda/\eps' = 2\W\cdot \eps' \leq \W\cdot \eps/8\;,
    \end{align*}
    where the last inequality follows as $\eps'\leq \eps/4$.
    This completes the proof.
    
\end{proof}
 
\newcommand{\im}{i^*}
\newcommand{\tc}{6\log n}
\newcommand{\tcb}{10k\log n}
\newcommand{\nmi}{n_\mathrm{min}}
\section{Proof of \Cref{thm4:main}}\label{sec:complexity}
The main result of this section is the proof of \Cref{thm4:main} given the implementations of $\Cialg$ and $\FindHalg$, which are presented in
\Cref{sec:count heavy,sec:find heavy} respectively.
Recall that for a $k$-disjoint set $V=V_1\sqcup V_2\sqcup \ldots \sqcup V_k$, and a set $U\subseteq V$, we use $\gs(U)$ to denote the size $\prod_{i=1}^k \abs{U\cap V_i}$.
\thmMain* 

We describe how to implement the algorithm $\algmain$ using the algorithm $\apx$.
\begin{mdframed}[frametitle={The algorithm $\algmain$}, frametitlealignment=\centering, backgroundcolor=gray!10]
For $i=0$ to $i=k\log n$ do:
    \begin{enumerate}
        \item $\htt\gets \apx(\Vin,\gs(\Vin)\cdot 2^{-i},\eps)$.
        \item If $\htt\neq \bbot$, then output $\htt$.
\end{enumerate}
\end{mdframed}

As promised by \Cref{lemma2:alg apx}, if $\W\leq m/2$, then the algorithm outputs $\htt=m\apm$ with probability at least $1-1/n^6$, which ensures that this process terminates with the same probability.
To prove that it terminates with the correct output, we use the second guarantee of \Cref{lemma2:alg apx}, which states that if the output $\htt\neq \bbot$, then $\htt=m\apm$ with probability at least $1-1/n^6$.
In other words, an output which is not $\bbot$ is correct with high probability, and the algorithm eventually outputs an output which is not $\bbot$.

~\\To complete the proof of \Cref{thm4:main}, we use the following theorem that bounds the complexity of the algorithm $\apx$.
\begin{theorem}\label{thm:apx complexity}
    The algorithm $\apx(V,\W,\eps)$ queries the \HO at most $T$ times, and has query measure of at most $R$, where
    \begin{align*}
        T=(k\log n)^{O(k^3)}/\eps^{2k}\;, &  & R=\frac{\gs(V)}{\W}\cdot \frac{(k\cdot\log n)^{O(k^3)}}{\eps^{2k}}\;.
\end{align*}
\end{theorem}
We prove \Cref{thm4:main} using \Cref{thm:apx complexity}.
\begin{proof}[Proof of \Cref{thm4:main} using \Cref{thm:apx complexity}]

    The algorithm $\algmain$ invokes the algorithm $\apx$ at most $k\log n$ times.
    The $i$-th execution of $\apx$ includes the following parameters $\apx(\Vin,\mu(\Vin)\cdot 2^{-i},\eps)$.
    Let $r$ be a random variable that is equal to the number of times that the algorithm $\algmain$ executes the algorithm $\apx$.
    By the guarantee of \Cref{lemma2:alg apx}, with probability at least $1-1/n^6$, the random variable $r$ satisfies $\mu(\Vin)/2^r\geq m/4$.
    To see why, note that for any $i$ for which $n^k/2^i\leq m/2$, the algorithm $\apx$ outputs $\bbot$ with probability at most $1/n^6$.
    If it does not output $\bbot$, then the algorithm $\algmain$ stops and outputs a value obtained from an execution of the algorithm $\apx$.
    This value is guaranteed to be a $\apm$ approximation for $m$ with probability at least $1-1/n^6$.
    This completes the proof of the correctness of the algorithm $\algmain$.

    ~\\ We analyze the number of queries and the query measure of the algorithm $\algmain$.
    Each call to the algorithm $\apx$ queries the \HO at most $T$ times independently of the input.
    Since the algorithm $\algmain$ makes at most $k\log n$ calls to the algorithm $\apx$ (with different inputs) the total number of times the oracle is queried is at most $k\log n\cdot T$.

    \newcommand{\QM}{M_0}
    Let $\QM \triangleq\frac{(k\cdot\log n)^{O(k^3)}}{\eps^{2k}}$. 
    Let $R_i$ denote the query measure of the $i$-th call $\apx(\Vin,\gs(\Vin)\cdot 2^{-i},\eps)$.
    The query measure of the algorithm $\algmain$ is $\max_{i} R_i$.
    From \Cref{thm:apx complexity}, we have that 
    \begin{align*}
        R_i\leq \QM\cdot \frac{\gs(\Vin)}{\gs(\Vin)\cdot 2^{-i}}=\QM\cdot 2^i\;.
    \end{align*}
    In other words, the sequence $\set{R_i}_i$ is monotonically increasing.
    Recall that we use $r$ to denote a random variable that is equal to the number of calls that the algorithm $\algmain$ makes to the algorithm $\apx$. We showed that with probability at least $1-1/n^6$, $r$ satisfies $\gs(\Vin)/2^r\geq m/4$, and therefore we can bound $\max_{i} R_i$ by $\QM \cdot 2^r$. 
    Therefore, 
    \begin{align*}
        \QM \cdot 2^r \leq 4\cdot \QM \cdot \frac{\gs(\Vin)}{m} \leq \frac{\gs(\Vin)}{m}\cdot \frac{(k\cdot\log n)^{O(k^3)}}{\eps^{2k}}\;,
    \end{align*}
    and therefore completes the proof.
\end{proof}

For the rest of this section we prove \Cref{thm:apx complexity}.
In \Cref{ssec:overview}, we start with an overview of the structure of the recursion created by the  algorithm $\apx$. In \Cref{ssec:number of q}, we analyze the number of queries made by the algorithm $\apx$.
In \Cref{ssec:query measure}, we analyze the query measure of the algorithm $\apx$. The reader may find it useful to consult with \Cref{fig:call-flow} for the analysis.

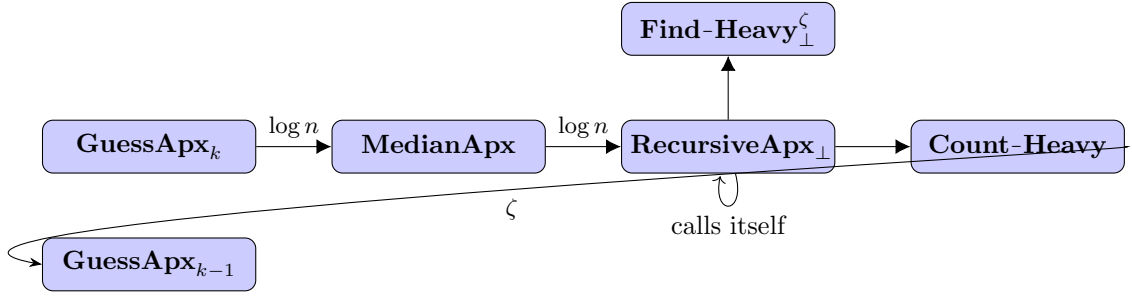
\begin{figure}[h]
    \tikzset{
block/.style={
        rectangle,
        draw,
        fill=blue!20,
        text width=7.4em,
        text centered,
        minimum height=2em,
        rounded corners
    },
line/.style={
draw,
-{Latex[length=2mm,width=2mm]}
},
zpath/.style={
to path={|- (\tikztotarget)},
draw,
-{Latex[length=2mm,width=2mm]}
},
curlypath/.style={
draw,
-{Latex[length=2mm,width=2mm]},
looseness=1,
bend right=60
},
spath/.style={
draw,
-{Latex[length=2mm,width=2mm]},
to path={let \p1 = ($(\tikztotarget) - (\tikztostart)$),
        \n2 = {veclen(\x1,\y1)},
        \n3 = {int(mod(\n2,2))} in
        .. controls ($(\tikztostart)!0.5!(\tikztotarget) + (0, \n3 * 0.5cm)$) and
        ($(\tikztostart)!0.5!(\tikztotarget) - (0, \n3 * 0.5cm)$) ..
        (\tikztotarget)}
}
}
\hspace{-2cm}
\begin{tikzpicture}[node distance=.85cm and 1cm]
\node [block] (appk) {$\apx_k$};
    \node [block, right=of appk] (c) {$\algWrapDNP$};
    \node [block, right=of c] (arec) {$\algCorezNP$};
    \node [block, above=of arec] (findheavy) {$\FindHalgB$};
    \node [block, right=of arec] (countheavy) {$\Cialg$};
    \node [block, below=of appk] (appk1) {$\apx_{k-1}$};

\path [line] (appk) -- (c) node[midway, above] {\small $\log n$};
    \path [line] (c) -- (arec) node[midway, above] {\small $\log n$};
    \path [line] (arec) edge [loop below] node {calls itself} (arec);
    \path [line] (arec) -- (findheavy);
    \path [line] (arec) -- (countheavy);
\draw [->, >=Stealth] (countheavy.east) .. controls (14,0) and (-5,-1) .. (appk1.west) node[midway, below] {\small $\zeta$};
\end{tikzpicture}

     \caption{This figure presents the call flow of the different algorithms.
        A directed edge from algorithm $A$ to algorithm $B$ indicates that algorithm $A$ makes a call to algorithm $B$. A value $x$ on such edge indicates an upper bound on the number of times that algorithm $A$ makes a call to algorithm $B$. This figure explains the structure of the recursion tree created by the algorithm $\apx$. An algorithm with out-degree zero is a leaf in the recursion tree, where only leaves query the \HO.
    }
    \label[figure]{fig:call-flow}
\end{figure}

We use the following parameters in the analysis.
Recall that
$g(k)= \cclog$,
$\zeta= k\cdot \cclog\cdot \frac{\Qd}{{\eps'}^2}\cdot \log^4 n$,
$\Qd=\Qdval$,
$\psi=4^k\cdot k^2\log(k)\cdot\log^{k+2} n$,
and define $\rho\triangleq \Qd/{\eps'}^2$.
Also, for every $\eps$ we define $\eps'\triangleq \neweps$.
\subsection{Overview of the Recursion}\label{ssec:overview}
The algorithm $\apx(V,\W,\eps)$, detailed in \Cref{lemma2:alg apx}, is a recursive algorithm that computes a $\apm$ approximation for the number of hyperedges $m_V$ in $G[V]$.
It attempts to compute a $\apm$ approximation for $m_V$ and will succeed \whp if $\W \leq m_V/2$. If $\W > m_V/2$, the algorithm may fail to compute such an approximation and outputs $\bbot$ instead.
The advantage of $\apx(V,\W,\eps)$ lies in its running time, which depends on $\W$ rather than on $m_V$.
This means that if $m_V$ is small and approximating it is too costly then the algorithm indicates that this is the case by outputting $\bbot$.

We also have several ``internal'' algorithms, which are used by the main algorithms. These are the algorithms $\algWrapDNP$, $\algCorezNP$, $\FindHalgB$, and $\Cialg$.

\paragraph*{$\algWrapDNP$.} The algorithm $\algWrapD{V,\W,\eps'}$, specified in \Cref{lemma:alg trunk}, is an amplification of the algorithm $\algCorezNP$.
It takes as input a vertex set $V$, a heaviness threshold $\W$, and an error parameter $\eps$. Let $\eps'=\neweps$.
The algorithm $\algWrapDNP$ executes $\algCorez{V,\W\cdot \eps'/\Qd}$ for $\Theta(\log n)$ times. Let $\hat{s}$ denote the median of the results, while ignoring the outputs that are $\bot$.
If the number of outputs that are $\bot$ is too large, the algorithm outputs $\bot$, and otherwise it outputs $\hat{s}$.

\paragraph*{$\arec(V,\Lambda)$.}
The algorithm $\arec(V,\Lambda)$, specified in \Cref{lemma:simp}, is a recursive algorithm that takes as input a vertex set $V$, a heaviness threshold $\Lambda$, and an error parameter $\eps'$. It outputs either $\bot$ or a value $\htt$ such that
\begin{align*}
    \Pr{\htt=m_V\apm\pm \Qd\cdot 2\Lambda/\eps'}\geq 1-1/\log n\;.
\end{align*}
If $\Lambda$ is sufficiently large then the algorithm outputs $\bot$ with probability at most $2/\log^3 n$.
The algorithm $\algCorez{V,\Lambda}$ operates as follows.
In case that $\Lambda>1$, the algorithm makes three calls in the following order.
\begin{enumerate}
    \item A call to $\FindHalgB(V,\Lambda)$, which computes a set $\vl$ or outputs $\bot$.
    \item A call to $\Cialg(V,\vl,\llow)$, where $\llow=\Lambda/(4^k\cdot \cclog)$.
    \item A call to $\algCorez{V[p]-\vl,\Lambda\cdot p^k}$.
\end{enumerate}
It uses the procedure $\FindHalgB$ to compute the set $\vl$.
It then uses the procedure $\Cialg$ to approximate $m_V(\vl)$ which is the number of hyperedges in $G[V]$ that contain at least one vertex from $\vl$.
The final call is a recursive call to $\algCorez{V[p]-\vl,\Lambda\cdot p^k}$, where $p$ is a parameter that is set to $1/2$, and the set $V[p]-\vl$ is a random subset of vertices taken from the set $V-\vl$, where every vertex is added to the new set independently with probability $p$.

The base case of the recursion is when $\Lambda\leq 1$.
In this case, the algorithm makes only the call to the algorithm $\FindHalgB(V,\Lambda)$ to compute a set $\vl$. This call can also output $\bot$ and this case $\arec$ also outputs $\bot$.
If the output of $\FindHalgB(V,\Lambda)$ is not $\bot$, then it is a set of vertices $\vl$ such that $\abs{\vl}\leq \zeta$.
Moreover, the set $\vl$ contains only $\llow$-heavy vertices and all $\Lambda$-heavy vertices \whp, where $\llow=\Lambda/(4^k\cdot \cclog)$.
Since $\Lambda\leq 1$, the set of $\Lambda$-heavy vertices contains all non isolated vertices in $G[V]$. Therefore, $m_V(\vl)$ is equal to $m_{\vl}(\vl)$ which is the number of edges in $G[\vl]$.
The number of edges in $G[\vl]$ is approximated by calling the algorithm $\algPrev(\vl,\eps')$, which was presented in \cite{DellLM22,bhattacharya2024faster} and is formally detailed in \Cref{thm:detect to count}.
This algorithm makes at most $\frac{(k\log n)^{O(k)}}{\eps'^2}$ queries to a \HO, while querying sets of vertices $U\subseteq \vl$, and therefore has query measure of at most $\gs(\vl)\leq \abs{\vl}^k\leq \zeta^k$.
The last inequality holds because the procedure $\FindHalgB$ never outputs a set of vertices of size larger than $\zeta$.
This completes the description of the algorithm $\arec(V,\Lambda)$.

\paragraph{$\FindHalgB$.} The procedure $\FindHalgB(V,\Lambda)$, which we introduced in the previous section in \Cref{algo:find heavy bb bounded}, is a procedure that takes as input a vertex set $V$ and a heaviness threshold $\Lambda$.
It outputs either $\bot$ or a set $\vl$ of all $\Lambda$-heavy vertices, and no $\llow$-light vertices with probability at least $1-1/n^4$, where $\llow=\Lambda/(4^k\cdot \cclog)$.
Later on, in \Cref{sec:find heavy}, and more specifically in \Cref{thm3:find heavy bounded}, we prove that it does not make any additional recursive calls.
We also prove that it only queries the oracle at most $O(\psi\cdot \zeta)$ times, while only querying sets of vertices $U$ for which $\gs(U)=\TO{\gs(V)/\Lambda}$.

\paragraph{$\Cialg$.} The procedure $\Cialg(V,\vl,\llow)$,
which we introduced in the previous section in \Cref{algo:count heavy}, is a procedure that takes as input a vertex set $V$ together with a subset $\vl\subseteq V$ and a heaviness threshold $\llow$.
It computes an approximation for $m_V(\vl)$.
Later on, in \Cref{sec:count heavy} we provide a formal description of this procedure, and how it is implemented.
We provide an overview of this procedure, which is also recursive, and makes additional calls to the algorithm $\apx$.

The procedure $\Cialg$ selects a vertex $v$ from $\vl$, approximates its degree $d_v$ in the hypergraph $G[V]$, denoted by $\hdv$, and then removes $v$ from the vertex set $V$. Consequently, all hyperedges containing $v$ are removed from $G[V]$. This process repeats until all vertices in $\vl$ are removed. The output of $\Cialg$ is the sum of the approximated degrees, $\sum_{v\in \vl} \hdv$. Removing $v$ from $V$ after approximating its degree ensures that the same hyperedge is not counted multiple times.

To approximate the degree of a vertex $v$ in a $k$-partite hypergraph, we use our main algorithm $\apx$ to approximate the number of hyperedges in a newly formed $(k-1)$-partite hypergraph $G_v$.
This hypergraph is derived by removing $v$ and all vertices in the same vertex set $\Vin_\ell$ that satisfies $v\in \Vin_\ell$ from $G$.
The set of hyperedges of the new graph $G_v$ is obtained as follows.
For every hyperedge $e$ in $E_V(v)$, we define a new hyperedge $e' = e \setminus \set{v}$, and add it to the set of hyperedges of the new graph $G_v$.
We then approximate the number of hyperedges in $G_v$, providing an estimate of $v$'s degree in $G$.
In other words, the call $\Cialg(V,\vl,\llow)$ makes $\abs{\vl}$ calls to the
algorithm $\apx_{k-1}$.
This algorithm is equipped with a new oracle, denoted by $\HO'$, which operates over $G_v$. We can easily simulate $\HO'$ using the original \HO: Given a set $U$ of vertices, we can simulate the query $\HO'({U})$ by
making the query $\HO(U\cup \sset{v})$.

The base case of this recursion is when $k=1$. In this case, the procedure $\Cialg$ makes no recursive calls, and directly computes the degree of a vertex $v$ using $\HO'$ via the query $\BB'(\set{v})$.
This means that at the base case, this procedure queries the \HO at most $\abs{\vl}$ times, where every such query has query measure of $O(1)$. This completes the overview of the procedure $\Cialg$.

~\\We summarize the three types of calls that access the oracle.
\begin{enumerate}
    \item A call to the algorithm $\algPrev(\vl,\eps')$.
          Such a call queries the \HO at most $\frac{(k\log n)^{O(k)}}{\eps'^2}\leq \zeta$ times,  while querying sets of vertices $U\subseteq \vl$ and therefore $\gs(U)\leq \zeta^k$.
\item A call to the procedure $\Cialg(U,\vl,\llow)$ over a $1$-partite hypergraph. Such a call queries the \HO at most $\BO{\abs{\vl}}$ times, which is at most $\zeta$. Each such query has query measure of $O(1)$.
    \item A call to the procedure $\FindHalgB(V,\Lambda)$.
          Such a call queries the \HO at most $\BO{\psi\cdot \zeta}$ times,
          where each query has query measure of $\TO{\gs(V)/\Lambda}$.
\end{enumerate}
This completes the overview of the structure of the recursion.

~\\We define the recursion tree of the algorithm $\apx(V,\W,\eps)$, which will be use to bound the number of queries and the query measure of the algorithm.
Each node in the recursion tree corresponds to a call to a specific algorithm, where the root of the tree corresponds to the call $\apx(V,\W,\eps)$.
The children of a node $x$ in the recursion tree correspond to the calls that are made by the algorithm that is associated with the node $x$.
The nodes that query the \HO are the leaves of the recursion tree.

Recall that the procedure $\Cialg$ invokes the algorithm $\apx$ to approximate the number of hyperedges in a $(k-1)$-partite hypergraph.
We use $\apx_h$ to emphasize a call to the algorithm $\apx$ on a $h$-partite hypergraph.
We define the dimension of a node $x$ in the recursion tree as follows.
For a node $x$ that is associated with a call to the algorithm $\apx_h(V,\W,\eps)$, which means that the vertex set $V$ contains $h$ vertex sets,
we say that the dimension of $x$ is $h$.
For a node $x$ that is not associated with a call to the algorithm $\apx$,
we define the dimension of $x$ as the dimension of its parent node.
\subsection{Bounding the Number of Queries}\label{ssec:number of q}
In this subsection, we prove that the algorithm $\apx(V,\W,\eps)$ makes at most $\BO{\log^{6k}n\cdot \zeta^{k}\cdot \psi}$ queries to the \HO.

We analyze the recursion tree generated by the algorithm $\apx_k$.
We can ignore the parameters in each call, as the number of queries made by any call in this tree is always bounded by $O(\psi\cdot \zeta)$ independently of the parameters given to the call.
It therefore suffices to bound the number of leaves of the recursion tree.

For that, we bound the number of leaves of dimension $h$ for $h\in[k]$.
Moreover, we bound the number of nodes of dimension $h$ that are associated with a call to the algorithm $\apx_{h}$, for $h\in[k]$.
This provides us with a recursive relation that bounds the total number of leaves in the recursion tree.
We need to prove the following properties.
\begin{enumerate}
    \item Each leaf queries the \HO at most $\BO{\psi\cdot \zeta}$ times.
\item A call $\apx_h(V,\W,\eps)$ has at most $\alpha\leq 3\log^3 n$ leaves of the same dimension $h$.
    \item A call $\apx_h(V,\W,\eps)$ has at most $\beta=\log^3 n \cdot \zeta$ descendant nodes which are associated with a call to the algorithm $\apx_{h-1}$.
\end{enumerate}
Note that the number of queries that are made indirectly by a node $x$ which is associated with a call to the algorithm $\apx_h$ is bounded by sum of the number of queries made by the leaves of dimension $h$, and the nodes of dimension $h-1$.

Once we show the above properties, we can bound the number of leaves of dimension $h$ in the recursion tree of the algorithm $\apx_k$ as follows.
For every $h\in[k]$, the number of leaves of dimension $h$ is at most the number of nodes of dimension $h$ that are associated with a call to the algorithm $\apx_h$ times the maximum number of leaves of dimension $h$ each such node can have.
The number of nodes of dimension $(k-h)$ that are associated with a call to the algorithm $\apx_{k-h}$ is at most $\beta^{h}$, and the maximum number of leaves that each such node can have is at most $\alpha$.
Therefore, the total number of leaves of dimension $(k-h)$ is at most $\alpha\cdot \beta^{h}$.
Therefore, the total number of leaves in the tree is
\begin{align*}
    \sum_{i=0}^{k-1} \alpha\cdot \beta^{i}
    \leq \alpha\cdot 2\beta^{k-1}\;,
\end{align*}
where the inequality holds because this is a sum of a geometric series with ratio $\beta\geq 2$.
Using the first property, we get that the number of queries made by the algorithm $\apx$ is at most
\begin{align*}
    \BO{2\alpha\beta^{k-1}\cdot \psi\cdot \zeta}
    =\BO{\log^{3k}n\cdot \zeta^{k}\cdot \psi}\;.
\end{align*}
To show that this is the desired bound, we need to show that
\begin{align*}
    \log^{3k}n\cdot \zeta^{k}\cdot \psi =(k\log n)^{O(k^3)}/\eps^{2k}\;.
\end{align*}
Clearly, $\log^{3k}n=(k\log n)^{O(k^2)}$, so we are left with showing that
$\psi = (k\log n)^{O(k^2)}$, and that $\zeta = (k\log n)^{O(k^2)}/\eps^{k}$.
Recall that $\psi=4^k\cdot k^2\log(k)\cdot\log^{k+2} n$, which implies that $\psi=O(4^k\cdot \log^{k+3} n)=(k\log n)^{O(k^2)}$.
To show the bound on $\zeta$, recall that $\zeta = k\cdot \cclog\cdot \frac{\Qd}{{\eps'}^2}\cdot \log^4 n$, and that $\Qd=\Qdval$, which means that $\zeta = (k\log n)^{O(k^2)}/\eps^{2k}$.
Therefore, we conclude that a call to the algorithm $\apx(V,\W,\eps)$ queries the oracle at most $(k\log n)^{O(k^3)}/\eps^{2k}$ times, which proves the first part of \Cref{thm:apx complexity}.

~\\We are left with proving the properties.
The first property follows by the definition of the leaves of the recursion tree.
We explain how to get the bounds on $\alpha$ and $\beta$.
The algorithm $\apx_h$ makes at most $\log n$ calls to $\crec$, which in turn makes at most $\log n$ calls to $\arec$.
Simplifying the analysis, we can assume each call to $\apx(V,\Lambda)$ results in at most $\log^2 n$ calls to $\algCorez{V,\Lambda}$.
$\algCorez{V,\Lambda}$ itself makes up to $\log n$ calls each to $\FindHalgB$ and $\Cialg$, and one call to $\algPrev$.
Thus, each call of $\apx_h$ results in at most $\log^2 n$ calls to $\algPrev$, and up to $\log^3 n$ calls to both $\FindHalg$ and $\Cialg$.
This proves that $\alpha\leq 3\log^3 n$ and therefore the second property holds.
Since each call to $\Cialg$ results in at most $\zeta$ calls to $\apx_{h-1}$, we have that $\beta\leq \log^3 n \cdot \zeta$, which proves that the third property holds.
Therefore, all three properties hold, and we have that the algorithm $\apx$ makes at most $(k\log n)^{O(k^3)}/\eps^{2k}$ queries to the \HO, which was the goal of this subsection.

\subsection{Bounding the Query Measure}\label{ssec:query measure}
\sloppy{In this subsection, we prove that the algorithm $\apx(V,\W,\eps)$ has query measure at most }
\begin{align*}
    \BO{\frac{\gs(V)}{\W}\cdot \frac{(k\cdot\log n)^{O(k^3)}}{\eps^{2k}}}\;.
\end{align*}

Recall that we say that an algorithm has query measure
$R$ if every set $U$ which is queried by the algorithm satisfies that $\gs(U)\leq R$.
We define the query measure of a node $x$ in the recursion tree as the maximum query measure of its direct descendants. In other words, we can ignore the number of times one algorithm calls another, and we only care about the parameters that are given to the call, which determine the query measure of the call.
As before, it suffices to bound the query measure of leaves of the recursion tree, as they are the only nodes that query the \HO.
The query measure of the call $\apx(V,\W,\eps)$ is the maximum query measure of the leaves of the recursion tree of the algorithm $\apx(V,\W,\eps)$.
The query measure of leaves associated with a call to the algorithms $\algPrev$ and $\Cialg$ is at most $\BO{\zeta^k}$, which is at most $(k\log n)^{O(k^3)}/\eps^{2k}$. This value is at most $\BO{(\gs(V)/\W) \cdot \zeta^k}$, because $(\gs(V)/\W)$ is larger than $1$.
We are left with bounding the query measure of leaves associated with a call to the procedure $\FindHalgB(U,\Lambda)$.
We will provide an analysis of the procedure $\FindHalgB$, in \Cref{sec:find heavy}, where we also prove an that the query measure of the the call
$\FindHalgB(U,\Lambda)$ is at most $\TO{\gs(U)/\Lambda}$ (\Cref{thm3:find heavy bounded})
We therefore need to show that for every pair of parameters $U$ and $\Lambda$ given to the procedure $\FindHalgB$, we have that $\gs(U)/\Lambda= \TO{\gs(V)/\W}$ with probability at least $1-1/n^6$, where $V$ and $\W$ are the parameters given to the call $\apx(V,\Lambda,\eps)$, which is associated with the call to root of the recursion tree.
We use $\xi=4^k\cdot \cclog$, where $\xi=\Lambda/\llow$, and let $\rho'=\rho \cdot (\tc)^k$.

We split the analysis into two parts.
We show that the following two properties hold for every node $s$ of dimension $h$ that is associated with a call to the algorithm $\apx_h(V,\W,\eps)$.
\begin{enumerate}
    \item The query measure of leaves of the same dimension as the node $s_0$ is at most $\BO{\frac{\gs(V)}{\W}\cdot \rho'}$ with probability at least $1-1/n^6$.
    \item Every descendant node $s'$ of dimension $h-1$ that is associated with a call to the algorithm $\apx_{h-1}(V',\W',\eps)$ satisfies that $\frac{\gs(V')}{\W'}\leq \frac{\gs(V)}{\W} \cdot \rho'\cdot \xi$ with probability at least $1-1/n^6$.
\end{enumerate}

These two properties prove that the query measure of every leaf of dimension $(k-h)$ in the recursion tree that is associated with the call $\apx_k(V,\W,\eps)$ is at most
\begin{align*}
    \BO{\frac{\gs(V)}{\W}\cdot \rho' \cdot\brak{\rho' \cdot\xi}^h}
\end{align*}
with probability at least $1-1/n^6$.
Note that the overall size of the recursion tree is at most $(k\log n)^{O(k^3)}$, as it has at most $(k\log n)^{O(k^3)}$ leaves, and height at most $k$.
Therefore, we can use a union bound on all nodes, and get that with probability at least $1-1/n^5$ the query measure is at most
\begin{align*}
    \BO{\frac{\gs(V)}{\W}\cdot \rho' \cdot\brak{\rho' \cdot\xi}^k}\;.
\end{align*}
To complete the proof which bounds the query measure of the call $\apx(V,\W,\eps)$, we need to show that
\begin{align*}
    \frac{\gs(V)}{\W}\cdot \rho' \cdot\brak{\rho' \cdot\xi}^k = \frac{\gs(V)}{\W}\cdot \frac{(k\cdot\log n)^{O(k^3)}}{\eps^{2k}}\;.
\end{align*}
This follows as $\rho'=(k\log n)^{O(k^2)}/{\eps^2}$, and $\xi=(k\log n)^{O(k^2)}$. Both hold by the definition of $\rho'$ and $\xi$:
\begin{align*}
    \rho' & =Q/{\eps'}^2 \cdot (\tc)^k = (k\log n)^{O(k)}/\eps^{2k}\;, \\
    \xi   & =4^k\cdot \cclog=(k\log n)^{O(k^2)}\;,
\end{align*}
which completes the proof that the algorithm $\apx(V,\W,\eps)$ has query measure at most $\BO{\frac{\gs(V)}{\W}\cdot \frac{(k\cdot\log n)^{O(k^3)}}{\eps^{2k}}}$.
~\\ We are left with proving the two properties.

\paragraph{Proof of property (1).} We say that a node $y$ associated with a call to the algorithm $\arec$ is
\emph{fresh}, if its parent node $y'$ is not associated with a call to the algorithm $\arec$, and instead is associated with a call to the algorithm $\crec$.
The algorithm $\apx_h(V,\W,\eps)$ makes $O(\log^2 n)$ calls to nodes of the same dimension, associated with the algorithm $\arec(V,\Lambda)$, where $\Lambda=\W/\rho$.
We only count the number of fresh nodes of dimension $h$ in the recursion tree of the call $\apx_h(V,\W,\eps)$, which are associated with a call to the algorithm $\arec(V,\Lambda)$.
We fix one fresh node $a$, and analyze the query measure of the leaves of the same dimension as $a$, in $a$'s recursion tree, which are associated with a call the procedure $\FindHalgB$.
The node $a$ has at most $O(\log n)$ such descendants
, where each such decedent is associated with a call to the procedure $\FindHalgB(U_i,\Lambda\cdot p^{ki})$,
where the vertex sets $\set{U_i}$ are random subsets of $V$ that are obtained as follows (See \Cref{alg:template recA}).
The vertex set $U_0$ is $U_0=V$. The vertex set $U_{i+1}$ is obtained by choosing each vertex from $U_i$ independently with probability $p$, and then removing some vertices from the obtained set. We can ignore this removal step, which can only decrease the size of $U_{i+1}$, and therefore only decrease the query measure of the associated call.
As we previously mentioned, we later on prove (in \Cref{thm3:find heavy bounded}) that the query measure of a call $\FindHalgB(U_i,\Lambda\cdot p^{ki})$ is at most $\TO{\gs(U_i)/\Lambda}$.
By standard Chernoff bounds, we have that the size of $U_i$ is at most $\abs{U_0}\cdot (\tc)^k$, which we proved in \Cref{claim:c2}
Therefore, the query measure of the leaf associated with the call $\FindHalgB(U_i,\Lambda\cdot p^{ki})$ is at most
$\BO{\frac{\gs(V)}{\Lambda}\cdot (\tc)^{k}}$ with probability at least $1-k/n^6$.
The term $\frac{\gs(V)}{\Lambda}\cdot (\tc)^{k}$ can also be written as $\frac{\gs(V)}{\W}\cdot \rho'$, which completes the proof of the first property.

\paragraph{Proof of property (2).} We use the same notation as in the proof of the first property.
Recall that we have a fresh node $a$ associated with a call to the algorithm $\arec(V,\Lambda)$, where $a$ is of dimension $h$.
Node $a$ makes one call to the procedure $\FindHalgB(V,\Lambda)$, which outputs a vertex set $\vl$. If this output is $\bot$, then the algorithm $\arec$ does not make additional calls.
If it is not $\bot$, then the algorithm $\arec$ makes a call to the procedure $\Cialg(V,\vl,\llow)$, where $\llow=\Lambda/(4^k\cdot \cclog)$, and another recursive call to itself with the following parameters $\arec(V[p]-\vl,\Lambda\cdot p^k)$, where $p=1/2$.
We unfold
the recursion, and think of $\arec(V,\Lambda)$ as such that has $\log n$ iterations. In iteration $i$ it makes one call to the procedure $\FindHalgB(U_i,\Lambda\cdot p^{ki})$, where we denote the output of this call by $\vl^i$.
If $\vl^i\neq\bot$, then a matching call of $\Cialg(U_i,\vl^i,\llow\cdot p^{ki})$ is made.
As previously explained, we have $U_0=V$ and the set $U_{i+1}$ is obtained by choosing each vertex from $U_i$ independently with probability $p$ and then removing some vertices from the obtained set. We ignore this removal step, which can only decrease the size of $U_{i+1}$, and therefore only decrease the query measure of the associated call.
We previously proved that the size of $U_i$ is at most $\abs{U_0}\cdot (\tc)^k$ with probability at least $1-k/n^6$.

Each call to the procedure $\Cialg(U_i,\vl^i,\llow\cdot p^{ki})$ generates $\abs{\vl^i}$ calls to the algorithm $\apx_{h-1}$, with the parameters $\apx_{h-1}(U_i',\llow\cdot p^{ki},\eps)$, where $U_i'\subseteq U_i$.
We therefore get that with probability at least $1-k/n^6$, we have that
\begin{align*}
    \frac{\gs(U_i')}{\llow\cdot p^{ki}}
    \leq \frac{\gs(V)}{\Lambda}\cdot \xi\cdot (\tc)^k
    \leq \frac{\gs(V)}{\W}\cdot \rho'\xi\;,
\end{align*}
which completes the proof of the second property, and therefore the proof of \Cref{thm:apx complexity}.
 \renewcommand{\L}{\Lambda^\prime}
\newcommand{\apmm}{(1\pm2\eps)}
\newcommand{\hdvp}{\hat{d}^{\pi}_{v}}
\newcommand{\mdvp}{d^{\pi}_{v}}
\newcommand{\hdvi}{\hat{d}^{\pi}_{v_i}}
\newcommand{\Wa}{W_{\pi\mathrm{\mhyphen heavy}}}
\newcommand{\Wb}{W_{\pi\mathrm{\mhyphen light}}}
\section{Implementing the $\Cialg$ Procedure}\label{sec:count heavy}
In this section, we explain how to implement the $\Cialg$ procedure, which is given two sets of vertices $V$ and $V_\Lambda$, a heaviness threshold $\llow$, and a precision parameter $\eps>0$.
If the set $V_\Lambda$ contains only $\llow$-heavy vertices, then \whp the algorithm outputs a $\apm$ approximation $\hmvl$ for $m_V(V_\Lambda)$.
The following lemma is the main result of this section.

\begin{restatable}[{$\Cialg$}]{theorem}{CountHeavyLemma}\label{prop:count main}
There exists an algorithm $\Cialg(V,V_{\Lambda},\llow)$ 
that takes as input two sets of vertices $V$ and $V_\Lambda$, a heaviness threshold $\llow$, and a precision parameter $\eps>0$.
    Assuming the vertex set $V_\Lambda$ contains only $\llow$-heavy vertices,
    the algorithm outputs $\hmvl$ such that
    \begin{align*}
        \Pr{\hmvl=\apm m_U(V_\Lambda)}\geq 1-\abs{V_\Lambda}/n^5\;.
    \end{align*}
The algorithm makes $\abs{\vl}$ calls to the algorithm $\algDeg(v,\llow/k,V)$, specified in \Cref{lemma:degree approx}, where $v\in \vl$.
\end{restatable}
Recall that for two sets of vertices $A,B$ we use $E_A(B)$ to denote the set of hyperedges $e$ such that $e\in G[A]$ and $e\cap B\neq\emptyset$, and we denote the number of hyperedges in $E_A(B)$ by $m_A(B)$.
Roughly speaking, to compute such an approximation we approximate the degree of every vertex $v\in\vl$ and sum the approximations of the degrees.
This simplified approach faces a few challenges.

First, the sum of the degrees of vertices in $V_\Lambda$ is not equal to the number of hyperedges that intersect $V_\Lambda$.
This is because some hyperedges might contain more than a single vertex from $V_\Lambda$, and therefore the sum of the degrees of the vertices in $V_\Lambda$ might double count some hyperedges.
To deal with this issue, we choose some arbitrary ordering $\pi$ over the vertices in $V_\Lambda$, and approximate the degree of the first vertex $v_1$ in the ordering, then remove $v_1$ from $V$, to obtain an updated vertex set $U$, and approximate the degree of the next vertex $v_2$ in the ordering over the new updated graph $G[U]$, and so on.
When we remove a vertex $v_i$ from $V$, we also remove all the hyperedges $e$ for which $v_i\in e$ from $V$. This way, we avoid double counting the hyperedges, as this ensures that every hyperedge in $E_V(V_\Lambda)$ is counted exactly once.

To approximate the degree of a vertex $v$ over a $k$-partite hypergraph $G$, we build a new $(k-1)$-partite hypergraph $G_v$, which is obtained as follows.
Let $V_\ell$ be the part in the $k$-partition of $G$'s vertices that contains $v$. 
The vertex set of $G_v$, denoted by $V'$, is defined as $V'=V\setminus V_\ell$.
The edge set $E'$ of $G_v$ is defined as follows.
For every hyperedge $e\in E_V(v)$, we define a new hyperedge $e'$ by $e'=e\setminus \set{v}$, and add $e'$ to $E'$.
Note that the degree of $v$ over $G$ is equal to the number of hyperedges in $G_v$. 

We explain how to approximate the degree $d_v$ of a vertex $v\in V_\Lambda$ within the subgraph $G$ or equivalently, the number of hyperedges in $G_v=(V',E')$.
While it might seem natural to estimate $d_v$ by approximating the number of hyperedges in $G_v$ using \Cref{thm4:main}, this might be too costly.
The reason is that the degree of $v$ might be significantly smaller than $\llow$, although initially all vertices in $V_\Lambda$ are $\llow$-heavy.
This occurs because the algorithm gradually removes the processed vertices from $U$, potentially decreasing the degrees of the remaining vertices.
Therefore, instead of approximating the degree of $v$ on $G_v$ using \Cref{thm4:main}, we approximate the degree of $v$ with an additive error of $\eps\cdot \llow/k$, by utilize our algorithm $\apx$, detailed in \Cref{lemma2:alg apx}, invoked with the parameters $\apx(V',\llow/k,\eps)$. 
The advantage of using $\apx(V',\llow/k,\eps)$ lies in its efficiency -- the number of queries it requires and its query measure are determined by $\llow/k$ and $\eps$ and are independent of the number of hyperedges in the graph it is operates on.
This comes at the price of getting an additive approximation instead of a multiplicative one which the algorithm $\algmain$ specified in \Cref{thm4:main} obtains. Recall that the algorithm $\apx(V',\llow/k,\eps)$ outputs either $\bbot$ or a value $\hdv$ such that $\hdv=m_V(v)\pm \eps\cdot \llow/k$.
In other words, the guarantee of the algorithm $\apx(V',\llow/k,\eps)$ is weaker than the guarantee of \Cref{thm4:main}.

We provide a simplified overview on why such an approximation is sufficient for our needs.
For every vertex $v\in V_\Lambda$, define a vertex set $V_{\pi(v)}$ that contains all vertices in $V$ without the vertices $u\in\vl$ for which $\pi(u)<\pi(v)$.
We use $\mdvp$ to denote the degree of $v$ over $G[V_{\pi(v)}]$, and $\hdvp$ to denote the output of the algorithm which approximates the degree of $v$ over $G[V_{\pi(v)}]$.
The approximation $\hdvp$ is such that with probability at least $1-1/n^5$ we have that $\hdvp=d_v\pm \eps\cdot \sigma_v$, where
$\sigma_v\triangleq\max\sset{\mdvp,\llow/k}$.
We partition the set of vertices $V_\Lambda$ into two disjoint sets $\Wa$ and $\Wb$.
A vertex $v\in \vl$ is in $\Wa$ if $\mdvp\geq \llow/k$ and is in $\Wb$ otherwise.
Note that for every vertex $v\in \Wa$ we have that $\mdvp\geq \llow/k$ and therefore $\sigma_v=\mdvp$, whereas for every vertex $v\in \Wb$ we have that $\sigma_v=\llow/k$.
In other words, the approximation $\hdvp$ is a $\apm$ approximation for $\mdvp$ for vertices in $\Wa$, and an approximation with an additive error of $\eps\cdot \llow/k$ for vertices in $\Wb$.

Finally, we show that $\sum_{v\in \Wa}\hdvp$ is a $\apm$ approximation for $m_V(\Wa)$, and that $\sum_{v\in \Wb}\hdvp=m_V(\Wb) \pm \eps m_V(\vl)$. We therefore get that the sum of the approximations $\hdvp$ is a $\apmm$ approximation for $m_V(V_\Lambda)$.

~\\ We also keep in mind that the set $\vl$ is obtained from a call to the algorithm $\FindHalgB(V,\Lambda)$, which means that $\abs{\vl}$ is at most $\zeta$, where recall that $\zeta=2^k\cdot 5k^2\cdot \cclog\cdot \log^7 n/{\eps'}^2$.

Recall that for every vertex $v\in V$, we use $\chi(v)$ to denote the index of the vertex set that $v$ belongs to. In other words, $v\in V_{\chi(v)}$.
The following lemma provides an algorithm that approximates the degree of a vertex.

\begin{lemma}\label{lemma:degree approx}
    There exists an algorithm $\algDeg(v,\Lambda',V)$ that takes as input a single vertex $v$, a heaviness threshold $\Lambda'$, a subset of vertices $V$, and a precision parameter $\eps>0$.
The algorithm outputs a value $\hdv$ such that
$\hdv=d_U(v) \pm \sigma_v \cdot \eps$, with probability at least $1-1/n^5$, where $\sigma_v\triangleq\max\sset{d_U(v),\Lambda'}$.
The algorithm makes a single call to the algorithm $\apx(U',\Lambda',\eps)$ specified in \Cref{lemma2:alg apx}, where $U'=U\setminus V_{\chi(v)}$, and might make a single call to the algorithm $\algWrapD{U',\Lambda',V}$ specified in \Cref{lemma:alg trunk}.
\end{lemma}

We prove \Cref{prop:count main} using \Cref{lemma:degree approx}.
\begin{proof}[Proof of \Cref{prop:count main} using \Cref{lemma:degree approx}]
    We start with a description of the algorithm $\Cialg(V,V_{\Lambda},\llow)$.

    \newcommand{\hwta}{\htt_V(\Wa)}
    \newcommand{\hwtb}{\htt_V(\Wb)}
    \begin{algorithm}[H]
        \caption{The procedure $\Cialg(V,V_{\Lambda},\llow)$ specified in \Cref{prop:count main}.}\label{alg2:count main}
        \setcounter{AlgoLine}{0}
        \KwIn{
            Two sets of vertices $V$ and $V_\Lambda$, a heaviness threshold $\llow$, and a precision parameter $\eps>0$.}
\KwOut{An approximation $\hmvl$. If $\vl$ contains only $\llow$-heavy vertices, then $\Pr{\htt_V(\vl)=\apm m_V(\vl)}\geq 1-\abs{\vl}/n^5$.}

        \medskip
        $\hmvl\gets 0,\;\;\L\gets \llow/k$.\\
\For{$v_i\in \vl$}
        {
$\hdvp\gets \algDeg(v_i,\L,V)$;\\
            $\hmvl\gets \hmvl+ \hdvp$;\\
            $V\gets V\setminus \set{v_i}$;
        }
        \Return $\hmvl$.
    \end{algorithm}

    We use the following notation.
    Let $\pi$ denote the order in which the vertices in $V_\Lambda$ are processed. Any ordering can be used, and we choose an arbitrary one.
    Recall that we use $\pi(v)$ to denote the index of the vertex $v$ in the ordering $\pi$,
and that $V_{\pi(v)}$ denotes the vertex set $U$ after removing all vertices $u$ for which $\pi(u)<\pi(v)$.
    Recall that $\mdvp$ denotes the degree of a vertex $v$ over $G[V_{\pi(v)}]$, and that $\hdvp$ denotes the output of the algorithm $\algDeg(v,\L,V_{\pi(v)})$, where $\L=\llow/k$.
    For every vertex $v\in \vl$, let $\sigma_v\triangleq\max\sset{d_v,\L}$.
    Let $\FF$ denote the event that for every vertex $v\in \vl$ we have
    \begin{align*}
        \hdvp=\mdvp\pm \sigma_v\cdot \eps\;.
    \end{align*}
The event $\FF$ occurs with probability at least $1-\abs{\vl}/n^5$, by \Cref{lemma:degree approx} and a union bound over all vertices in $V_\Lambda$.
We prove that if the event $\FF$ occurs, then  the output $\hmvl$ of the algorithm $\Cialg(V,V_{\Lambda},\llow)$ is a $\apmm$ approximation for $m_V(V_\Lambda)$.

    We partition the vertices in $V_\Lambda$ into two sets $\Wa$ and $\Wb$.
    A vertex $v\in V_\Lambda$ is in $\Wa$ if $\mdvp\geq \L$, and is in $\Wb$ otherwise.
    We emphasize that the algorithm does not know which vertex is in $\Wa$ or $\Wb$ and we only make this distinction for the analysis.
    We define the following values.
    \begin{align*}
\hwta\triangleq \sum_{v\in \Wa}\hdvp\;, &  &
        \hwtb\triangleq \sum_{v\in \Wb}\hdvp\;.
    \end{align*}
    Note that $\hmvl$ is equal to $\hwta+\hwtb$:
    \begin{align*}
        \hmvl = \sum_{v\in \vl}\hdvp =
        \sum_{v\in \Wa}\hdvp + \sum_{v\in \Wb}\hdvp
        =\hwta+\hwtb\;.
    \end{align*}
Our goal is to show that the following hold.
\begin{align}
        \hwta=m_V(\Wa)\cdot \apm\;, &  &
        \hwtb=m_V(\Wb) \pm m_V(\vl) \cdot \eps\;. \label{eq:htt}
    \end{align}
    Using \Cref{eq:htt}, we can complete the proof, as we get that
\begin{align*}
        \hmvl=\htt_{\Wa}+ \htt_{\Wb} =m_V(\vl)\cdot (1\pm 2\eps)\;.
    \end{align*}

    We are thus left with proving \Cref{eq:htt}.
    First, we prove that $m_V(\Wa)=\sum_{v\in \Wa}\mdvp$ and $m_V(\Wb)=\sum_{v\in \Wa}\mdvp$, as follows.
    The ordering $\pi$ defines a partition $\set{E_{V_{\pi(v)}}(v)}_{v\in \vl}$ of the hyperedges $E_V(V_\Lambda)$.
To prove that this is a partition, we show that every edge $e$ is in exactly one set $E_{V_{\pi(v)}}(v)$.
    Fix some edge $e\in E_V(V_\Lambda)$.
    Te prove that $e$ is in at least one set of the partition, let $v$ be the first vertex in the ordering $\pi$ such that $v\in e$. Then, $e\in E_{V_{\pi(v)}}(v)$, as every vertex $u\in e$ satisfies $\pi(u)\geq \pi(v)$ and therefore $u\in V_{\pi(v)}$.
    To prove that $e$ is in at most one set of the partition, assume towards a contradiction that
    an edge $e$ is in two sets $E_{V_{\pi(v)}}(v)$ and $E_{V_{\pi(u)}}(u)$ for some $u\neq v$, where $\pi(v)<\pi(u)$. Then, $v\notin V_{\pi(u)}$ and therefore $e$ is not in $G[V_{\pi(u)}]$, which implies that
    $e\notin E_{V_{\pi(u)}}(u)$.
    This completes the proof that $\set{E_{V_{\pi(v)}}(v)}_{v\in \vl}$ forms a partition of the hyperedges $E_V(V_\Lambda)$.

    We therefore have that
    \begin{align*}
        m_V(V_\Lambda) = \sum_{v\in\vl} \abs{E_{V_{\pi(v)}}(v)} = \sum_{v\in\vl} \mdvp\;.
    \end{align*}
    Since we did not use any property of the set $\vl$, the same claim holds for every subset $W\subseteq V_\Lambda$. That is, for any set $W\subseteq \vl$, we have that the collection $\set{E_{V_{\pi(v)}}(v)}_{v\in W}$ is a partition of the hyperedges $E_V(W)$.
    We therefore get that
    \begin{align}
        m_V(\Wa) = \sum_{v\in \Wa} \abs{E_{V_{\pi(v)}}(v)} = \sum_{v\in \Wa} \mdvp
        \label{eq1:m-is-sum-d} \\
        m_V(\Wb) = \sum_{v\in \Wb} \abs{E_{V_{\pi(v)}}(v)} = \sum_{v\in \Wb} \mdvp\;,
        \label{eq2:m-is-sum-d}
    \end{align}
    as both $\Wa$ and $\Wb$ are subsets of $V_\Lambda$.
This proves that $m_V(\Wa)=\sum_{v\in \Wa} \mdvp$ and $m_V(\Wb)=\sum_{v\in \Wb} \mdvp$.

    ~\\ Second, we show that
    \begin{align}
        \hwta=\sum_{v\in \Wa} \mdvp\cdot \apm\;, &  &
        \hwtb=(\sum_{v\in \Wb} \mdvp) \pm m_V(\vl)\cdot \eps\;. \label{eq2:htt}
    \end{align}
We conditioned on the event $\FF$, which means that for every vertex $v\in V_\Lambda$, we have $\hdvp=\mdvp\pm \sigma_v\cdot \eps$, where $\sigma_v\triangleq\max\sset{\mdvp,\L}$.
Note that by the definition of $\Wa$, for every vertex $v\in \Wa$ we have $\mdvp\geq \L$, and therefore $\sigma_v=\mdvp$.
    At the same time, for every vertex $v\in \Wb$ we have $\sigma_v=\L$.
    We therefore get
\begin{align*}
        \hwta
        = \sum_{v\in \Wa}\hdvp
        = \sum_{v\in \Wa}\brk{\mdvp\pm \sigma_v \cdot \eps}
        = \sum_{v\in \Wa}\brk{\mdvp\pm \mdvp\cdot \eps}
        = \sum_{v\in \Wa}\mdvp\apm\;,
    \end{align*}
and
    \begin{align*}
        \hwtb
        = \sum_{v\in \Wb}\hdvp
        = \sum_{v\in \Wb}\brk{\mdvp\pm \sigma_v \cdot \eps}
        = \sum_{v\in \Wb}\brk{\mdvp\pm \L\cdot \eps}  = 
        \brk{\sum_{v\in \Wb}\mdvp} \pm \L\cdot \eps\cdot \abs{\Wb}\;.
\end{align*}
    To complete the proof of \Cref{eq2:htt}, we show that $\L\cdot \abs{\Wb}\leq m_V(\vl)$.
    First note that $\Wb\subseteq \vl$. Therefore, it suffies to show that $\L\cdot \abs{\vl}\leq m_V(\vl)$.
    Since every vertex in $V_\Lambda$ is $\llow$-heavy, we have that $\sum_{v\in \vl}d_v\geq \llow\cdot \abs{\vl}$. 
    Moreover, $m_V(V_\Lambda)\geq \sum_{v\in \vl}d_v/k$ as every hyperedge contains exactly $k$ vertices. 
     We thus get that $m_V(V_\Lambda)\geq \llow\cdot \abs{\vl}/k = \L\cdot \abs{\vl}$, which proves \Cref{eq2:htt}. Together with \Cref{eq1:m-is-sum-d,eq2:m-is-sum-d}, this completes the proof of \Cref{eq:htt}.

\end{proof}

We are left with proving \Cref{lemma:degree approx}.
To this end, we provide a recursive implementation for $\algDeg$.
We start with a $k$-partite hypergraph $G[V_{\pi(v)}]$ where we want to approximate the degree of a vertex $v$, and
reduce the problem to approximating the number of hyperedges in
a new $(k-1)$-partite hypergraph.
This way, we reduce the dimension of the hypergraph by one in each recursive call.
The stopping condition for the recursion is when $k=1$ and the hypergraph is $1$-partite, i.e., every hyperedge contains exactly one vertex.
In this degenerate case the degree of a vertex is either zero or one,
and then we can exactly compute the degree of a vertex $v$ by invoking the query $\BB(\set{v})$.
When $k>1$, we use \Cref{lemma2:alg apx} to approximate the number of hyperedges in a new $(k-1)$-partite hypergraph as we formally prove in what follows.
\begin{proof}[Proof of \Cref{lemma:degree approx} using \Cref{lemma2:alg apx}]
    \newcommand{\newV}{V^\prime}
    If $G$ is a $1$-partite hypergraph, then we complete the proof as previously explained:
    We output $\BB(\set{v})$ which is either zero or one as the output of the algorithm $\algDeg(v,\Lambda,V)$.
    Clearly, the output is equal to $d_v$ so the correctness follows.

We prove the case that $k>1$.
We explain how to implement a call to the algorithm $\algDeg(v,\Lambda',V)$ using the algorithm $\apx(\newV,\Lambda',\eps)$ specified in \Cref{lemma2:alg apx}.
    For that, we define a new hypergraph $G_v$, which is a $(k-1)$-partite hypergraph, with the property that the degree of $v$ in $G[V]$ is equal to the number of hyperedges in $G_v$.
The vertex set of the new hypergraph $G_v$ is defined as $\newV\triangleq V\setminus V_{\chi(v)}$, where $\chi(v)$ is an index in $[k]$ such that  $v\in \Vin_{\chi(v)}$.
The hyperedge set of $G_v$, denoted by $\EE'$, is defined by $\EE'=\set{e\setminus\set{v}\mid e\in \EE\wedge v\in e}$.
    In words, for every hyperedge $e\in E_V(v)$, we define a new hyperedge $e'$ obtained by removing $v$ from $e$. We then add $e'$ to $\EE'$.
    Clearly, $\abs{\EE'}$ is equal to the degree of $v$ in $G[V]$.

    We use $\BB'$ to denote a new oracle over $G_v$.
    We explain how to simulate $\BB'$ over $G_v$, using a $\BB$ over $G$.
To simulate the query $\BB'(U)$, we query the set $U\cup \set{v}$ to $\BB$ and output accordingly.
We execute the algorithm $\apx(\newV,\Lambda',\eps)$ on the graph $G_v$, to
    get an approximation for the number of hyperedges in $G_v$ with additive error $\sigma_v\cdot \eps$, where $\sigma_v\triangleq\max\sset{d_v,\Lambda'}$.

    Let $\hat{s}$ denote the output of the call $\apx\brak{\newV,\Lambda',\eps}$.
    We need to prove that $\hat{s}=d_v\pm\sigma_v\cdot \eps$, with probability at least $1-1/n^5$.
    If $\hat{s}\neq \bbot$, then $\hat{s}=d_v\apm$ with probability at least $1-1/n^5$, by the correctness of \Cref{lemma2:alg apx}.
    Since $d_v\apm\subseteq d_v\pm\sigma_v\cdot \eps$, we get that $\hat{s}=d_v\pm\sigma_v\cdot \eps$ with probability at least $1-1/n^5$.

    We are left with the case that the output $\hat{s}$ is $\bbot$.
    If $\Lambda'\leq d_v/2$, then the probability that $\hat{s}=\bbot$ is at most $1/n^6$. This means that if $\Lambda'\leq d_v/2$, then $\hat{s}=d_v\pm\sigma_v\cdot \eps$ with probability at least $1-1/n^5$.
    So, if $\hat{s}$ is $\bbot$, we can assume that $\Lambda'>d_v/2$.
    In this case, we proceed by executing the algorithm $\algWrapD{\newV,\Lambda',\eps}$, detailed in \Cref{lemma:alg trunk}, with access to the detection oracle $\BB'$. Let $\hat{s}_2$ denote its output.
    We output $\hat{s}_2$ as the output of the algorithm $\algDeg(v,\Lambda',V)$.
    We prove that $\hat{s}_2=d_v\pm\sigma_v\cdot \eps$ with probability at least $1-1/n^5$.
    Since $\Lambda'\geq d_v/2$ the algorithm $\algWrapD{\newV,\Lambda',\eps}$ outputs $\hat{s}_2=d_v\pm \Lambda'\cdot \eps/8$, with probability at least $1-1/n^5$, by the correctness of \Cref{lemma:alg trunk}.
    Note that $\Lambda'\geq d_v/2$ implies that $\sigma_v\leq 2\Lambda'$,
    and therefore the output $\hat{s}_2=d_v\pm\sigma_v\cdot \eps/4$ with probability at least $1-1/n^5$. This completes the proof of the correctness of the algorithm $\algDeg(v,\Lambda',V)$.

    We analyze the running time. The algorithm $\algDeg(v,\Lambda,V)$ makes a single call to the algorithm $\apx\brak{\newV,\Lambda/2,\eps}$, and a single call to the algorithm $\algWrapD{\newV,\Lambda',\eps}$,
    which completes the proof of \Cref{lemma:degree approx}.
\end{proof}
 \newcommand{\alge}{\mathbf{Find\mhyphen Non\mhyphen Isolated}}
\newcommand{\algei}{\mathbf{Find\mhyphen Non\mhyphen Isolated_h}}
\newcommand{\lkm}{\log(2km_V)}

\newcommand{\ja}{\mathsf{J}}
\newcommand{\jc}{\mathsf{J}_2}
\newcommand{\jb}{\cl}
\newcommand{\thres}{\sigma}
\newcommand{\jk}{4k\cdot\jb}
\newcommand{\lmid}{\Lambda_{\mathrm{mid}}}
\newcommand{\ACP}{\mathbf{Discovery\mhyphen Experiment}}
\newcommand{\ZU}{\ZZ}
\newcommand{\zU}{z}
\renewcommand{\Prod}{\mathsf{Product}_k}
\section{Implementation of the $\FindHalg$ Procedure}\label{sec:find heavy}

In this section we implement the procedure $\FindHalg$, using the $\BB$.
We use the following notation. For every $\Lambda$ let $\lmid=\Lambda/\cclog$ and $\llow=\lmid/4^k$. Recall that we say that a vertex is $\Lambda$-heavy if its degree is at least $\Lambda$, and say that it is $\Lambda$-light otherwise.
Recall that we defined the measure $\gs$ on a subset $U\subseteq \Vin$ by $\gs(U)=\prod_{i=1}^k |U\cap \Vin_i|$.
The main theorem of this section is the following.
\begin{restatable}[$\FindHalg$]{theorem}{FindHeavyThm}\label{thm2:find heavy}
    There exists a randomized algorithm $\FindHalg(V,\Lambda)$ that gets as input the set of vertices $V$ and a heaviness threshold $\Lambda$.
    The algorithm outputs a set of vertices $V_\Lambda$ that contains all $\Lambda$-heavy vertices and no $\llow$-light vertices, with probability at least $1-1/n^4$.
The algorithm makes at most $\BO{k^3\log(k)\cdot\log^{k+2} n \cdot m_v/\lmid}$ queries to the $\BB$, while querying only sets of vertices $U$ for which $\gs(U)=\TO{\gs(V)/\Lambda}$.
\end{restatable}

A corollary from \Cref{thm2:find heavy}, which we use in \Cref{sec:rec}, is the following.
\begin{restatable}[$\FindHalgB$]{theorem}{ThmFindHB}\label{thm3:find heavy bounded}
    There exists a randomized algorithm $\FindHalgB(V,\Lambda)$ that gets as input a set of vertices $V$, a heaviness threshold $\Lambda$, and a threshold $\zeta$.
    The algorithm either outputs $\bot$, or a set of vertices $V_\Lambda$ containing at most $\zeta$ vertices, such that with probability at least $1-1/n^4$, the set $V_\Lambda$ contains all $\Lambda$-heavy vertices and no $\llow$-light vertices.
If $km_V/\llow\leq \zeta$, then the probability that the algorithm outputs $\bot$ is at most $1/n^4$.
The algorithm makes at most $O(\psi\cdot \zeta)$ queries to \HO, where $\psi\triangleq 4^k\cdot k^2\log(k)\cdot\log^{k+2} n$.
    The algorithm does so
    while querying only sets of vertices $U$ for which $\gs(U)= \TO{\gs(V)/\Lambda}$.
\end{restatable}
Informally, the algorithm $\FindHalgB(V,\Lambda)$ is a truncated version of the algorithm $\FindHalg(V,\Lambda)$.
It simulates the algorithm $\FindHalg(V,\Lambda)$, but stops the simulation if the number of queries made by the algorithm exceeds $\BO{\psi\cdot \zeta}$, or the computed set is of size larger than $\zeta$ and outputs $\bot$.
Otherwise, the algorithm outputs the set of vertices computed by $\FindHalg(V,\Lambda)$.
If the condition $km_V/\llow\leq \zeta$ holds, then the probability that the algorithm outputs $\bot$ is at most $1/n^4$.
We need the truncated version of the algorithm to get keep the number of queries made to the \HO small, while also only querying sets of vertices $U$ for which $\gs(U)=\TO{\gs(V)/\Lambda}$.
Following is a formal proof.

\begin{proof}[Proof of \Cref{thm3:find heavy bounded} using \Cref{thm2:find heavy}]
    Let $C$ be the constant in the statement of \Cref{thm2:find heavy}. 
    That is, the number of times the procedure $\FindHalg(V,\Lambda)$ queries the \HO is at most $(C\cdot k^2\log (k)\cdot \log^{k+2}n\cdot m_V/\lmid)$, where this number of queries can also be written as $(\psi\cdot km_V/\Lambda)$.
The algorithm for $\FindHalgB(V,\Lambda)$ works as follows.
    It simulates the procedure $\FindHalg(V,\Lambda)$ and stops the simulation if the number of times the procedure queries the oracle exceeds $C\cdot \psi\cdot \zeta$.
    If the algorithm $\FindHalg(V,\Lambda)$ does not complete within $C\cdot \psi\cdot \zeta$ queries, or it does but the computed set of heavy vertices is of size larger than $\zeta$, then the algorithm outputs $\bot$.
    Otherwise, the algorithm outputs the set of vertices computed by $\FindHalg(V,\Lambda)$.

    To complete the proof, we need to show that if $km_V/\llow\leq \zeta$, then the probability that the algorithm $\FindHalgB$ outputs $\bot$ is at most $1/n^4$.
    Let $\mathcal{E}_1$ denote the event that the procedure $\FindHalg(V,\Lambda)$ computes a set of vertices $V_\Lambda$ that contains all $\Lambda$-heavy vertices and no $\llow$-light vertices, which happens with probability at least $1-1/n^4$.
In other words, the event $\mathcal{E}_1$ denotes the event that the set $\vl$ contains only $\llow$-heavy vertices and therefore contains at most $km_V/\llow$ vertices, which is at most $\zeta$ by the assumption.
    Recall that by the guarantee of \Cref{thm2:find heavy} the number of queries made by the algorithm $\FindHalg(V,\Lambda)$ is at most $C\cdot \psi\cdot km_V$, which is at most $C\cdot \psi\cdot \zeta$, by the assumption that $km_V/\llow\leq \zeta$.
    This means that if $\mathcal{E}_1$ occurs, then the simulation of the execution of the procedure $\FindHalg(V,\Lambda)$ by the procedure $\FindHalgB(V,\Lambda)$ completes within $C\cdot \psi\cdot \zeta$ queries, and with a set of vertices $V_\Lambda$ is of size at most $km_v/\llow$ which is at most $\zeta$.
    Therefore, conditioned on the event $\mathcal{E}_1$, the procedure $\FindHalgB(V,\Lambda)$ outputs the set of vertices $V_\Lambda$ computed by $\FindHalg(V,\Lambda)$, and not $\bot$, with probability $1$.
    As the event $\mathcal{E}_1$ occurs with probability at least $1-1/n^4$, we get that the probability that the algorithm $\FindHalgB(V,\Lambda)$ outputs $\bot$ is at most $1/n^4$.
\end{proof}

It thus remains to prove \Cref{thm2:find heavy}. We first define the required notation, and then provide an overview of the algorithm specified in the theorem.
Let $\Vin=(\Vin_1,\ldots, \Vin_k)$ be a set of vertices and let $\EE$ be a set of $k$-uniform hyperedges, where each hyperedge $e\in \EE$ contains exactly one vertex from each $\Vin_i$ for $i\in[k]$.
For a vector $P=(p_1,\ldots, p_k)\subseteq [0,1]^k$, we define the weight of $P$, denoted by $w(P)$, as $\prod_{i=1}^{k}p_i$. We say that a vector $P$ is a \emph{simple} if for every $i\in[k]$, we have $p_i=2^{-j}$ for some $j\in\zrn{\ceil*{\log(n)}}$.
We also define the following set of vectors.
\begin{restatable}[$\Prod(\Lambda)$]{definition}{prodDef}\label{def:prod}
    For every integer $k$ let $\Prodd_k(\Lambda)$ be a subset of vectors in $[0,1]^k$, which contains all vectors $P\in[0,1]^k$ that satisfy $w(P)\leq 1/\Lambda$, and are simple. A vector $P=(p_1,\ldots, p_k)$ is simple if for every $i\in[k]$, we have $p_i=2^{-j}$ for some $j\in\zrn{\ceil*{\log(n)}}$.
\end{restatable}
Note that the set $\Prodd_k(\Lambda)$ contains at most $(\log(n)+2)^k\leq \log^k(4n)$ vectors, as there are at most $\log(n)+2$ possible values for each coordinate of a vector $P\in \Prodd_k(\Lambda)$.

The key to proving \Cref{thm2:find heavy} lies in sampling an induced subhypergraph of $G$ using the vectors in $\Prod(\Lambda)$.
For a vector $P\in[0,1]^k$, we define the random varaible $V[P]$ to be the set of vertices obtained by sampling each vertex $v\in \Vin_i$ with probability $p_i$ for every $i\in[k]$. We call this the $P$-discovery experiment.
We denote by $\ZZ_{V[P]}$ the set of non-isolated vertices in the induced random graph $G[V[P]]$, which are the vertices with non-zero degree in $G[V[P]]$. We also say that each vertex in $\ZZ_{V[P]}$ is discovered by $P$ or $P$-discovered.
The next lemma provides a characterization on the probability that a vertex is discovered by a vector $P$.
\begin{restatable}[Sampling Lemma]{lemma}{LemSampleA}\label{lemma:discovery}
    Let $G$ be a $k$-partite hypergraph with $n$ vertices and $m$ hyperedges.
    Then, there exists a \emph{simple} vector $P\in[0,1]^k$ such that $m\cdot w(P)\geq \cclog$, for which the random induced graph $G[V[P]]$ contains at least one hyperedge with probability at least $(\eo)^{k}$.
\end{restatable}

We use this lemma to show that for any $\Lambda$-heavy vertex
$v$ there exists a vector $P\in\Prod(\lmid)$ such that $v$ is discovered by $P$ with probability at least $(\eo)^k$.
We use $v[P]$ to denote the probability that $v$ is discovered by $P$.
To relate the lemma which operates on the entire graph to a specific vertex we can make the following reduction.
Fix some vertex $v$, and define a new subhypergraph $H$ of $G$, which contains only the hyperedges from $G$ that contain $v$.
We claim that the probability that $H[V[P]]$ contains at least one edge is equal to $v[P]$, the probability that $v$ is discovered by $P$ over $G$.
We can now get a a lower bound on $v[P]$ by applying the sampling lemma over $H$.

The main application of the sampling lemma is to find a superset of the $\llow$-heavy vertices, which contains all $\Lambda$-heavy vertices.
Using the sampling lemma, we know that by performing the $P$-discovery experiment for a vector $P\in\Prod(\lmid)$, we can find a superset of the $\Lambda$-vertices.
We later show that the probability that a $\llow$-light vertex is discovered by some vector $P\in\Prod(\lmid)$ is at most $1/4^k$.

\paragraph{Roadmap.}
The rest of this section is organized as follows.
In \Cref{ssec:applications} we provide some claims we need for the proof of \Cref{thm2:find heavy}.
In \Cref{ssec:implementation of FHB} we prove \Cref{thm2:find heavy} given \Cref{lemma:discovery}, and in
 \Cref{ssec:sampling lemma proof} we prove
 \Cref{lemma:discovery}. 
\subsection{Properties of the $P$-Discovery Experiment}\label{ssec:applications}
In this subsection we prove several properties of the discovery experiment. 
We show in \Cref{lemma:vP} that if a vertex is $\Lambda$-heavy then there exists a vector $P\in\Prod(\lmid)$ that discovers it with probability at least $(\eo)^k$.
If a vertex is $\llow$-light then for any vector $P$ with $w(P)\leq 1/\lmid$, the probability that the vertex is discovered by $P$ is at most $(\eo)^k/2$.
These two properties allow us to distinguish between $\Lambda$-heavy vertices and $\llow$-light vertices, by approximating the discovery probability of a vertex by all vectors $P\in\Prod(\lmid)$.
Also, in \Cref{claim:few light vertices}, we show that the expected number of vertices that are discovered by a vector $P\in\Prod(\lmid)$ is at most $ w(p) \cdot  2km_V\cdot\log(2k m_V) $.

For a vector $P\in[0,1]^k$ and vertex $v\in V$, we use $v[P]$ to denote the probability that $v\in \ZZ_{V[P]}$. We refer to this event as the event that $v$ is discovered by $P$.
For each vector $P$ we define the set $S_P$ to be the set of vertices for which $v[P]\geq (\eo)^k$.
The following claim, used in the proofs of both lemmas, provides a simple upper bound on $v[P]$ using Markov's inequality.

\begin{claim}\label{claim:markov discovery}
    For every vertex $v$ and every vector $P\in[0,1]^k$, we have that $v[P]\leq d_v\cdot w(P)$.
\end{claim}
\begin{proof}
    Let $X$ be a random variable equal to degree of $v$ over $G[V[P]]$.
    Clearly, we have that $\Exp{X}=d_v\cdot w(P)$, as every hyperedge $e$ is also in $V[P]$ with probability $w(P)$.
    Since $X$ is a non-negative random variable that takes integral values, we get that $\Pr{X>0}=\Pr{X\geq 1}\leq \Exp{X}$.
    We use Markov's inequality and get that $\Pr{X>0}\leq \Exp{X}$.
    Overall, we get
    \begin{align*}
        v[P]=\Pr{X>0}\leq \Exp{X}=d_v\cdot w(P)\;,
    \end{align*}
as claimed.
\end{proof}

The next lemma draws a connection between the heaviness of a vertex and its discovery probability.

\begin{lemma}\label{lemma:vP}
    Fix some $\Lambda$. Recall that $\lmid=\Lambda/\cclog$ and let $\llow=\Lambda/(4^k\cclog )$.
    \begin{enumerate}
        \item For every $\Lambda$-heavy vertex $v$ there exists a vector $P\in\Prod(\lmid)$ such that $v[P]\geq (\eo)^k$.
        \item For every $\llow$-light vertex $v$ and any vector $P$ with $w(P)\leq 1/\lmid$, we have that $v[P]\leq (\eo)^k/2$.
    \end{enumerate}
\end{lemma}
\begin{proof}
    We start with Item (1).
    Fix a $\Lambda$-heavy vertex $v$.
    Define an auxiliary hypergraph $G_v$ with the vertex set $V$, and the edge set $E(v)$, defined as the set of hyperedges $e\in E$ for which $v\in e$.
    Note that $\abs{E(v)}=d_v$, where $d_v$ is the degree of $v$, which satisfies $d_v\geq \Lambda$ since $v$ is $\Lambda$-heavy.
We apply the sampling lemma (\Cref{lemma:discovery}) on $G_v$, and get that there exists a simple vector $P$ with $w(P)\leq \lmid$
such that the probability that $G_v[V[P]]$ contains at least one hyperedge is at least $(\eo)^k$.
This follows as $d_v\cdot w(P)\geq d_v/\lmid = \cclog$, and therefore the conditions of the sampling lemma are satisfied.
    Fix a vertex set $U$. Note that $G_v[U]$ contains an hyperedge, then the degree of $v$ in $G[U]$ is at least one.
    We have that $G_v[V[P]]$ contains at least one hyperedge with probability at least $(\eo)^k$, therefore $v$ is non-isolated in $G[V[P]]$ with probability at least $(\eo)^k$. In other words, $v[P]\geq (\eo)^k$.
This completes the proof of Item (1).
    Item (2) follows from \Cref{claim:markov discovery}, as
    \begin{align*}
        v[P]\leq d_v\cdot w(P)\leq \llow/\lmid=1/4^k\leq (\eo)^k/2\;,
    \end{align*}
    where we use the fact that $d_v\leq \llow$.
\end{proof}
The following claim provides a bound on the expected number of $\llow$-light vertices that are discovered by some vector $P\in\Prod(\lmid)$.
A trivial bound we get from \Cref{claim:markov discovery} is that this expectation is at most $n/4^k$, which is too weak for our purposes.
\begin{lemma}\label{claim:few light vertices}
Fix some simple vector $P$.
    Sample a random set of vertices $V[P]$, then $\Exp{\abs{\ZZ_{V[P]}}}\leq 2km_V\cdot w(P) \cdot \lkm$.
\end{lemma}
\begin{proof}
We partition the vertices of $V$ into classes based on their degrees, as follows.
    Let $W_j=1/(2^j\cdot w(P))$ for every $j\in\zrn{\log(km_V)}$, and define
    \begin{align*}
        Q_j=\set{v\in V\mid d_v\in [W_{j},W_{j-1})},\;Q_0=\set{v\in V\mid d_v\geq W_0}.
    \end{align*}
    We claim that:
    \begin{enumerate}
        \item[(1)] For every $j$, we have $\abs{Q_j}\leq km_V/W_j$.
        \item[(2)] For every $j$ and every vertex $v\in Q_j$, we have that $v[P]\leq 2/2^{j}$.
    \end{enumerate}
    Using the above claims, we get that
    \begin{align*}
        \Exp{\abs{\ZZ_U}}
         & = \sum_{v\in V}v[P]                                  \\
         & = \sum_{j=0}^{\log(km_V)}\sum_{v\in Q_j}v[P]         \\
         & \leq \sum_{j=0}^{\log(km_V)}\sum_{v\in Q_j}2/2^j     \\
         & = 2\sum_{j=0}^{\log(km_V)}\abs{Q_j}/2^j              \\
         & \leq 2\sum_{j=0}^{\log(km_V)}(km_V/W_j)\cdot (1/2^j) \\
         & = 2\sum_{j=0}^{\log(km_V)}km_V\cdot w(P)             \\
         & = 2km_V\cdot w(P)\cdot (\log(km_V) +1)               \\
         & = 2km_V\cdot w(P)\cdot \log(2km_V) \;.
\end{align*}
    which completes the proof.
    The first inequality follows by Item $(2)$ and the second inequality follows by Item $(1)$.

We now prove Items (1) and (2).
    Item (2) follows from \Cref{claim:markov discovery}, which states that $v[P]\leq d_v\cdot w(P)$.
    Fix a vertex $v\in Q_j$, and note that $d_v\in [W_{j},W_{j-1})$, where $W_j=1/(2^j\cdot w(P))$. Therefore, we get that
    $v[P]\leq \frac{W_{j-1}}{w(P)}=2/2^{j}$.

    To prove Item (1), we use a counting argument.
    Note that the sum of the degrees of all vertices in a $k$-partite hypergraph
    is equal to $km_V$, as every edge contributes $k$ to this sum.
    Therefore, for every positive $X$ there are at most $km_V/ X$ vertices with degree at least $X$. In our case, all vertices in $Q_j$ have degree at least $W_j$ and therefore we get that $\abs{Q_j}\leq km_V/W_j$.
\end{proof}

\subsection{Implementation of the Procedure $\FindHalg$}\label{ssec:implementation of FHB}

In this subsection we explain how to implement the procedure $\FindHalg$ using a Hyperedge-Oracle.
The procedure $\FindHalg$ finds a superset of the $\llow$-heavy vertices, which contains all $\Lambda$-heavy vertices.
The procedure queries Hyperedge-Oracle at most $\TO{km/\Lambda}$ queries.
It gives Hyperedge-Oracle only sets of vertices $U$ for which $\gs(U)=\prod_{i=1}^k |U\cap \Vin_i|\leq \TO{n^k/\Lambda}$.
We explain the connection between the heaviness threshold $\Lambda$ and the parameters $\lmid$ and $\llow$.
An ideal algorithm would output a set of vertices that contains all $\Lambda$-heavy vertices and no $(\Lambda-1)$-light vertices.
However, we need to allow some slack in the heaviness threshold. Thus,
our algorithm finds a superset $\vl$ of the $\llow$-heavy vertices (thus does not contain any $\llow$-light vertices), which contains all $\Lambda$-heavy vertices. To get the desired set of vertices, we run the $P$-discovery experiment
for every vector $P\in\Prod(\lmid)$.
Each $\Lambda$-heavy vertex is discovered with probability at least $\alpha$ by some vector $P\in\Prod(\lmid)$, as shown in \Cref{lemma:vP}. Also by
\Cref{lemma:vP}, each $\llow$-light vertex is discovered with probability at most $\alpha/2$, by any vector $P\in\Prod(\lmid)$.
Intuitively, we need the weight of the vector $P$ to be between the heaviness thresholds $\Lambda$ and $\llow$.

A rough description of the 
algorithm presented in this subsection is as follows.
Run the $P$-discovery experiment for every vector $P\in\Prod(\lmid)$ and output the set of vertices that are discovered sufficiently many times.
On top of this, we need to add a truncated search, which simulates the previous algorithm, with a twist. Whenever the truncated search realizes that the current iteration is going to take too many queries, or the computed set is too large, it stops this iteration and moves on to the next one.

The following aforementioned theorem is the main result of this section.
\FindHeavyThm*

We need the following lemma to prove the theorem, which can be thought of as an algorithmic version of \Cref{lemma:vP}.

\begin{lemma}\label{lemma2:find S_P}
    There exists a randomized algorithm $\ACP(V,P)$ that gets as input the set of vertices $V$ and a vector $P\in[0,1]^k$.
The algorithm outputs a set of vertices $\hat{S}_P$, such that  with probability at least $1-1/n^5$, for every $v\in V$, we have that
    \begin{enumerate}
        \item If $v$ satisfies $v[P]\geq (\eo)^k$, then $v\in \hat{S}_P$.
        \item If $v$ satisfies $v[P]\leq (\eo)^k/2$, then $v\not\in \hat{S}_P$.
    \end{enumerate}
The algorithm queries the oracle at most $\BO{k^3\log(k)\cdot\log^2 n \cdot m\cdot w(P)}$ times, while querying only sets of vertices $U$ for which $\gs(U)= \TO{\gs(V)\cdot w(P)}$.
\end{lemma}
Recall that for a vector $P=(p_1,\ldots,p_k)$, we use $w(P)=\prod_{i=1}^k p_i$.
We explain how to use \Cref{lemma2:find S_P} to prove \Cref{thm2:find heavy}.
Let $S_H$ denote the set of vertices, for which there exists a vector $P\in\Prod(\lmid)$ such that $v[P]\geq (\eo)^k$, and let $S_L$ denote the set of vertices for which every vector $P\in\Prod(\lmid)$ satisfies $v[P]\leq (\eo)^k/2$.
The two sets are disjoint by definition.
Clearly, we can compute a set $\hat{S}$ that contains all vertices in $S_H$ and no vertices in $S_L$, by running the algorithm $\ACP(V,P)$ for every vector $P\in\Prod(\lmid)$ and taking the union of all sets $\hat{S}_P$.
To prove \Cref{thm2:find heavy}, it suffices to show that the computed set $\hat{S}$ contains all $\Lambda$-heavy vertices and no $\llow$-light vertices, which follows from \Cref{lemma:vP}.

\begin{proof}[Proof of \Cref{thm2:find heavy} using \Cref{lemma2:find S_P}]
    The algorithm $\FindHalg(V,\Lambda)$ works as follows.
    For every vector $P\in\Prod(\lmid)$ execute the algorithm $\ACP(V,P)$ and let $\hat{S}_P$ be its output.
    Compute the set $\hat{S}$, which is the union of all sets $\hat{S}_P$, and output $\hat{S}$.

    We prove that the set $\hat{S}$ contains all $\Lambda$-heavy vertices and no $\llow$-light vertices, with probability at least $1-1/n^4$, based on the properties of \Cref{lemma2:find S_P}.

    \Cref{lemma:vP} promises that for every $\Lambda$-heavy vertex $v$, there exists a vector $P\in\Prod(\lmid)$ such that $v[P]\geq (\eo)^k$.
    Fix a vertex $v$ which is $\Lambda$-heavy and a vector $P$ that satisfies $v[P]\geq (\eo)^k$.
    By \Cref{lemma2:find S_P}, we have that every vertex $v$ for which $v[P]\geq (\eo)^k$ is in the set $\hat{S}_P$, with probability at least $1-1/n^5$.
    We apply a union bound over all $P\in\Prod(\lmid)$ to get that the set $\hat{S}$ contains all $\Lambda$-heavy vertices with probability at least $1-\abs{\Prod(\lmid)}/n^5$.

    Next, we show that for every $P\in\Prod(\lmid)$, we have that $\hat{S}_P$ does not contain any $\llow$-light vertex.
    \Cref{lemma:vP} promises that for every $\llow$-light vertex $v$ and any vector $P\in\Prod(\lmid)$, we have $v[P]\leq (\eo)^k/2$.
    Fix a vertex $v$ which is $\llow$-light and a vector $P\in \Prod(\lmid)$.
By \Cref{lemma2:find S_P}, we have that every vertex $v$ for which $v[P]\leq (\eo)^k/2$ is in the set $\hat{S}_P$, with probability at most $1/n^5$.
    We apply a union bound over all $P\in\Prod(\lmid)$ to get that the set $\hat{S}$ contains a $\llow$-light vertex with probability at most $\abs{\Prod(\lmid)}/n^5$.

    Overall, we get that the set $\hat{S}$ contains all $\Lambda$-heavy vertices and no $\llow$-light vertices, with probability at least $1-(2\log n)^k/n^5\geq 1-1/n^4$, where $\abs{\Prod(\lmid)}\leq (\log (4n))^k$, which follows from its definition.

    ~\\To bound the running time of the algorithm $\FindHalg$, note that the described algorithm makes $\abs{\Prod(\lmid)}$ calls to the algorithm $\ACP(V,P)$, and that the number of queries to $\BB$ made by each such call is at most
    \begin{align*}
        \BO{k^3\log(k)\cdot\log^2 n \cdot m_v\cdot w(P)}
        = \BO{k^3\log(k)\cdot\log^2 n \cdot m_v/\lmid}\;.
    \end{align*}
Therefore, the total number of queries to $\BB$ is at most
    \begin{align*}
        \BO{k^3\log(k)\cdot\log^{k+2} n \cdot m_V/\lmid}\;.
    \end{align*}
The sets of vertices that are queried by the algorithm $\ACP(V,P)$ are sets $U$ for which $\gs(U)=\TO{\gs(V)\cdot w(P)}=\TO{\gs(V)/\lmid}$, where the last equality follows because $P\in \Prod(\lmid)$, and therefore $w(P)\leq {1/\lmid}$.
This gives us the desired bound on the running time of the algorithm $\FindHalg$.
\end{proof}

To prove \Cref{lemma2:find S_P}, we use the following deterministic proposition. 
\begin{proposition}\label{prop2:finding in edge}
    For a set of vertices $U$, let $S$ denote the set of non-isolated vertices over $G[U]$.
    There exists a deterministic algorithm $\alge(U,\thres)$, that takes as input a set of vertices $U$ and a threshold $\thres\geq 2$.
    If $\abs{S}\geq \thres$ then the algorithm outputs $\bot$. Otherwise, the algorithm outputs the set $S$.
    The algorithm makes at most $\BO{k\cdot \thres \cdot \log n}$ queries to the $\BB$, while only querying sets of vertices $U'$ that are contained in $U$.
\end{proposition}

We explain how to use \Cref{prop2:finding in edge} to prove \Cref{lemma2:find S_P}.
The core of the algorithm $\ACP(V,P)$ specified in \Cref{lemma2:find S_P} is to
run the $P$-discovery experiment, which includes sampling a subset of vertices $U\gets V[P]$ and computing the set of vertices $\ZZ_U$, which is the set of non-isolated vertices in the induced subgraph $G[U]$.
To compute this set, we use the algorithm $\alge$.

We add two additional steps on top of the $P$-discovery experiment, which are amplification and doubling.
We describe them in the following.
We fix some threshold $\sigma$ and repeat the discovery experiment $r=\Omega(\log n)$ times.
If the number of iterations in which the algorithm outputs $\bot$ is too large, then we increase the threshold $\sigma$ by a factor of $2$ and repeat the process.
Otherwise, we output the set of vertices that are discovered in at least $3r\cdot q/4$ iterations, where $q=(\eo)^k$.

We use doubling to find the correct threshold for the following reason.
If the threshold is too low, then the algorithm might output $\bot$ even if
all vertices $v$ for which $v[P]\geq (\eo)^k$ are discovered. On the other hand, if the threshold is too high, then the query-measure of the algorithm might be too high.

We note that the expected number of discovered vertices is at most $km_V\cdot w(P)\cdot \lkm$, by \Cref{claim:few light vertices}.
We also prove that \whp the final threshold we use is indeed bounded by
$O(km_V\cdot w(P)\cdot \lkm)$.
The formal proof is as follows.
\begin{proof}[Proof of \Cref{lemma2:find S_P} using \Cref{prop2:finding in edge}]
    \sloppy{We provide an implementation for the algorithm
$\ACP$ specified in \Cref{lemma2:find S_P}, using the algorithm $\alge$ specified in \Cref{prop2:finding in edge}.}
    \newcommand{\aca}{q}
    \newcommand{\hS}{\hat{S}}
    \newcommand{\hSP}{\hat{S}_P}
    \newcommand{\hcv}{\hat{c}_v}
    \newcommand{\uij}{U_{i,j}}
    \newcommand{\sij}{S_{i,j}}
    \newcommand{\hsij}{\hat{S}_{i,j}}

    \begin{algorithm}[H]
        \caption{Implementation of the algorithm $\ACP(V,P)$, specified in \Cref{lemma2:find S_P}.}\label{alg:imp ACP}
        \setcounter{AlgoLine}{0}
        \KwIn{A set of vertices $V$ and a vector $P\in[0,1]^k$.}
        \KwOut{A set of vertices $\hat{R}$, containing all vertices $v$ for which $v[P]\geq (\eo)^k$ and no vertex $v$ for which $v[P]\leq (\eo)^k/2$, with probability at least $1-1/n^5$.}

        \medskip
        Let $\aca\gets (\eo)^k$, and $r\gets 600\log n/\aca$;\\
        \smallskip
        \For{$j=0$ to $j=\log n$}{
            \smallskip
            Set $\hcv\gets 0$ for every vertex $v$, $b_j\gets 0$, and $\sigma=2^j$;\\
            \smallskip
            \For{$i=1$ to $i=r$}{
                \smallskip
                Sample a set of vertices $\uij\gets V[P]$\Comment{Let $\sij$ denote the set of non-isolated vertices in $G[\uij]$.}
                \lIf{$\abs{\uij}\geq \Exp{\uij}\cdot (6\log n)^k$}{
                    Continue to the next iteration
}
Define $\hsij\gets \alge(\uij,\sigma\cdot 100/\aca)$;\\
                \lIf{$\hsij=\bot$}{
                    $b_j\gets b_j + 1$
                }
                \For{every vertex $v\in \hsij$}{
                    $\hcv\gets \hcv + 1$.
                }
            }
            \If{$b_j\leq  r\cdot \aca / 10$}{
                \smallskip
                $\hat{R}_j\gets \set{v\mid \hcv\geq 3r\aca/4}$\Comment{
                    The set $\hat{R}_j$ to be the set of vertices $v$ for which $\hcv\geq 3r\aca/4$.}
                \smallskip
                \Return{$\hat{R}_j$}.
            }
        }
    \end{algorithm}
    We use the following notation.
    For every $j\in\zrn{\log n}$ and $i\in[r]$ we use $\uij$ to denote the set of vertices sampled in the $i$-th iteration with threshold $\sigma=2^j$.
    We also use $\sij$ to denote the set of non-isolated vertices over $G[\uij]$,
    and $\hsij$ to denote the output of the call $\alge(\uij,2^j\cdot 100/\aca)$.
    For fixed $j$ the algorithm maintains a counter $\hcv$ for every vertex $v$, which counts the number of times the vertex $v$ is contained in the set $\hsij$.
    The algorithm also maintains a counter $b_j$ that counts the number of times the algorithm outputs $\bot$, with threshold $\sigma=2^j$.
    We also define a set $\hSP$ to be the set of vertices $v$ for which $\hcv\geq 3r\aca/4$,
    For the analysis, we keep a counter $c_v$ for every vertex that counts the number of times the vertex $v$ is in the set $\sij$.
    We emphasize that we do not have access to the value of $c_v$ or to the set $\sij$, and we only use them for the analysis.
    We do have access to the counters $\hcv$ and the sets $\hsij$.

    ~\\We make the following observation.
    If $\hsij$ is not $\bot$, then the sets $\hsij$ and $\sij$ are equal.
This follows by the guarantee of the algorithm $\alge$, provided in \Cref{prop2:finding in edge}, which is a deterministic algorithm.
    Therefore, for every $v$ and fixed $j$ we have $c_v-b\leq \hcv\leq c_v$.
    We use this fact together with a concentration bound on $c_v$ to prove a concentration bound on $\hcv$.
    We prove the correctness of the algorithm by proving the following two properties.
    \begin{enumerate}
        \item For any $j\in\zrn{\log n}$, if $b_j\leq r\aca/10$, then the set $\hSP$ contains all vertices $v$ with $v[P]\geq \aca$ and no vertex $v$ with $v[P]\leq \aca/2$, with probability at least $1-1/n^5$.
\item For $j$ such that $2^j\geq 2km_V\cdot w(P)\cdot \lkm$, the probability that $b_j> r\aca/10$ is at most $1/n^5$.
\end{enumerate}
We start with proving Property (1). We fix some iteration $j$ for which $\hSP\neq \bot$.
    We show that
    \begin{enumerate}
        \item For every vertex $v$ with $v[P]\geq \aca$, we have that $c_v\geq  \aca \cdot r\cdot 5/6$, with probability at least $1-1/n^6$.
        \item For every vertex $v$ with $v[P]\geq \aca/2$, we have that $c_v<  \aca \cdot r\cdot 3/4$, with probability at least $1-1/n^6$.
    \end{enumerate}
    Note that $c_v$ is a binomial random variable with $r$ trials and success probability $v[P]$.
    We now apply Chernoff's inequality.
For a vertex $v$ with $v[P]\geq q$, we get
    \begin{align*}
        \Pr{c_v< \aca \cdot r\cdot 5/6}
        =\Pr{c_v< \aca \cdot r\cdot (1-1/6)}\leq \exp(-\aca \cdot r/100)\leq 1/n^6\;.
    \end{align*}
    For $v$ with $v[P]\leq \aca/2$, we get
    \begin{align*}
        \Pr{c_v\geq   \aca \cdot r\cdot 3/4}
        =\Pr{c_v\geq  \aca \cdot r/2\cdot (1+1/2)}\leq \exp(-\aca \cdot r/10)\leq 1/n^6\;.
    \end{align*}
    We apply a union bound over all $n$ vertices, and get that for every vertex $v$ with $v[P]\geq \aca$ we have that $c_v\geq 5r/6$, and for every vertex $v$ with $v[P]\geq \aca/2$ we have that $c_v\leq 3r/4$, with probability at least $1-1/n^5$.
    Using the assumption that $b_j\leq r\aca/10$, we get that if $c_v\geq r\cdot 5/6$ then $\hcv\geq 3r\aca/4$, and if $c_v\leq r\cdot 3/4$ then $\hcv< 3r\aca/4$.
    Therefore, we conclude that if $b_j\leq r\aca/10$, then the set $\hSP$ contains all vertices $v$ with $v[P]\geq \aca$ and no vertex $v$ with $v[P]\leq \aca/2$, with probability at least $1-1/n^5$.

    ~\\We now prove Property (2).
    That is, we prove that for $2^j\geq 2km_V\cdot w(P)$, the probability that $b>r\aca/10$ is at most $1/n^5$.
Let $X$ be a random variable that is equal to $\ZZ_{V[P]}$.
    By \Cref{claim:few light vertices}, we have that $\Exp{X}\leq 2km_V\cdot w(P)\cdot \lkm$. By Markov's inequality, we have that
    \begin{align*}
        \Pr{X\geq \Exp{X}\cdot 100/\aca}\leq \aca/100\;.
    \end{align*}
    Therefore, if $2^j\geq km_V\cdot w(P)\cdot \lkm$, then $b_j$ is a binomial random variable with $r$ trials and success probability at most $\aca/100$, and therefore its expectation is at most $r\cdot \aca/100$, which we denote by $\beta$.
    We apply Chernoff's inequality on $b_j$, under the assumption that $\beta\geq \Exp{b_j}$, and get
    \begin{align*}
        \Pr{b_j\geq r\cdot \aca/10}=
        \Pr{b_j\geq 10\beta}\leq 2^{-10\beta}\leq 1/n^5\;.
\end{align*}
    This completes the proof of the two properties.

    ~\\It remains to analyze the running time of the algorithm.
    The number of calls to the algorithm $\alge$ is at most $r\cdot \log n$.
    In each such call the threshold is at most $km_V\cdot w(P)$ by \Cref{claim:few light vertices}.
    This means that each such call makes at most $\BO{k^2\cdot m_V\cdot w(P)\cdot \lkm\cdot \log n}$ queries to the oracle $\BB$, by \Cref{prop2:finding in edge}.
    Therefore, the number of queries to the oracle $\BB$ is at most
    \begin{align*}
        \BO{k^2\cdot m_V\cdot w(P)\cdot \lkm\cdot \log n\cdot r}
        =\BO{k^3\log(k)\cdot m\cdot w(P)\cdot \log^2 n}\;,
    \end{align*}
    where we use the fact that $r=O(\log n)$, and that $m_V\leq n^k$.
We also note that the algorithm only queries sets of vertices $U$ obtained by sampling $V[P]$. Using \Cref{claim:c2}, we have that $\gs(U)\leq \BO{\gs(V)\cdot w(P) \cdot (6\log n)^k}$, with probability at least $1-1/n^6$.
    Therefore, the algorithm only queries sets of vertices $U$ for which $\gs(U)\leq \TO{\gs(V)\cdot w(P)}$.
    This completes the proof of the lemma.
\end{proof}

In what follows we prove \Cref{prop2:finding in edge}. 
Given a set $U$, and a threshold $\sigma$ we would like to find the subset of vertices $S$ of non-isolated vertices over $G[U]$, assuming $\abs{S}\leq \sigma$, where otherwise we output $\bot$.
We focus on finding $S\cap \Vin_1$, that is the set of vertices in $S$ that are contained in $\Vin_1$, where we compute $S\cap \Vin_i$ seperately for $i\in[k]$.
Informally, we look at a single part $U_1=\Vin_1\cap U$, and partition it into $\thres$ disjoint sets $L^1,\ldots,L^\tau$.
We then query the \HO on the sets $R_i=(U\setminus U_1)\cup L^i$ for $i\in[\thres]$. In words, $R_i$ is the entire set $U$, where the only set of vertices from $\Vin_1$ that are in $R_i$ are those in $L^i$.
If there are less than $\thres$ vertices in $S$, then there exists at least one set $L^i$ for which $\BB(R_i)=0$, and therefore we can discard all vertices that are in $L^i$, as none of them are in $S$.
Following is the formal proof.
\begin{proof}[Proof of \Cref{prop2:finding in edge}]
We use the following notation. Let $U_h=U\cap \Vin_h$ for $h\in[k]$.
We explain how to either find the set of vertices $S\cap U_1$, or find a witness for $\abs{S\cap U_1}\geq \thres$, in which case we output $\bot$.
    The same process can be applied to find $S\cap U_h$ for other values $h\in[k]$.
We initialize a set $L\gets U_1$, and repeat the following process:
    \begin{enumerate}
        \item Partition the vertices of $L$ into $\thres$ disjoint sets $L^1,\ldots, L^{\thres}$, as balanced as possible, such that each set has size in $\set{\floor{\abs{L}/\thres},\ceil{\abs{L}/\thres}}$.
              For each $i\in[\thres]$, define the set $R_i=(U\setminus U_1)\cup L^i$.
        \item Query the \HO on every set $R_i$, i.e., invoke $\BB(R_i)$ for each $i\in[\thres]$.
        \item Remove from $L$ all vertices that are not contained in a set $L^i$ for which $\BB(R_i)=1$. If no such set exists, we exit the loop, and otherwise go back to step 1.
    \end{enumerate}
    We first analyze the case where $\abs{S\cap U_1}< \thres$.
    In this case, as long as there are at least $\thres$ non-empty sets in $L$, there exists at least one set $L^i$ which does not contain a vertex from $S$, and therefore $\BB(R_i)=0$.
    Therefore, the process removes at least a $1/(\thres + 1)$ fraction of the vertices in $L$ in each iteration, which means that after at most $2\thres\log n$ iterations, there less than $\thres$ vertices left in $L$.
    We then make $\abs{L}$ trivial queries:
    we invoke the query $\BB((U\setminus U_1)\cup\set{v})$, for $v\in L$, which is the set of all vertices not in $U_1$ together with a single vertex $v$ from $L$.
    Define $S_1$ to be the set of vertices $v$ for which $\BB\brak{(U\setminus U_1)\cup\set{v}}=1$. We then output the set $S_1$.

    We now analyze the case where $\abs{S\cap U_1}\geq \thres$.
    In this case, There exists an iteration $j$ in which no set is removed from $L$,
    and there are at least $\thres$ non-empty sets in $L$ at this iteration.
    The sets $R_1,\ldots, R^{\thres}$ are the witnesses that $\abs{S\cap U_1}\geq \thres$, and therefore the algorithm outputs $\bot$.
\end{proof}

\subsection{Proof of \Cref{lemma:discovery}}\label{ssec:sampling lemma proof}
In this subsection we prove \Cref{lemma:discovery}.
\LemSampleA* \renewcommand{\Prod}{\mathsf{Product}_k}
\newcommand{\EC}{\mathcal{E}}
We work with the following notation.
Recall that we consider the following set of vectors.
\prodDef* 

For every vector $P\subseteq [0,1]^k$, we define the random variable $\XC(V[P])$ as an indicator random variable for the event that the random induced graph
$G[V[P]]$ contains at least one hyperedge.
Therefore, a different formulation of \Cref{lemma:discovery} is as follows.

\begin{restatable}[Equivalent Sampling Lemma]{lemma}{LemSample}\label{lemma:sampling tag}
    If $G$ has at least $m=m'\cdot \cclog$ edges then there exists a simple vector $P\in \Prod(m')$ such that $\XC(V[P])=1$ with probability at least $(\eo)^k$.
\end{restatable}

We prove \Cref{lemma:sampling tag} by induction on $k$.
The base case is when $k=2$, which we prove in the following.
~\\\textbf{Proof of the Base Case.}
\renewcommand{\L}{m^\prime}
Let $K_0=\log (\L)+1$, and let $W_i\triangleq \L/2^i$.
For $i\in\zrn{K_0}$, let $P_i=(\tfrac{2^i}{\L},\tfrac{1}{2^{i}})$, and let
$Z_i=e(G[P_i])$, denote a random variable that counts the number of edges in the random induced subgraph $G[P_i]$.
We partition $V_2$ into classes as follows. For $i\in\zrn{K_0}$, let
\begin{align*}
    Q_i=\set{u\in V_2\mid \deg(u)\in\left[W_i,2\cdot W_i\right)}\;, &  & Q_0=\set{u\in V_2\mid \deg(u)\geq W_0}\;.
\end{align*}
In words, $u\in Q_i$ if and only if its degree is at least $W_i$ and less than $2W_i$. The set $Q_0$ consist of all vertices $u\in V_2$, with degree at least $W_0=\L$.
We prove the following two claims.
\begin{enumerate}
    \item There exists $j\in\set{0,1,\ldots, K_0}$ such that $\abs{Q_j}\geq 2^j$.
    \item For such $j$, we have $\Pr{Z_j>0}\geq (\eo)^2$.
\end{enumerate}
Proving both of these claims completes the proof of the proposition.

We start with a proof of the second claim.
Fix $j\in\set{0,1,\ldots, K_0}$ where $\abs{Q_j}\geq 2^j$.
For any non-empty subset $S\subseteq Q_j$, let $\Em{1}(S)$ denote the event that $\set{Q_j\cap V_2[2^{-j}]=S}$.
That is, $\Em{1}(S)$ denote the event the set of vertices that were sampled from $V_2\cap Q_j$ is exactly the set $S$.
Let $N(S)$ denote the neighbors of the set $S$ (in $V_1$)
Let $\Em{2}(S)$ denote the event that $\set{V_1[1/W_j]\cap N(S)\nemp}$.
That is, $\Em{2}(S)$ denote the event that the set of vertices that were sampled from $V_1\cap N(S)$ is not empty.
Next, we show that
\begin{align}
    v[P_j]=\Pr{Z_j>0}=\sum_{S: \emptyset\subsetneq S\subseteq Q_j} \Pr{\Em{1}(S)\cap \Em{2}(S)} \geq (\eo)^2 \label{eq3:v Pk}\;.
\end{align}
For every fixed $S\subseteq Q_j$, which is not empty, the events $\Em{1}(S)$ and $\Em{2}(S)$ are independent, since $\Em{1}(S)$ addresses sampling vertices from $V_2$, while $\Em{2}(S)$ addresses sampling vertices from $V_1$, where the two samples are independent of each other.
Also note that $\Em{1}(S)\cap \Em{2}(S)$ is contained in the event that $Z_j$ is positive, and since the events $\set{\Em{1}(S)}_{S\subseteq Q_j}$ are pairwise disjoint, so are the events $\set{\Em{1}(S)\cap \Em{2}(S)}_{S\subseteq Q_j}$.
Therefore, the event that $Z_j$ is positive, contains the union of the following pairwise disjoint events $\sset{\bigcup_{S: \emptyset\subsetneq S\subseteq Q_j}\Em{1}(S)\cap\Em{2}(S)}$.
We get
\begin{align*}
    \Pr{Z_j>0}
     & \geq \Pr{\bigcup_{S: \emptyset\subsetneq S\subseteq Q_j}\Em{1}(S)\cap\Em{2}(S)} \\
     & = \sum_{S: \emptyset\subsetneq S\subseteq Q_j} \Pr{\Em{1}(S)\cap \Em{2}(S)}     \\
     & = \sum_{S: \emptyset\subsetneq S\subseteq Q_j} \Pr{\Em{1}(S)}\Pr{\Em{2}(S)}\;.
\end{align*}
To complete the proof, we need to show that
\begin{enumerate}
    \item $\sum_{S: \emptyset\subsetneq S\subseteq Q_j}\Pr{\Em{1}(S)}\geq \eo$
    \item for every $S\subseteq Q_j$, which is not empty, we have $\Pr{\Em{2}(S)}\geq \eo$.
\end{enumerate}
The first claim follows as
\begin{align*}
    \sum_{S: \emptyset\subsetneq S\subseteq Q_j}\Pr{\Em{1}(S)}=\Pr{Q_j[2^{-j}]\nemp}=1-(1-2^{-j})^{\abs{Q_j}}\geq \eo\;,
\end{align*}
where the inequalities hold for the following reasons. The first two equalities follows from definition, and the last inequality follows from the assumption that $\abs{Q_j}\geq 2^j$.

The second claim follows as every non-empty subset $S\subseteq Q_j$ satisfies $\abs{N(S)}\geq W_j$;
To see this, fix some vertex $u\in S$, and note that $N(u)\subseteq N(S)$, and $\deg(u)=\abs{N(u)}\geq W_j$ because $u\in Q_j$.
Therefore, for any such $S$, we have
\begin{align*}
    \Pr{\Em{2}(S)}=\Pr{N(S)[1/W_j]\nemp}=1-(1-1/W_j)^{\abs{N(S)}}\geq 1-(1-1/W_j)^{W_j}\geq \eo\;.
\end{align*}
This completes the proof of \Cref{eq3:v Pk}.

\medskip
To complete the proof of the induction base, it remains to show that the first claim indeed holds, i.e., that there exists $j\in\set{0,1,\ldots, K_0}$ such that $\abs{Q_j}\geq 2^j$.
Assume towards a contradiction that this is not the case, that is we have $\abs{E}\geq m$, the set $Q_0$ is empty, and for every $j\in\zrnone{K_0}$, we have $\abs{Q_j}<2^j$. Then,
\begin{align*}
    \abs{E}
     & =   \sum_{u\in[V_2]}\deg(u)                                                     \\
     & <   \sum_{u\in[V_2]}\sum_{1\leq k\leq K_0}\chi_{\set{u\in Q_j}}\cdot 2\cdot W_j \\
     & =   \sum_{1\leq k\leq K_0}\sum_{u\in[V_2]}\chi_{\set{u\in Q_j}}\cdot 2\cdot W_j \\
     & =   \sum_{1\leq k\leq K_0}\abs{Q_j}\cdot 2\cdot W_j
\end{align*}
Because we assumed towards a contradiction that for every $j\in\set{0,1,\ldots, K_0}$ we have $\abs{Q_j}<2^j$, we get
\begin{align*}
    \sum_{1\leq j\leq K_0}\abs{Q_j}\cdot 2\cdot W_j<
    \sum_{1\leq j\leq K_0}2^j\cdot 2\cdot W_j=
    \sum_{1\leq j\leq K_0}2\L =2K_0\cdot \L \;.
\end{align*}
This is a contradiction, as
\begin{align*}
    2K_0\cdot \L
    = 2\L(\log(\L)+1)
    = 2\cdot \frac{m}{6\log n}\cdot (\log(\frac{m}{6\log n})+1)
    \leq 2\cdot \frac{m}{6\log n}\cdot (2\log n)\leq 2m/3<m\;.
\end{align*}
The first equality follows from the definition of $K_0$, the next one by the fact that $\L=m/(6\log n)$.
This completes the proof of the induction base.
\renewcommand{\EE}[1]{\mathcal{E}(#1)}
\newcommand{\EEs}[1]{\mathcal{E}^{\ast}(#1)}
~\\Next, we prove the induction step.
We use the following notation.
Let $\beta=k\log 4n$, and $B=m'\cdot \beta^{(k-1)^2}$.
Recall that $m'=\frac{m}{\beta^{k^2}}$. We also define $K_0=\ceil*{\log(n)}$, $W_j=B/2^j$ for $j\in\zrn{K_0}$, and finally $\gamma\triangleq (\log 4n)^{k-1}$, where $\gamma\geq \abs{\Prod(B)}$.

~\\\textbf{Proof of the Induction Step of \Cref{lemma:sampling tag}.}
For every vertex $v$ define $\EE{v}=\set{e\in\EC\mid v\in e}$, and let $d_v=\abs{\EE{v}}$, which we call the degree of $v$.
We partition the vertices of $V_k$ into classes by their degrees.
For $j\in\zrnone{K_0}$, define the sets $Q_j$ as follows.
\begin{align*}
    Q_j=\set{u\in V_h\mid d_u\in[W_j,2W_j)}\;, &  & Q_0=\set{u\in V_h\mid d_u\geq W_0}\;.
\end{align*}
In words, $u\in Q_j$ if and only if its degree is at least $W_j$ and less than $2W_j$. The set $Q_0$ consist of all vertices $u\in V_2$, with degree at least $W_0=m'$.

We claim that at least one of the following must hold.
The set $Q_0$ is not empty, or there exists a value $j\in\zrnone{K_0}$ such that $\abs{Q_j}\geq 2^j\cdot \gamma$.
Assuming this is true, we can complete the proof.

We first prove this assuming that the set $Q_0$ is not empty, and let $u$ denote a vertex in $Q_0$.
We define a new $(k-1)$-partite graph $H_u$ as follows.
The vertex set of $H_u$ is $V(H_u)=V\setminus V_{k}$, and the edge set of $H_u$ is $E(H_U)=\set{e\setminus\set{u}\mid e\in\EC(u)}$.
Note that $H_u$ has at least $d_u$ edges, where $d_u$ is the degree of $u$ in $G$.
By the induction hypothesis there exists a simple vector $P_u\in \Prodd_{k-1}(d_u')$, where $d_u'=\frac{d_u}{\beta^{(k-1)^2}}$, such that $\Pr{\XC(V(H_u)[P_u])=1}\geq (\eo)^{k-1}$.
To complete the proof, we extend $P_u$ to a vector $Z_u\in \Prodd_k(m')$
and show that $\Pr{\XC(V[Z_u])=1}\geq (\eo)^{k}$.
The vector $Z_u$ is defined using $P_u$ as follows.
If $P_u=(p_1,\ldots, p_{k-1})$, then $Z_u=(p_1,\ldots, p_{k-1},p_k)$, where $p_k=1$.
We need to show that $Z_u\in \Prodd_k(m')$, which follows as
\begin{align*}
    w(Z_u)=w(P_u)
    = \frac{1}{d_u'}
    = \frac{\beta^{(k-1)^2}}{d_u}
    \leq \frac{\beta^{(k-1)^2}}{B}
    = \frac{1}{m'}\;.
\end{align*}
Intuitively, the vector $Z_u$ always samples the vertex $u$.
It then operates as $P_u$ on the remaining vertices.
To see that, note that for every realization $S\gets V(H_u)[P_u]$, the subgraph $H_u[S]$ contains an edge if and only if the subgraph $G_u[S\cup\set{u}]$ contains an edge.
Therefore, we have that
\begin{align*}
    \Pr{\XC(V[Z_u])=1}
    =\Pr{\XC(V(H_u)[P_u])=1}\geq (\eo)^{k-1}\;,
\end{align*}
which completes the proof of the case in which $Q_0$ is not empty.

~\\The other cases are a bit more complicated.
The difficulty arises for the following reason. If the ``good'' set $Q_j$, contains more then one vertex, then it might be the case that each vector has different vector that discovers it. Therefore there is no clear way to define $P_u$ as we did before. Moreover, choosing one such vector and extending it to a $k$-dimensional vector requires that the last coordinate of the extended vector is $1/2^j$, which means we choose a sampling vector $P_u$ for some vertex $u$ in $Q_j$, but might sample the vertex $u'$ from $Q_j$ which has a completely different vector that discovers it.
For that reason we use a stronger assumption, which is that the set $Q_j$ contains at least $2^j\cdot \gamma$ vertices. Since the number of possible vectors is at most $\gamma$, we can use the pigeonhole principle to claim that there exists a vector $P$ that discovers at least $2^j$ vertices in $Q_j$.

~\\\textbf{Completing the proof under the assumption.}

For every vertex $u\in Q_j$ we define a new $(k-1)$-partite graph $H_u$ as follows.
The vertex set of $H_u$ is $V\setminus V_{k}$, and the edge set of $H_u$ is $\set{e\setminus\set{u}\mid \EC(u)}$.
Note that for every $u\in Q_j$, has degree exactly $d_u$ which is
at least $W_j$, and therefore, $H_u$ has at least $W_j$ edges.
By the induction hypothesis, for each $H_u$ there exists a simple vector $P_u\in \Prodd_{k-1}(m'')$, where $m''=\frac{W_j}{\beta^{(k-1)^2}}$, such that $\Pr{\XC(H_u[P_u])=1}\geq (\eo)^{k-2}$.
We show that there exists at least one vector $Y\in \Prodd_{k-1}(m'')$, such that $\Pr{\XC(H_u[Y])=1}\geq (\eo)^{k}$, for at least $2^j$ vertices $u\in Q_j$.
We then show that the vector $Y$, which is a $(k-1)$-dimensional vector, can be extended to a $k$-dimensional vector $Z$ for which we have $\Pr{\XC(V[Z])=1}\geq (\eo)^{k}$, which completes the proof.

We first explain why $Y$ exists.
For every vector $P\in \Prodd_{k-1}(m'')$, define a set of vertices $Q_{j,P}\subseteq Q_j$ as follows.
A vertex $u\in Q_j$ is in $Q_{j,P}$ if $\Pr{\XC(H_u[P])=1}\geq (\eo)^{k-1}$.
Next, we use the pigeonhole principle, where vertices are pigeons, and vectors are holes, to claim that there exists at least one set $Q_{j,P}$ that contains at least $2^j$ vertices.
Note that the set of vectors $\Prodd_{k-1}(m'')$ has size at most $(\log n + 2)^{k-1}\leq (\log 4n)^{k-1}=\gamma$, while the set $Q_j$ has at least $2^j\gamma$ vertices.
Therefore, there exists a vector $Y\in \Prodd_{k-1}(m'')$ such that $\abs{Q_{j,Y}}\geq 2^j$.

To complete the proof, we extend the vector $Y$ that belongs to $\Prodd_{k-1}(m'')$ into a vector $Z\in \Prodd_k(m')$ as follows.
If $Y=(p_1,\ldots, p_{k-1})$, then $Z=(p_1,\ldots, p_{k-1},p_k)$, where $p_k=2^{-j}$.
To see that $Z\in \Prodd_k(m')$, note that
\begin{align*}
    w(Z)
    =w(Y)/2^j
    \leq \frac{1}{2^j\cdot m''}
= \frac{\beta^{(k-1)^2}}{2^j\cdot W_j}
    = \frac{\beta^{(k-1)^2}}{B}
    = \frac{1}{m'}\;.
\end{align*}
Next, we prove that $\Pr{\XC(V[Z])=1}\geq (\eo)^{k}$.
We define a random subset of vertices $R\subseteq V\setminus V_k$, obtained by sampling each vertex $v\in V_i$ with probability $p_i$, for every $i\in[k-1]$, where $p_i$ is the $i$-th coordinate of $Y$.
We claim that for every vertex $u\in Q_{j,Y}$, we have that $R$ contains an edge from $H_u$ with probability at least $(\eo)^{k-1}$, this is because the event that $R$ contains an edge from $H_u$ is exactly the event that $\XC(H_u[Y])=1$.
Next, we define $R_k$ to be a random subset of vertices obtained by sampling each vertex $v\in V_k$ with probability $p_k=2^{-j}$.
Finally, we define the following events. For every nonempty set $S\subseteq Q_{j,Y}$ we define the event $\Em{1}(S)$ to be the event that $\set{R_k=S}$, and the event $\Em{2}(S)$ to be the event that $R$ contains at least one edge from $H_u$, for at least one vertex $u\in S$.

Next, we show that
\begin{align}
    \Pr{\XC(V[Z])=1}\geq \sum_{S: \emptyset\subsetneq S\subseteq Q_{j,Y}} \Pr{\Em{1}(S)\cap \Em{2}(S)} \geq (\eo)^{k-1} \label{eq4:v Pk}\;.
\end{align}
For every fixed nonempty set $S\subseteq Q_{j,Y}$, the events $\Em{1}(S)$ and $\Em{2}(S)$ are independent, since $\Em{1}(S)$ addresses sampling vertices from $V_k$, while $\Em{2}(S)$ addresses sampling vertices from $V-V_k$, where the two samples are independent of each other.
Also note that $\Em{1}(S)\cap \Em{2}(S)$ is contained in the event that $\sset{\XC(V[Z])=1}$, and since the events $\set{\Em{1}(S)}_{S\subseteq Q_{j,Y}}$ are pairwise disjoint, so are the events $\set{\Em{1}(S)\cap \Em{2}(S)}_{S\subseteq Q_{j,Y}}$.
Therefore, the event that $\sset{\XC(V[Z])=1}$, contains the union of the following pairwise disjoint events $\sset{\bigcup_{S: \emptyset\subsetneq S\subseteq Q_{j,Y}}\Em{1}(S)\cap\Em{2}(S)}$.
We get
\begin{align*}
    \Pr{\XC(V[Z]=1)}
     & \geq \Pr{\bigcup_{S: \emptyset\subsetneq S\subseteq Q_{j,Y}}\Em{1}(S)\cap\Em{2}(S)} \\
     & = \sum_{S: \emptyset\subsetneq S\subseteq Q_{j,Y}} \Pr{\Em{1}(S)\cap \Em{2}(S)}     \\
     & = \sum_{S: \emptyset\subsetneq S\subseteq Q_{j,Y}} \Pr{\Em{1}(S)}\Pr{\Em{2}(S)}\;.
\end{align*}
To complete the proof, we need to show that
\begin{enumerate}
    \item $\sum_{S: \emptyset\subsetneq S\subseteq Q_{j,Y}}\Pr{\Em{1}(S)}\geq \eo$
    \item for every $S\subseteq Q_{j,Y}$, which is not empty, we have $\Pr{\Em{2}(S)}\geq (\eo)^{k-2}$.
\end{enumerate}
The first claim follows as
\begin{align*}
    \sum_{S: \emptyset\subsetneq S\subseteq Q_j}\Pr{\Em{1}(S)}=\Pr{Q_j[2^{-j}]\nemp}=1-(1-2^{-j})^{\abs{Q_j}}\geq \eo\;,
\end{align*}
where the inequalities hold for the following reasons. The first two equalities follows from definition, and the last inequality follows from the assumption that $\abs{Q_j}\geq 2^j$.

For the second claim, note that we previously explained why for every $u\in Q_{j,Y}$, we have that $R$ contains at least one edge from $H_u$ with probability at least $(\eo)^{k-2}$.
Therefore, since $S$ is a non-empty set that is contained in $Q_{j,Y}$,
it contains at least one vertex $u$ for which $R$ contains an edge from $H_u$ with probability at least $(\eo)^{k-2}$, which means that $\Pr{\Em{2}(S)}\geq (\eo)^{k-2}$.
This completes the proof of \Cref{eq4:v Pk}.

~\\\textbf{Proving that either $Q_0$ is not empty, or
    there exists $j\in\zrnone{K_0}$ such that $\abs{Q_j}\geq 2^j\cdot \gamma$.}
Assume towards a contradiction that this is not the case, that is, that $\abs{\EC}= m=m' \cdot \beta^{k^2}$ and that $Q_0$ is empty, and for every $j\in\zrnone{K_0}$, we have $\abs{Q_j}<2^j\cdot \gamma$.
Then,
\begin{align*}
    m & =   \sum_{u\in[V_h]}d_u                                                   \\
      & <   \sum_{u\in[V_h]}\sum_{1\leq j\leq K_0}\chi_{\set{u\in Q_j}}\cdot 2W_j \\
      & =   \sum_{1\leq j\leq K_0}\sum_{u\in[V_2]}\chi_{\set{u\in Q_j}}\cdot 2W_j \\
      & =   \sum_{1\leq j\leq K_0}\abs{Q_j}\cdot 2W_j
\end{align*}
Because we assumed towards a contradiction that for every $j\in\zrnone{K_0}$ we have $\abs{Q_j}<2^j \gamma$, we get
\begin{align*}
    \sum_{1\leq j\leq K_0}\abs{Q_j}\cdot 2W_j
    < \sum_{1\leq j\leq K_0}2^j \gamma \cdot \frac{2B}{2^j}
    = K_0\cdot \gamma \cdot 2B
\end{align*}
We have that
\begin{enumerate}
    \item $K_0\leq \log (4n) = \beta/k$
    \item $\gamma= (\log (4n))^{k-1} \leq (\beta/k)^{k-1}$
    \item $B= m'\cdot \beta^{(k-1)^2}$
\end{enumerate}
Therefore, we get that
\begin{align*}
    K_0\cdot \gamma \cdot 2B
    \leq (1/k)^{k}\cdot \beta^{k+(k-1)^2} \cdot m'
    = (1/k)^{k}\cdot \beta^{k^2-k+1} \cdot m'
<m'\cdot \beta^{k^2-1}<m\;,
\end{align*}
which is a contradiction.

~\\We have thus proved the assumption and this completes the proof of the induction step and therefore of \Cref{lemma:sampling tag}

\section{Applications to specific problems via hyperedge oracles}
\label{sec:reduction}

\renewcommand{\phi}{\varphi}

\subsection{$k$-Cliques}\label{ssec:clique}
In this subsection, we explain how to approximate the number $m$ of $k$-cliques in a graph $H=(V_H,E_H)$.

First, we will employ a standard reduction from $k$-clique to $k$-clique in a $k$-partite graph:
\begin{proposition}
    Given a graph $H=(V_H,E_H)$ with $|V|=n$, one can create in $O(n^2)$ time a $k$-partite graph $H'=(V',E_{H'})$ such that the number of cliques in $G$ is exactly $k!$ times the number of cliques in $H$.
\end{proposition}

The proof of the proposition is folklore: We set $V'$ to be $V_1\cup V_2\cup\ldots\cup V_k$ where each $V_i$ is an independent set of size $n$ and every vertex $v\in V_H$ has a copy $v_i$ in $V_i$ for each $i\in [k]$. For every edge $(u,v)\in E_H$ we add an edge $(u_i,v_j)$ to $E_{H'}$ for every $i\neq j$, $i,j\in [k]$. It is clear that every $k$-clique in $H'$ corresponds to a $k$-clique in $H$ and that every $k$-clique in $H$ has $k!$ copies in $H'$, one for every permutation of its vertices.

From now on we assume that our given graph is $k$-partite.
The main result of this section is the following theorem.
\begin{theorem}[$k$-clique]\label{theorem:clique}
    Let $G$ be a given graph with $n$ vertices and let $k\geq 3$ be a fixed integer.
    There is a randomized algorithm that outputs an approximation $\htt$ for the number $m$ of $k$-cliques set in $G$ such that
    $\Pr{(1-\eps)m \leq \htt\leq (1+\eps)m}\geq 1-1/n^2$, for any constant $\eps>0$. The running time is bounded by $\TO{f_{\mathrm{clique}}(n^k/m)}$.
\end{theorem}

The function $f_{\mathrm{clique}}(M)$ is defined in \Cref{f:clique} below.
We also get the following explicit upper bound on $f_{\mathrm{clique}}(M)$.
\begin{corollary}\label{cor:clique}
    Let $G$ be a graph with $n$ vertices and $n^t$ $k$-cliques.
    Let $r=k-t$.
    If $r\geq 2$, then,
\[  f_{\mathrm{clique}}(n^{r}) \leq n^2+ \max\set{
        n^{\frac{\omega(r-1,r-1,r+2)}{3}},n^{\frac{\omega(r-2,r+1,r+1)}{3}}}
        \leq n^2+n^{\frac{\omega(r-1,r,r+2)}{3}}\;. \]

\end{corollary}

Combining \Cref{cor:clique} with \Cref{theorem:clique} we obtain:

\cliqueSimp* 

We also get a better bound for $k=4$.
\tKfour* 

To prove \Cref{theorem:clique} we consider the $k$-clique approximate counting in the $k$-partite graph $H'$ as a hyperedge approximate counting problem on the implicit $k$-partite hypergraph $G$ defined on the same vertex set as $H'$ and such that a $k$-tuple of vertices is a hyperedge in $G$ if and only if they are a $k$-clique in $H'$. We provide a \HO for $G$ and
analyze its running time.
Let $f_{\mathrm{clique}}(M)$ denote the running time of our \HO on a set $U\subseteq V(G)$ with $\gs(U)\leq M$. We provide an upper bound on $f_{\mathrm{clique}}(M)$ and
then we use \Cref{thm4:main} together with our detection oracle to get a $\apm$ approximation for $m$, the number of $k$-cliques in $H'$ (which in turn gives us a $\apm$ approximation for $k!m$ as well, the number of $k$-cliques in the original graph $H$).

We also show that there's a one to one correspondence between detection oracles for $k$-clique and $\apm$ approximation algorithms:
We show that if there exists an algorithm that computes a $\apm$ approximation for the number of $k$-cliques $m$ in a $k$-partite $G$ in time $T$, assuming that $m\geq \tau$ for some value $\tau$, then there exists a detection oracle for $k$-cliques on sets $U$ with $\gs(U)\leq \TO{\gs(V)/\tau}$ operating in time $\TO{T}$.

\paragraph{Oracle Implementation.}
Let $H'=(V',E_{H'})$ be a $k$-partite graph with vertex set $V'=V_1\sqcup V_2 \sqcup \ldots \sqcup V_k$, and let $E'$ denote the set of $k$-cliques in $H'$, where every $k$-clique $e\in E'$ is a $k$-tuple of vertices, $e=(v_1,v_2,\ldots,v_k)$, such that $v_i\in V_i$ for every $i\in[k]$.
This defines an implicit $k$-partite hypergraph $G=(V',E')$, as discussed before.

Given a set of vertices $U\subseteq V'$, we explain how to implement the $\BB$ on the hypergraph $G[U]$, i.e., determine if $G[U]$ contains a hyperedge, or equivalently if $H'[U]$ contains a $k$-clique.
To do so, we modify a well-known reduction from $k$-clique detection to triangle detection by Ne\v{s}etril and Poljak \cite{nesetril}. The original reduction from $k$-clique detection starting from an $n$-node graph $H$ forms a new auxiliary graph whose vertices are those $k/3$-tuples of vertices of $H$ that are also $k/3$ cliques (if $k/3$ is not an integer, there are three types of nodes: $\lceil k/3\rceil$, $\floor k/3\rfloor$ and $k-\lceil k/3\rceil - \lfloor k/3\rfloor$-cliques). Then two $k/3$-cliques (i.e. nodes in the auxiliary graph) are connected by an edge if and only if together they form a $2k/3$ clique.

Our reduction is similar but with a twist. We start with a $k$-partite graph on parts $U_1,\ldots,U_k$ and instead of creating nodes corresponding to cliques of size $k/3$, we consider different sized cliques which have nodes from a pre-specified subset of the $U_i$s.

More formally, we create a new graph $F$ as follows.
The vertex set $U$ of the oracle input is partitioned into $k$ sets $U_1,U_2,\ldots,U_k$, where $U_i= V_i\cap U$ for every $i\in[k]$.
We then partition the index set $[k]$ into three sets, $I_1,I_2,I_3$, and define the vertex set of $F$ as follows. Define
\begin{align*}
    W_1=\bigtimes_{j\in I_1} U_j \;, &  & W_2=\bigtimes_{j\in I_2} U_j \;, &  & W_3=\bigtimes_{j\in I_3} U_j \;.
\end{align*}

Note that $W_i$ consists of $|I_i|$-tuples of vertices of $U$, one from each $U_j$ with $j\in I_i$.

The vertex set of $F$ is composed of three sets $A_1,A_2,A_3$, where $A_i\subseteq W_i$ for $i\in [3]$ defined as follows.
For every $i\in [3]$ and every element $S=(s_1,s_2,\ldots,s_{|I_i|})\in W_i$
, we add $S$ to the set $A_i$ if $H'[S]$ is a clique. So $A_i$ consists of those $|I_i|$-tuples of vertices that are $|I_i|$-cliques and that have exactly one vertex from $U_j$ with $j\in I_i$.

For any two vertices $S\in A_i,P\in A_j$ for $i\neq j$, let $V_{S,P}\subseteq V'$ denote a subset of vertices from $H'$, which are in the vertex tuples $S$ or $P$.
We add an edge $(S,P)$ to $F$ if $H'[V_{S,P}]$ is a clique.
We emphasize that if both $S$ and $P$ are in the same set $A_i$, then $H'[V_{S,P}]$ is not a clique,
and therefore the edge $(S,P)$ is not in $F$, making $F$ a tripartite graph. Also, since $A_i$ and $A_j$ are defined on disjoint subsets of $U$ when $i\neq j$, $H'[V_{S,P}]$ is a clique if and only if it is a clique of size $|I_i|+|I_j|$. Thus, any triangle in $F$ corresponds to a clique of size $|I_1|+|I_2|+|I_3|=k$.
As $H'[U]$ is $k$-partite, every $k$-clique also corresponds to a triangle in $F$, and thus the number of triangles in $F$ equals the number of $k$-cliques in $H'[U]$.

In what follows, we explain how to verify whether $F$ contains a triangle or not. We do so using rectangular matrix multiplication on $F$ as follows.
For $i,j\in[3]$, where $i\neq j$, we define a matrix $M_{i,j}$ of size $\abs{A_i}\times \abs{A_j}$, where the entry $M_{i,j}[x,y]$ is equal to $1$ if the $x$-th vertex in $A_i$ is connected to the $y$-th vertex in $A_j$ in $F$, and otherwise $M_{i,j}[x,y]=0$.
We compute the product $D=M_{1,2}M_{2,3}M_{3,1}$ and check whether the trace of $D$ is non-zero.
The output of the \HO is $1$ if the trace of $D$ is non-zero, and $0$ otherwise.

\paragraph*{Running Time.}
In the above description, we are free to pick the partition of the index set $[k]$ nto $I_1,I_2,I_3$ and we will do so as to minimize the running time.
There are at most $3^k$ such partitions, so finding the best one takes $O(3^k)$ time.
The running time for a specific partition $I_1,I_2,I_3$ is at most
$\MM{\abs{A_1},\abs{A_2},\abs{A_3}}$, plus the time to construct the matrices $M_{1,2},M_{2,3},M_{3,1}$, where both are upper bounded by $\MM{\abs{W_1}, \abs{W_2}, \abs{W_3}}$.

For any positive value $M$, let $f_{\mathrm{clique}}(M)$ denote the worst-case running time of the algorithm on a set $U$ with $\gs(U)\leq M$.
We need the following notation to provide an upper bound on $f_{\mathrm{clique}}(M)$.
Let $P_3([k])$ denote the set of all partitions of $[k]$ into three sets.
For $P\in P_3([k])$, where $P=(I_1,I_2,I_3)$, and a set $U=(U_1,U_2,\ldots, U_k)$, we define a triplet of numbers $x_1(P,U),x_2(P,U),x_3(P,U)$ as follows.
\begin{align*}
    x_j(P,U)= & \prod_{\ell\in I_j}\abs{U_\ell}\;, \quad \text{for } \ell\in[3]\;.
\end{align*}
We thus have the following upper bound on $f_{\mathrm{clique}}(M)$, where recall that $\gs(U)=\prod_{i=1}^3x_i(P,U)$:
\begin{align}
    f_{\mathrm{clique}}(M)=\max_{U:\gs(U)\leq M} \min_{P\in P_3([k])} \set{\MM{x_1(P,U),x_2(P,U),x_3(P,U)}}\;. \label{f:clique}
\end{align}
\begin{remark}
    For any non-negative $z$, the best we could hope for using this technique is
    $f_{\mathrm{clique}}(n^z) = n^{\tfrac{\omega\cdot z}{3}}$.
\end{remark}

\paragraph{Proof of \Cref{cor:clique}.}

To prove \Cref{cor:clique}, we need to analyze $f_{\mathrm{clique}}(n^r)$.
We need the following two properties.
First, that there exists a partition $P=(I_1,I_2,I_3)$ of $[k]$ into three parts such that
\begin{align}
    x_1(P,U)\leq x_2(P,U)\leq x_3(P,U)\leq n\cdot x_1(P,U)\;. \label{eq:3-batch}
\end{align}
Second, that for any such partition $P$ that satisfies \Cref{eq:3-batch}, we have that
\begin{align}
    \MM{x_1(P,U),x_2(P,U),x_3(P,U)}\leq \max\{n^{\omega(r-1,r-1,r+2)/3},n^{\omega(r-2,r+1,r+1)/3}\}\;, \label{eq:prop2}
\end{align}
where $n^{r}=\prod_{i\in[k]}\abs{U_i}$, and we use $n$ to denote the number of vertices in $V$.

Finally, we want to show that
\begin{align}\omega(r-1,r-1,r+2),\omega(r-2,r+1,r+1)\leq \omega(r-2,r,r+2)\; \label{eq:prop3}\end{align}

To prove the property in \Cref{eq:3-batch}, we analyze the greedy algorithm that partitions $[k]$ into three parts that are as ``balanced'' as possible.
We prove that there exists a partition $P$, that satisfies \Cref{eq:3-batch}.
Proving such a partition exists, is equivalent to showing that the set of numbers $R=\set{r_i}_{i\in [k]}$ where $r_i=\log_n(n_i)$,
with $\max(R)\leq 1$, can be partitioned into three sets $R_1,R_2,R_3$ such that
\begin{align*}
    \sum_{r\in R_1}r\leq \sum_{r\in R_2}r\leq \sum_{r\in R_3}r\leq 1+\sum_{r\in R_1}r\;.
\end{align*}
We prove this using a greedy strategy:
Start with a partition in which $R_1$ and $R_2$ are empty, and $R_3=R$.
While there are two sets $R_i,R_j$ such that $\sum_{r\in R_i}r + 1<\sum_{r\in R_j}r$, move the largest element from $R_j$ to $R_i$.
Assuming the process stops, the three obtained sets satisfy the required conditions (up to renaming the sets).
We prove the process always stops by induction on the number of elements.
The base case, when there is a single element is trivial.
For the step, assume the process stops for $k-1$ elements, and consider $k$ elements. Take one element out, denoted by $r_1$, and apply the induction hypothesis. Assume that when the process stops, there are three sets $R_1,R_2$ and $R_3$ which satisfy the required conditions.
We add the element $r_1$ to the set that has the smallest sum, i.e. $R_1$.
We now use case analysis. Let $w(R)$ denote $\sum_{r\in R}r$.
\begin{description}
    \item[Case 1:] $w(R_1)+r_1\leq w(R_2)$. Then the claim follows.
    \item[Case 2:] $w(R_1)+r_1\leq w(R_3)$. Then the claim follows by renaming: $R_1'\gets R_2$ and $R_2'\gets R_1 + \set{r_1}$.
    \item[Case 3:] $w(R_1)+r_1> w(R_3)$. Then the claim follows by renaming:$R_1'\gets R_2$, $R_2'\gets R_3$ and $R_3\gets R_1 + \set{r_1}$. Note that 
        $w(R_1') +1\geq w(R_1) +r_1 = w(R_3')$ as required.
\end{description}

Next we prove property \ref{eq:prop2} .
We have to show that for every $P\in P_3([k])$,
that satisfies \Cref{eq:3-batch}, we have that
\begin{align*}
    \MM{x_1(P,U),x_2(P,U),x_3(P,U)}\leq \max\{n^{\omega(r-1,r-1,r+2)/3},n^{\omega(r-2,r+1,r+1)/3}\}\;.
\end{align*}
For simplicity, we use $n^a,n^b,n^c$ to denote $x_1(P,U),x_2(P,U),x_3(P,U)$, respectively, where without loss of generality we assume that $a\leq b\leq c\leq a+1$.

Recall that the number of $k$-cliques is $n^t$ and $r=k-t$. As $\gs(U)=n^{a+b+c}\leq n^r$, we will assume that $a+b+c=r$ as that will maximize $\omega(a,b,c)$ and thus maximize $\MM{x_1(P,U),x_2(P,U),x_3(P,U)}=n^{\omega(a,b,c)}$.

Using the constraint $a+b+c=r$, we get that the set of feasible solutions for the constraints is a 2D plane in 3D space.
Each point on this plane can be written as a convex combination of three points obtained by setting $a=b=c$ ,$a=b\wedge c=a+1$, and $b=c=a+1$, which are
\begin{align*}
    \tfrac{1}{3}\cdot(r,r,r)\;,       &  &
    \tfrac{1}{3}\cdot(r-2,r+1,r+1)\;, &  &
    \tfrac{1}{3}\cdot(r-1,r-1,r+2)\;.
\end{align*}

In what follows we bound omega on a point $(x,y,z)$ that is in the set of feasible solutions.
We claim that the maximum value of $\omega(x,y,z)$ is obtained at either the point $(\tfrac{r-2}{3},\tfrac{r+1}{3},\tfrac{r+1}{3})$ or $(\tfrac{r-1}{3},\tfrac{r-1}{3},\tfrac{r+2}{3})$.
This follows as the function $\omega$ is convex in each of the variables $x,y,z$, and therefore the maximum value is obtained at one of the corners of the feasible set. We can ignore the point $(r/3,r/3,r/3)$ as the function omega is smaller as the arguments are closer to each other, as shown in \Cref{thm:MM balance is slower} below.
We get that the maximum value of $\omega(x,y,z)$ is obtained at either $(\tfrac{r-2}{3},\tfrac{r+1}{3},\tfrac{r+1}{3})$ or $(\tfrac{r-1}{3},\tfrac{r-1}{3},\tfrac{r+2}{3})$.

Finally, we would like to upper bound $\omega(r-1,r-1,r+2)$ and $\omega(r-2,r+1,r+1)$ by $\omega(r-2,r,r+2)$.

We use the following simple claim:
\begin{claim}\label{thm:MM balance is slower}
For every $y,x,z\geq 0$, $x>y$ and every $\eps \in [0,(x-y)]$, $\omega(x-\eps,y+\eps,z)\leq \omega(x,y,z)$.
\end{claim}
\begin{proof}
    Notice that
    \[x-\eps = \left(1-\frac{\eps}{x-y}\right)\cdot x + \frac{\eps}{x-y}\cdot y \textrm{ and } y+\eps = \left(1-\frac{\eps}{x-y}\right)\cdot y + \frac{\eps}{x-y}\cdot x.\]
    By the convexity of $\omega(\cdot)$ and because $\omega(x,y,z)=\omega(y,x,z)$ we get
    \[\omega(x-\eps,y+\eps,z)\leq \left(1-\frac{\eps}{x-y}\right) \cdot \omega(x,y,z)+ \frac{\eps}{x-y}\cdot \omega(y,x,z)=\omega(x,y,z).\]
\end{proof}

Applying \Cref{thm:MM balance is slower} with $x=r+2, y=r$ and $\eps=1$ we get that $\omega(r-2,r+1,r+1)\leq \omega(r-2,r,r+2)$.

Applying \Cref{thm:MM balance is slower} with $x=r, y=r-2$ and $\eps=1$ we get that $\omega(r-1,r-1,r+2)\leq \omega(r-2,r,r+2)$.

This completes the proof of \Cref{cor:clique}.

~\\We are left with proving \Cref{thms:k4}.
\begin{proof}[Proof of \Cref{thms:k4}]
    Fix $k=4$ and $r=k-t$ where $n^t$ is the number of $k$-cliques we want to approximate.
    We aim for an better upper bound on $f_{\mathrm{clique}}(n^r)$ for $r> 2$ since, as explained in the introduction, the interesting regime is when $t< k-2=2$.

    When $r\leq 2$ we just upper bound the running time by $\tilde{O_\eps}(n^2)$.

    We identify every set $U$ with $\gs(U)\leq n^r$ with a four dimensional point $(a_1,a_2,a_3,a_4)$ where $a_i$ is such that $n^{a_i}=\abs{U_i}$ for every $i\in[4]$.
    We also assume without the loss of generality that $a_1\leq a_2\leq a_3\leq a_4$ (up to renaming the color classes), where we know that $a_1+a_2+a_3+a_4=r$.
    We analyze the running time of the oracle on a set $U$, under the partition $P=(I_1,I_2,I_3)$, where $I_1=\{1,2\}$, $I_2=\{3\}$ and $I_3=\{4\}$.
    In other words, the two smallest set are ``packed'' together, and the two largest sets are separated.
    Let $S$ denote our universe where
    \begin{align*}
        S=  & \sset{(a_1,a_2,a_3,a_4)\in [0,1]^4\mid \sum_{i=1}^4 a_i=r\;\wedge a_1\leq a_2\leq a_3\leq a_4}\;. \\
        S'= & \set{(a_1+a_2,a_3,a_4)\mid (a_1,a_2,a_3,a_4)\in S}\;.
    \end{align*}
Let $A=a_1+a_2$.
    We need to upper bound the function $\omega(A,a_3,a_4)$ for $(A,a_3,a_4)$ under the following constraints:
    \begin{enumerate}
        \item $A+a_3+a_4=r$.

\item $A\geq 0$.
        \item $A\leq 2a_3$
        \item $a_3\leq a_4$.
        \item $a_4\leq 1$.
    \end{enumerate}

    The linear inequalities on $A,a_3,a_4$ define a convex set and since $\omega$ is convex function over the points  $(A,a_3,a_4)$, its maxima are attained at corners.

    The corners are:
    \begin{itemize}
        \item (constraints 5 and 4 tight): $a_4=a_3=1, A=r-2$, giving $\omega(r-2,1,1)$.
        \item (constraints 5 and 3 tight): $a_4=1,a_3=(r-1)/3, A=2(r-1)/3$, giving $\omega(2(r-1)/3,(r-1)/3,1)$.
        \item (constraints 5 and 2 tight): $a_4=1, A=0, a_3=r-1$ -- not a corner since we assumed that $r>2$ so $r-1>1$ and $a_3\leq a_4\leq 1$ is violated.
        \item (constraints 4 and 3 tight): $a_3=a_4=A/2=r/4$, giving $\omega(r/2,r/4,r/4)$.
        \item (constraints 4 and 2 tight): $a_4=a_3=r/2$ -- not feasible.
        \item (constraints 3 and 2 tight): $a_4=r$ -- not feasible.

    \end{itemize}

\end{proof}

\paragraph{From Fast Approximate Counting to a Fast Oracle.}
For any positive value $\tau$ the following is true.
We assume that there exists an algorithm $\textbf{A}$ that takes as input any graph $F$ and a threshold $\tau$, and outputs a $\apm$ approximation for $m_F$ the number of $k$-cliques in $F$ in time $T$, assuming that $m_F\geq \tau$.
We show that the existence of $\textbf{A}$ implies the existence of a detection oracle for $k$-cliques which takes time $\TO{T}$ on sets $U$ with $\gs(U)\leq \TO{\gs(V)/\tau}$.
In other words, if our detection-oracle implementation for $k$-cliques is optimal, then so is our approximate $k$-clique counting algorithm, up to polylogarithmic factors.

\begin{theorem}\label{thm:dup clique}
    Assume there exists an algorithm $\textbf{A}$ that takes as input any graph $F$ and a threshold $\tau$. 
    If $m_F\geq \tau$ the algorithm outputs a $\apm$ approximation for $m_F$ in time $T$. If $m_F<\tau$ the algorithm either outputs a $\apm$ approximation for $m_F$ or output $\bot$. 
    Then, there exists a detection oracle for $k$-cliques which takes time $\TO{T}$ on sets $U$ with $\gs(U)\leq \TO{\gs(V)/\tau}$.
\end{theorem}
\begin{proof}[Proof of \Cref{thm:dup clique}]
The input graph is $H=(V,E)$, and we would like to implement a detection oracle for $k$-cliques on sets $U\subseteq V$ with $\gs(U)\leq \gs(V)/\tau$, using the algorithm $\textbf{A}$.
    Let $H'=(V',E')$ be the $k$-partite graph obtained from the reduction.
    We explain how to implement the detection oracle on $H'$, which also implies the detection oracle on $H$.
Recall that $H'$ has the vertex set $V'=V_1\sqcup V_2 \sqcup \ldots \sqcup V_k$ and the edge set $E_{H'}$, and that we denote the set of $k$-cliques in $H'$ by $E'$.
    Let $U\subseteq V'$ be such that $\gs(U)\leq \gs(V')/\tau$.
    Recall that $E'_U$ denotes the set of $k$-cliques in the induced graph $H'[U]$.
    We would like to determine if $E'_U$ is empty or not.
To answer the query ``is $E'_U$ is empty'' we build a new auxiliary graph $F$, and run the algorithm $\textbf{A}$ on $F$.
    If $\textbf{A}$ outputs that $F$ does not contain any $k$-cliques, then we output that $E'_U$ is empty, and if $\textbf{A}$ outputs that $F$ contains at least $\tau$ $k$-cliques, then we output that $E'_U$ is not empty.

    We explain how to construct the graph $F$.
    For every $i$ let $U_i=U\cap V_i$. We define the vertex set $V^F$ of $F$ to be $V^F=V^F_1\sqcup V^F_2\sqcup\ldots\sqcup V^F_k$, where $V^F_i$ is obtained by duplicating each vertex in $U_i$ for $\ceil*{\abs{V_i}/\abs{U_i}}$ times.
    For $i\in [k]$, let $s_i=\ceil*{\abs{V_i}/\abs{U_i}}$, and for every vertex $v\in U_i$, let $\phi_v:[s_i]\to V^F_i$ be a function that maps any index $j\in[s_i]$ to the $j$-th duplicate of $v$ in $V^F_i$.
For every pair of vertices $x'',y''\in V^F$, where $x''$ and $y''$
are duplicates of two vertices $x',y'\in U$, we connect $x''$ and $y''$ in $F$ if and only if $x'$ and $y'$ are connected in $H'$.
    This means that $F$ is still a $k$-partite graph.
    Constructing $F$ takes time $O(\abs{V'}^2)$.
    Let $E''$ denote the set of $k$-cliques in $F$.

    We show that the following holds.
    $E'_U$ is empty if and only if $E''$ is empty.
    Moreover, if $E'_U$ is not empty, then $E''$ contains at least $\tau$ elements.
    Clearly, if $F$ contains a $k$-clique $C'=(v'_1,v'_2,\ldots,v'_k)$, where $v_i\in V^F_i$ for every $i\in[k]$, then the set of vertices $C$ such that $v_i$ is the original vertex in $U_i$ that $v_i'$ was duplicated from, is a $k$-clique in $H'$. This proves that if $E'_U$ is empty, then so is $E''$.
    The exact same argument shows that if $E'_U$ contains a $k$-clique $C$, where $C=(v_1,v_2,\ldots,v_k)$,
    then $C'=(\phi_{v_1}(i_1),\phi_{v_2}(i_2),\ldots,\phi_{v_k}(i_k))$ is a $k$-clique in $F$, for every choice of $i_1,i_2,\ldots,i_k\in[s_1]\times[s_2]\times\ldots\times[s_k]$.
    To complete the proof, we need to show that $\prod_{i\in[k]}s_i\geq \tau$.
    We have
    \begin{align*}
        \prod_{i\in[k]}s_i=\prod_{i\in [k]}\abs{V_i}/\abs{U_i}=\gs(V')/\gs(U)\;,
    \end{align*}
    and because we assumed that $\gs(U)\leq \gs(V')/\tau$, we have that $\prod_{i\in[k]}s_i\geq \tau$, which completes the proof.
\end{proof}
\subsubsection{Conditional Hardness}
Two popular hardness hypotheses in fine-grained complexity are that detecting a $k$-clique takes $n^{\omega(\floor{k/3},\ceil{k/3},k-\ceil{k/3}-\floor{k/3})-o(1)}$ time, and that any combinatorial algorithm for $k$-clique detection takes $n^{k-o(1)}$ time \cite{focsy}.
In \cite{jin2022tight} an equivalent variant of the combinatorial hypothesis is introduced:

\begin{hypothesis}[Unbalanced Combinatorial $k$-Clique Hypothesis{\cite[Hypothesis 2.1]{jin2022tight}}]\label{hyp0:clique}
    Let $G$ be a $k$-partite graph with vertex set $V=V_1\sqcup V_2 \sqcup \ldots \sqcup V_k$, where $\abs{V_i}=n^{a_i}$ for every $i\in[k]$.
    Let $s = \sum_{i=1}^k a_i$.
    In the word-RAM model with $O(\log n)$ bits words, any randomized algorithm for $k$-clique detection in $G$ needs $n^{s-o(1)}$ time.
\end{hypothesis}
A natural extension of \Cref{hyp0:clique} to any algorithm instead of only combinatorial algorithms is the following.
\begin{hypothesis}\label{hyp:clique}
    Let $G$ be a $k$-partite graph with vertex set $V=V_1\sqcup V_2 \sqcup \ldots \sqcup V_k$, where $\abs{V_i}=n^{a_i}$ for every $i\in[k]$.
    Let $s = \sum_{i=1}^k a_i$.
    In the word-RAM model with $O(\log n)$ bits words, any randomized algorithm for $k$-clique detection in $G$ needs $n^{\frac{\omega\cdot s}{3}-o(1)}$ time.
\end{hypothesis}

We show that assuming \Cref{hyp:clique} is true the running time of our algorithm is close to optimal.
\begin{theorem}
Assume \Cref{hyp:clique} and the Word-RAM model of computation with $O(\log n)$ bit words. Let $k\geq 3$ be any constant integer.
    Any randomized algorithm that, when given an $n$ node graph $G$, can
    distinguish between $G$ being $k$-clique-free and containing $\geq n^t$ copies of $k$-cliques needs
    $n^{\omega(k-t)/3 - o(1)}$ time.
\end{theorem}
\begin{proof}
    We start with a $k$-partite graph $G=(V,E)$ with vertex set $V=V_1\sqcup V_2 \sqcup \ldots \sqcup V_k$, where $\abs{V_i}=n^{a_i}$ for every $i\in[k]$. Let $s=\sum_{i=1}^k a_i$.
    We build a new graph $G'=(V',E')$ as follows.
    For every $i\in[k]$, we duplicate every vertex in $V_i$ for ${n^{1-a_i}}$ times. For every vertex $v$, we denote its duplicates by $(v,1),(v,2),\ldots$.
    For every edge $(u,v)\in E$, we add the edge $((u,i),(v,j))$ to $E'$ between every copy of $u$ and every copy of $v$.

    We show that $G$ has a $k$-clique if and only if $G'$ has $n^t$ $k$-cliques, where $t=k-s$.
    For the first direction, note that every $k$-clique in $G$ denoted by $C=(v_1,v_2,\ldots,v_k)$, where $v_i\in V_i$ for every $i\in[k]$, corresponds to $n^t$ copies of a $k$-clique in $G'$, by replacing each $v_i$ with one of its $n^{1-a_i}$ duplicates.
    Therefore, the number of $k$-cliques in $G'$ is
    \begin{align*}
        \prod_{i\in[k]}n^{1-a_i}=n^{k-s}=n^t\;.
    \end{align*}
    For the other direction, we show that if $G'$ has a single $k$-clique $C'=(v'_1,v'_2,\ldots,v'_k)$, then $G$ also have a $k$-clique.
    For every $i\in[k]$, let $v_i$ be the vertex in $V_i$ that $v'_i$ was duplicated from.
    Then the set of vertices $C=(v_1,v_2,\ldots,v_k)$ is a $k$-clique in $G$ by definition. This completes the proof of the statements that $G$ has a $k$-clique if and only if $G'$ has $n^t$ $k$-cliques, where $t=k-s$.

    To complete the proof of the theorem, we assume towards a contradiction that there exists an algorithm that can distinguish between $G'$ being $k$-clique-free and containing $\geq n^t$ copies of $k$-cliques in time $n^{\omega(k-t)/3 - \delta}$, for some constant positive $\delta$.
    Recall that we defined $t$ to be $k-s$, and therefore 
    $k-t=s$.
    However, such an algorithm can be used to distinguish between $G$ being $k$-clique-free and containing a $k$-clique in time $n^{\omega\cdot s/3 - \delta}$, which contradicts \Cref{hyp:clique}.
\end{proof}

We conclude that assuming \Cref{hyp:clique}, the running time of our algorithm which is bounded by $O(n^2+n^{\frac{\omega(s-1,s,s+2)}{3}})$ 
is close to optimal, since the best we could hope for is $O(n^{\frac{\omega\cdot s}{3}-o(1)})$.

\subsection{$k$-Dominating Sets}
In this subsection, we explain how to approximate the number $m$ of $k$-dominating sets (henceforth $k$-DS) in a graph $H=(V,E_H)$.

First, we will employ a standard reduction from $k$-DS to $k$-DS in a $k$-partite graph:
\begin{proposition}\label{ds:reduction}
    Given a graph $H=(V_H,E_H)$ with $|V|=n$, one can create in $O(n^2)$ time a $k$-partite graph $H'=(V',E_{H'})$ such that the number of DS in $G$ is exactly $k!$ times the number of DS in $H$.
\end{proposition}

From now on we assume that our given graph is $k$-partite.
The main result of this section is the following theorem.
\begin{theorem}[$k$-dominating sets]\label{theorem:ds}
    Let $G$ be a given graph with $n$ vertices and let $k\geq 3$ be a fixed integer.
    There is a randomized algorithm that outputs an approximation $\htt$ for the number $m$ of $k$-dominating sets in $G$ such that
    $\Pr{(1-\eps)m \leq \htt\leq (1+\eps)m}\geq 1-1/n^2$, for any constant $\eps>0$. The running time is bounded by $\TO{f_{\mathrm{DS}}(n^k/m)}$.
\end{theorem}

The function $f_{\mathrm{DS}}(M)$ is defined in \Cref{f:ds} below.
We also get the following explicit upper bound on $f_{\mathrm{ds}}(M)$.
\begin{corollary}\label{cor:f-ds}
    For every $r\geq 2$, we have that
    \begin{align*}
        f_{\mathrm{ds}}(n^r)\leq n^{\omega(1,(r-1)/2,(r+1)/2)}\;.
    \end{align*}
\end{corollary}

Combining \Cref{cor:f-ds} with \Cref{theorem:ds} we obtain:
\DSSimplified*

\begin{proof}[Proof of \Cref{ds:reduction}]
    We start with a graph $H=(V,E_H)$ with $n$ vertices, and define a new graph $H'=(V',E_{H'})$.
    We count the number of \emph{special} $k$-DS in $H'$, which we shortly define.
    The key point here is that unlike the previous reduction, we don't need the number of $k$-DS in $H'$ to be equal to the number of $k$-DS in $H$ times a fixed factor. Instead, we can count a subset of the $k$-DS in $H'$.

    The vertex set of $H'$ is $V'=V\times [k]$, and the edge set $E_{H'}$ is defined as follows. For every edge $(u,v)\in E_H$,
    we add the edge $((u,i),(v,j))$ to $E_{H'}$ for every $i,j\in[k]$, where unlike the previous reduction, we allow for $i=j$.
    We define two functions $\phi:V\to[k]$ and $\chi:V\to A$ as follows. For every vertex $v'=(v,i)\in V'$, we set $\phi(v')=i$ and $\chi(v')=v$.
    We say that a set $e\subseteq V'$ is \emph{colorful} if $\phi$ is injective on $e$, and say that $e$ is \emph{distinct} if $\chi$ is injective on $e$.
    We say that $e$ is \emph{special} if $e$ is both colorful and distinct.

    Let $E$ denote the set of $k$-DS in $H$ and let $m$ denote its size.
    In addition, let $E'$ denote the set of \emph{special} $k$-DS in $H'$ and let $m'$ denote its size.
    We encode every $k$-DS in $H$ as a $k$-tuple of vertices, $e=(v_1,v_2,\ldots,v_k)$, ordered lexicographically.

    Let $\Pi_k$ denote the set of all permutations of $[k]$.
    For every $k$-DS $C=(v_1,v_2,\ldots,v_k)$ in $H$, where the vertices in $C$ are ordered lexicographically, and every permutation $\pi\in \Pi_k$, we use the notation $C_\pi=((v_1,\pi(1)),(v_2,\pi(2)),\ldots,(v_k,\pi(k)))$.

    We show the following holds.
    $C\in E$ if and only if there exists a permutation $\pi\in \Pi_k$ such that $C_\pi\in E'$.

    To see the first direction, let $C$ be a $k$-DS in $H$.
    We show that for every $\pi\in \Pi_k$, the set $C_\pi$ is a special $k$-DS in $H'$.
    Clearly, $C_\pi$ is colorful as $\phi$ is injective on $C_\pi$, and $C_\pi$ is distinct as $\chi$ is injective on $C_\pi$, as $\chi(C_\pi)=C$, because $C$ contains $k$ distinct vertices.
    We are left with proving that $C_\pi$ is a $k$-DS in $H'$.
    We fix some vertex $u'\in V'$ and show it is in $C_\pi$ or has a neighbor in $C_\pi$. If it is in $C_\pi$ then the proof is complete, and we therefore assume that $u'\notin C_\pi$.
    We use $u'=(u,i)$, i.e., $\chi(u')=u$ and $\phi(u')=i$.
    Since $C$ is a $k$-DS in $H$, there exists a vertex $v\in C$ which is a neighbor of $u$.
    Therefore, there must be a vertex $v'\in C_\pi$ for which $\chi(v')=v$, which means that $v'$ is a neighbor of $u'$ in $H'$, which proves that $C_\pi$ is a $k$-DS in $H'$.

    To see the other direction, let $C_\pi$ be a special $k$-DS in $H'$.
    We prove that $C$ is a $k$-DS in $H$.
    Note that $C=\sset{\chi(v')\mid v'\in C_\pi}$. Since $C_\pi$ is distinct, the set $C$ contains exactly $k$ (distinct) vertices.
    We are left with proving that $C$ is a $k$-DS in $H$.
    We fix some vertex $u\in V$ and show that it is in $C$ or has a neighbor in $C$. If it is in $C$ then the proof is complete, and we therefore assume that $u\notin C$.
    We look at the vertex $(u,1)$ in $H'$, which is either in $C_\pi$ or has a neighbor in $C_\pi$. If $(u,1)\in C_\pi$ then $u\in C$, so we assume that $(u,1)\notin C_\pi$.
    Since $C_\pi$ is dominating, there exists a vertex $v'\in C_\pi$ which is a neighbor of $(u,1)$.
    Let $v=\chi(v')$, and note that $v\in C$ and $v$ is a neighbor of $u$ in $H$, which proves that $C$ is dominating.

\end{proof}

This completes the reduction from a general graph $H$ to a new graph $H'$. We have that if $\hat{m}$ is a $\apm$ approximation for the number of special $k$-DS in $H'$, then $\hat{m}/(k!)$ is a $\apm$ approximation for the number of $k$-DS in $H$, as required.

\paragraph{Oracle Implementation.}
We explain how to implement the \HO on a $k$-partite hypergraph $G$ obtained from $H'$ by considering every special $k$-DS as a hyperedge.
That is we provide an algorithm that takes as input a set $U\subseteq V'$, and determines if $H'[U]$ contains a special $k$-DS. The oracle has access to $H$, $H'$ and the two functions $\phi$ and $\chi$, and can evaluate them on any vertex in $H'$ in $O(1)$ time.

To do so, we use a well-known reduction from the $k$-DS problem to matrix multiplication, shown in \cite{eisenbrand2004complexity}.
In other words, we first find all $k$-DS in $H'$ and we then filter the non-special ones from this set.

The vertex set $U$ is partitioned in $k$ sets $U_1,U_2,\ldots,U_k$, where $U_i= V_i\cap U$ for every $i\in[k]$, where $V_i=V\times\set{i}$.
We partition the index set $[k]$ into two sets, $I_1,I_2$, and build a new graph $F$ as follows. We define two supersets $W_1,W_2$ as follows.
\begin{align*}
    W_1=\bigtimes_{j\in I_1}U_j \;, &  & W_2=\bigtimes_{j\in I_2}U_j \;.
\end{align*}
We say that an element $S\in W_i$ is special if it is colorful and distinct.
Let $A_i$ denote the set of all special elements in $W_i$ for every $i\in[2]$.
Note that every element $S\in W_i$ is a subset of $\abs{I_i}$ vertices, which is also colorful. However, $S$ might not be distinct -- the function $\chi$ might not be injective on $S$.
We keep in $A_i$ only the distinct elements in $W_i$.
For $i\in[2]$ we define a new matrix $M_i$ of size $\abs{A_i}\times \abs{V}$, where the entry $M_i[x,y]$ is equal to $0$ if the $y$-th vertex in $V$ has a neighbor in the $x$-th element in $A_i$, and otherwise it is equal to $1$.
Let $(M_2)^T$ denote the transpose of $M_2$.
We then compute the product $D=M_1(M_2)^T$.
The following holds:
\begin{align*}
    D[i,j]=\sum_{h=1}^{n} M_1[i,h]\cdot (M_2)^T[h,j]\;,
\end{align*}
and therefore $D[i,j]=0$ if and only if for every $h\in[n]$, the $h$-th vertex in $V$ is connected to the $i$-th element in $A_1$ or to the $j$-th element in $A_2$.
For every zero entry $D[i,j]=0$, the union of the $i$-th element in $A_1$ and the $j$-th element in $A_2$ is a colorful $k$-DS in $H'$, however, it might not be distinct. In what follows we explain how to filter the non-distinct sets.

Fix $i$ and $j$, and let $S$ denote the $i$-element in $A_1$, and $S'$ denote the $j$-element in $A_2$.
Then $S\cup S'$ is a colorful $k$-DS in $H'$ if and only if $D[i,j]=0$.
Testing whether $S\cup S'$ is distinct takes $O(1)$ time, as we can evaluate $\chi$ on $S\cup S'$ in $O(1)$ time, and $S\cup S'$ is distinct if and only if $\chi$ is injective on $S\cup S'$.
Therefore, The number of special $k$-DS in $H'$ is equal to the number of zero entries in $D$, which are associated with a distinct set of vertices.
Finding all such sets take at most $O(\abs{A_1}\cdot \abs{A_2})$ time.
Therefore, the output of the \HO is $1$ if we find a zero entry in $D$, which is associated with a special $k$-DS in $H'$, and it is $0$ otherwise.

\paragraph{Running Time.}
The running time of the oracle is $\MM{n,\abs{A_1},\abs{A_2}}$, which is bounded by $\MM{n,\abs{W_1},\abs{W_2}}$.
We choose the partition of the index set $[k]$ into $I_1,I_2$ that minimizes the running time of the oracle.

For any positive value $M$, let $f_{\mathrm{DS}}(M)$ denote the worst-case running time of the oracle on a set $U$ with $\gs(U)\leq M$.
We need the following notation to provide an upper bound on $f_{\mathrm{DS}}(M)$.
Let $P_2([k])$ denote the set of all partitions of $[k]$ into two sets.
For $P\in P_2([k])$, where $P=(I_1,I_2)$, and a set $U=(U_1,U_2,\ldots, U_k)$, we define a pair of numbers $x_1(P,U),x_2(P,U)$ as follows.
\begin{align*}
    x_1(P,U)= & \prod_{i\in I_1}\abs{U_i}\;, & x_2(P,U)= & \prod_{i\in I_2}\abs{U_i}\;.
\end{align*}

We have the following upper bound on $f_{\mathrm{DS}}(M)$:
\begin{align}
    f_{\mathrm{DS}}(M)=\max_{U:\gs(U)\leq M} \min_{P\in P_2([k])} \set{\MM{x_1(P,U),x_1(P,U),n}}\;.\label{f:ds}
\end{align}
This concludes the proof of \Cref{theorem:ds}.
In what follows we prove \Cref{cor:f-ds}.

For every partition $P\in P_2([k])$, where $P=(I_1,I_2)$ let $W_i=\bigtimes_{j\in I_i}\abs{U_j}$ for every $i\in[2]$, and define $X_i=\abs{W_i}$.
We therefore can always find a partition $I_1,I_2$ such that $X_1\leq X_2\leq n\cdot X_1$, using a greedy algorithm.
Because $X_1\cdot X_2=M$ and $X_2\leq X_1\cdot n$ there exists $s\in[0,1/2]$ such that $X_1=\sqrt{M}\cdot n^{s-1/2}$ and $X_2=\sqrt{M}\cdot n^{1/2-s}$.
We therefore have
\begin{align*}
    f_{\mathrm{DS}}(M)\leq \MM{n,\sqrt{M/n}\cdot n^s,\sqrt{M\cdot n}/n^{s}}
    \leq \MM{n,\sqrt{M/n},\sqrt{Mn}}\;,
\end{align*}
and to complete the proof of \Cref{cor:f-ds}, we plug in $M=n^r$.

\begin{remark}
    The best we can hope for is $f_{\mathrm{DS}}(M)=\MM{n,\sqrt{M},\sqrt{M}}$.
\end{remark}

We are left with proving \Cref{thms:ds3}
\begin{proof}
    We prove that for $k=3$ it holds that
    $f_{\mathrm{DS}}(n^{3-t})\leq n+n^{\psi}$, where
    \begin{align*}
        \psi=\max\set{\omega(1,1-t/3,2(t-1)/3),\omega(1,1,2-t)}\leq \omega(1,1,2(1-t/3))\;.
    \end{align*}
We are given a set $U\subseteq V$ where $U$ is partitioned into $3$ sets $U_1,U_2,U_3$, where $U_i=V_i\cap U$ for every $i\in[3]$.
    We assume that $\gs(U)= n^{3-t}$, for some $t\leq 2$.
    We use $n^{s_i}$ to denote $\abs{U_i}$ for every $i\in[3]$, where we assume without the loss of generality that $s_1\leq s_2\leq s_3$,
    as well as that $0\leq s_1\leq s_2\leq s_3\leq 1$, and that $s_1+s_2+s_3=3-t$.

    To prove the theorem, we partition the index set $[3]$ into two sets $I_1,I_2$, by choosing $I_1=\set{1,2}$ and $I_2=\set{3}$.
    This induces two vertex sets $W_1,W_2$, where $W_1=U_1\times U_2$ and $W_2=U_3$.
    We have that $\abs{W_1}= n^{s_3}$, and $\abs{W_2}= n^{s_1+s_2}$, and therefore $f_{\mathrm{DS}}(n^{3-t})\leq \MM{n,n^{s_3},n^{s_1+s_2}}$.
    We can write $s_1 + s_2 = 3-t - s_3$, and therefore
    \begin{align*}
        f_{\mathrm{DS}}(n^{3-t})\leq n^{\omega(1,s_3,3-t-s_3)}\;.
    \end{align*}
    Note that due to the convexity of omega, we have that
    the maximum of $\omega(1,s_3,3-t-s_3)$ is achieved at the endpoints, which are when $s_3$ is maximal, i.e., $s_3=1$, or when $s_3$ is minimal, i.e., $s_3=(3-t)/3=1-t/3$, which completes the proof.
\end{proof}

\subsection{$k$-Sum}
In this subsection, we explain how to approximate the number $m$ of $k$-sums in an array $A$ of $n$ numbers, where a $k$-sum is a $k$-tuple of elements in $A$ that sum to $0$.

First, we will employ a standard reduction from $k$-sum to colorful $k$-sum, which is defined as follows.
Let $A=A_1,A_2\ldots A_k$ be a set of $k$-arrays of integers and let $k\geq 3$ be a fixed integer.
Compute a $\apm$ approximation for the number of colorful $k$-sums,
where a $k$-tuple $s=(a_1,a_2,\ldots,a_k)$ is considered as a colorful $k$-sum if $a_i\in A_i$ for every $i\in[k]$, and $\sum_{i=1}^{k}a_i=0$.

For that, we prove the following proposition.
\begin{proposition}
    Given an array $A$ of $n$ numbers, one can create in $O(n)$ time $k$ arrays $A_1,A_2,\ldots,A_k$ such that the number of $k$-sums in $A$ is exactly $k!$ times the number of (colorful) $k$-sums in $A_1,A_2,\ldots,A_k$.
\end{proposition}
\begin{proof}
    We start with an instance $A$ of the $k$-Sum problem, which has $m$ $k$-sums, and get a new instance $A'$ of the Colorful-$k$-Sum problem, which has $m'$ colorful $k$-sums, where $m'= (k!) m$. The reduction is oblivious to the $m$.

    Let $A$ be an array of $n$ integers, and let $k\geq 3$ be a fixed integer.
    Assume that all elements in $A$ fall inside the interval $[-U,U]$, where $U$ is a positive number.
    We build a new instance $A'$ of the Colorful-$k$-Sum problem as follows.
    We define a new array $A_i=A + 3^i \cdot 10U$ for every $i\in[k-1]$, and let $A_k=A - U'$, where $U'=10U\cdot \sum_{i=1}^{k-1}3^i$.
    This completes the reduction.

    For any solution $s=(a_1,a_2,\ldots,a_k)$ to the $k$-Sum problem, where $a_i\in A$ for every $i\in[k]$, and every permutation $\pi\in \Pi_k$,
    we have that $s'=(a_{\pi(1)}+3^1\cdot 10U,a_{\pi(2)}+3^2\cdot 10U,\ldots,a_{\pi(k)}-10U')$ is a solution to the Colorful-$k$-Sum problem.
    Note that the offsets we added to each $A_i$ ensures that
    for every $s\subseteq A'$ that satisfies $\sum_{a\in s}a_i=0$, means that $s$ is colorful.
\end{proof}

From now on we assume that our given input contains k-arrays. The main result of this section is the following theorem.

\sumSimp*

The function $f_{\mathrm{clique}}(M)$ is defined in \Cref{f:ksum} below.
We also get the following explicit upper bound on $f_{\mathrm{sum}}(M)$.
\begin{corollary}\label{cor:Sum}
    Let $A$ be an array of $n$ integers, and let $k\geq 3$ be a fixed integer.
    Assume that $A$ contains $n^t$ solutions for the $k$-sum problem.
    Let $r=k-t$.
If $r\geq 2$, then,
    \[  f_{\mathrm{sum}}(n^{r}) \leq \TO{n^{(r+1)/2}}\;. \]
\end{corollary}

Our final result in this section is the following theorem.
\sumThree* 

To prove \Cref{thms:sum} we consider the colorful $k$-sum approximate counting over the arrays $A'=(A_1,A_2,\ldots, A_k)$ as a hyperedge approximate counting problem on the implicit $k$-partite hypergraph $G$, whose vertices are the elements of $A'$ and the hyperedges are exactly the set of colorful $k$-sums. We provide a \HO for $G$ and analyze its running time.
Let $f_{\mathrm{sum}}(M)$ denote the running time of our \HO on a set $U\subseteq V(G)$ with $\gs(U)\leq M$. We provide an upper bound on $f_{\mathrm{sum}}(M)$ and then we use \Cref{thm4:main} together with our detection oracle to get a $\apm$ approximation for $m$, the number of colorful $k$-sum in $H'$ (which in turn gives us a $\apm$ approximation for $k!m$ as well, the number of $k$-sum in the original array $A$).

\paragraph{Oracle Implementation.}
Let $V$ be a set of $k$-arrays $V_1, V_2, \ldots ,V_k$, and let $E$ denote the set of colorful $k$-tuple of numbers, $e=(v_1,v_2,\ldots,v_k)$, such that $v_i\in V_i$ for every $i\in[k]$, for which $\sum_{i=1}^{k}v_i=0$.
Then $G=(V,E)$ be a $k$-partite hypergraph.

Given a set of vertices $U\subseteq V$, we explain how to implement a query on $U$ over the hypergraph $G[U]$, i.e., determine if $G[U]$ contains a hyperedge, or equivalently if $U$ contains a $k$-tuple from $E$.
The vertex set $U$ is partitioned in $k$ pairwise disjoint sets $U_1,U_2,\ldots,U_k$, where $U_i= V_i\cap U$ for every $i\in[k]$.

To do so, we modify a well-known algorithm for the $k$-sum problem, which
first builds two sets $W_1,W_2$ as follows, by having $W_1=\bigtimes_{i\in[k/2]} U_i$ and $W_2=\bigtimes_{i\in\sset{k/2+1,\ldots, k}} U_i$, and for every element $s\in W_1\cup W_2$ computing $\sum_{a\in s}a$.
Let $R_1=\sset{\sum_{a\in s}a\mid s\in W_1}$ and $R_2=\sset{\sum_{a\in s}a\mid s\in W_2}$.
Assume without the loss of generality that $\abs{R_1}\leq \abs{R_2}$.
The algorithm proceeds by iterating over all elements in $R_1$, and for every $x\in R_1$ checking if $-x\in R_2$.
If such elements exist, the algorithm outputs $1$, and $0$ otherwise.
The running time of this algorithm is $O(\abs{R_1}\log(\abs{R_1})+\abs{R_2}\log(\abs{R_2}))$.

Our reduction is similar but with a twist. We start with a $k$-partite graph on parts $U_1,\ldots,U_k$ and instead of creating nodes corresponding to elements  of size $k/2$, we consider different sized tuples which have numbers from a pre-specified subset of the $U_i$s.

More formally, we create two new sets $W_1$ and $W_2$ as follows.
We partition the index set $[k]$ into two sets, $I_1,I_2$, and define the
two sets
\begin{align*}
    W_1=\bigtimes_{j\in I_1} U_j \;, &  & W_2=\bigtimes_{j\in I_2} U_j \;.
\end{align*}
Note that $W_i$ consists of $|I_i|$-tuples of elements from $U$, one from each $U_j$ with $j\in I_i$.
We then define two sets $R_1,R_2$ as before:
\begin{align*}
    R_j=\sset{\sum_{a\in s}a\mid s\in W_j}\;, \quad j\in\set{1,2}\;.
\end{align*}
From this point forward our reduction is the same as the standard one.
It involves sorting the sets $R_1$ and $R_2$, and then checking if
for every $x\in R_1$ there exists an element $-x\in R_2$.

\paragraph*{Running Time Analysis.}
We choose the partition of the index set $[k]$ into $I_1,I_2$ that minimizes the running time of the algorithm.
Let $f_{\mathrm{sum}}(M)$ denote the running time of the algorithm for a set $U$ with $\gs(U)\leq M$. We provide an upper bound on $f_{\mathrm{sum}}(M)$.
Let $\UU$ denote the collection of sets $U$ with $\gs(U)\leq M$ for some value $M$. Then the running time of the algorithm is
\begin{align}
    f_{\mathrm{sum}}(M)=\max_{U:\gs(U)\leq M} \min_{P\in P_2([k])} \set{x_1(P,U),x_2(P,U)}\label{f:ksum}\;,
\end{align}
Where recall that $P_2([k])$ is the set of all partitions of $[k]$ into two sets.
This completes the proof of \Cref{thms:sum}.

In the following, we provide an explicit bound on $f_{\mathrm{sum}}(M)$, which proves \Cref{cor:Sum}.
We can provide an explicit bound on $f_{\mathrm{sum}}(M)$, using the fact that $\abs{U_i}\leq n$ for every $i\in[k]$.
We therefore can always find a partition $I_1,I_2$ where $\abs{R_1}\leq \abs{R_2}\leq \abs{R_1}\cdot n$.
This follows by a greedy approach that starts with a trivial partition where all elements are in $I_2$ and while $\abs{R_2}\geq\abs{R_1}$, picks an index from $I_2$ and moves it to $I_1$.
By the previous discussion, we can assume that there exists a partition $I_1,I_2$ that induces two sets $R_1, R_2$ such that $\abs{R_1}\cdot\abs{R_2} =M$ and  $\abs{R_2}\leq \abs{R_1}\leq \abs{R_2}\cdot n$, where $n=\abs{V}$.
We can also write $\abs{R_1}$ and $\abs{R_2}$ as follows.
There exists $s\in[0,1/2]$ such that $\abs{R_1}=\sqrt{M}\cdot n^{s-1/2}$ and $\abs{R_2}=\sqrt{M}\cdot n^{1/2-s}$.
If $s=0$, then $\abs{R_1}=\abs{R_2}$ and if $s=1/2$, then $\abs{R_1}=\abs{R_2}\cdot n$.
The running time of the algorithm is therefore at most $O(\abs{R_2}\cdot \log(\abs{R_2}))$ which is at most $O(\sqrt{M}\cdot n^{1/2-s}\cdot \log(n))$.
By plugging in $M=n^k/m$ we get that the running time of the algorithm is at most $\sqrt{n^{k+1}/m}$.

\paragraph{Refined Analysis for Small Values of $k$.}
In what follows we prove \Cref{thms:sum3}.
\sumThree*

\begin{proof}[Proof of \Cref{thms:sum3}.]
    We prove that $f_{\mathrm{sum}}(n^{3-t})\leq n+n^{2(1-t/3)}$.
    In other words, we prove that given three arrays $A_1,A_2,A_3$, where
    $A_i$ has cardinality $n^s_i$ for every $i\in[3]$, and $s_1+s_2+s_3=3-t$,
    we can compute the number of colorful $3$-sums in $A_1,A_2,A_3$ in time $O(n+n^{2(1-t/3)})$.
    We assume without the loss of generality that $s_1\leq s_2\leq s_3$, and consider the partition $I_1=\sset{1,2},I_2=\sset{3}$.
    This partition induces two sets $W_1=A_1\times A_2, W_2=A_3$, and two sets $R_1,R_2$, where
    \begin{align*}
        R_1=\sset{\sum_{a\in s}a\mid s\in W_1}\;, &  & R_2=\sset{\sum_{a\in s}a\mid s\in W_2}\;.
    \end{align*}
    Clearly, $\abs{R_2}\leq n$.
    Moreover, we have that $\abs{R_1}\leq n^{s_1+s_2}$, and using the following facts
    \begin{enumerate}
        \item $s_1+s_2 +s_3=3-t$
        \item $0\leq s_1\leq s_2\leq s_3\leq 1$
    \end{enumerate}
    we get that the expression $s_1+s_2=3-t-s_3$ is maximized when $s_1$ is large as possible, and $s_3$ is as small as possible, which occurs when $s_1=s_2=s_3=(3-t)/3=1-t/3$.
    Therefore the running time of constructing and sorting $R_1$ is at most $O(n^{2(1-t/3)}\log n)$, and the running time of iterating over every element $x$ in $R_2$ and checking if $-x\in R_2$ takes at most $O(\abs{R_2})=O(n)$ time. Overall, the running time is at most $O(n+n^{2(1-t/3)})$, which completes the proof of \Cref{thms:sum3}.
\end{proof}

\subsection{General Problems}
We characterize a class of problems for which the reduction from approximate counting to oracle detection implementation holds.
We think of a ``problem'' $\PP$ as a family of functions taking as input a $k$-partite graph $H=(V,E_H)$, and a coloring $\phi:V\to[k]$, where $V_i$ denotes the set of vertices in $V$ that are mapped to $i$ by $\phi$.
The ``problem'' $\PP$ outputs a subset $E\subseteq \bigtimes_{i\in[k]}V_i$.
We write $\PP(H)=E$ to denote that $E$ is the output of $\PP$ on input $H$.
Note that $E\subset V^k$, and for every $e\in E$ we have $\phi(e)=[k]$.
We denote by $\CL$ the class of problems for which the reduction from approximate counting to oracle detection holds.
We explain which problems are in the class $\CL$.
Fix some $k$-partite graph $H=(V,E_H)$ a $k$-coloring $\phi$ and let $E=\PP(H)$.
Define a new graph $H'$ obtained from $H$ by adding to it a new vertex $v'_i$ which is a duplicate of $v_i\in V_i$. That is, we extend $\phi$ such that $\phi(v'_i)=\phi(v_i)$.
We also add every edge $(u,v'_i)$ to the graph, where $u\in N(v_i)$.
Let $E'=\PP(H')$.
Then the problem $\PP$ is in the class $\CL$ if the following holds.
For every $e\in E$ we have $e\in E'$.
Moreover, for any $e\in E$, if $v_i\in e$ and $e'=(e\setminus\set{v_i})\cup\set{v'_i}$ then $e'\in E'$.
In words, in every element $e\in E$ we can replace every vertex $v_i$ with its duplicate $v'_i$ and get a new element $e'$ that is also in $E'$.
~\\The main result of this subsection is the following.

\begin{theorem}\label{thm:duplicatable}
    Fix a problem $\PP$ taken from the class $\CL$. Let $H=(V,E_H)$ be any graph, and let $\phi$ be a $k$-coloring of $V$.
    If there exists an algorithm that can distinguish whether $\PP(H)=E$ is empty or contains at least $\tau$ elements, in time $T$, then there exists an implementation of the detection oracle in time $\TO{T}$ over sets $U$ with $\gs(U)\leq \gs(V)/\tau$.
\end{theorem}
\begin{corollary}\label{cor:duplicatable}
    For every problem $\PP$ in the class $\CL$, if there exists an algorithm that approximates the size of $\abs{E}$ in time $T$, assuming that $\abs{E}\geq \tau$, then there exists an implementation of the detection oracle in time $\TO{T}$ over sets $U$ with $\gs(U)\leq \gs(V)/\tau$.
\end{corollary}

\begin{proof}[Proof of \Cref{thm:duplicatable}]
    Let $H=(V,E_H)$ and let $\phi$ be a $k$-coloring of $V$, where $V_i$ denotes the set of vertices in $V$ that are mapped to $i$ by $\phi$.
    We also use $E=\PP(H)$ to denote a set of $k$ distinct vertices from $V$ where we would like to compute a $\apm$ approximation for $\abs{E}$.
Fix some positive value $\tau$.
    Let $U\subseteq V$ be such that $\gs(U)\leq \gs(V)/\tau$.
    We would like to determine if $E_U\subseteq E\cap U^k$.
    We assume we have an algorithm $\textbf{A}$ that takes as input a graph $H$ and can determine if $E=\PP(H)$ is empty or contains at least $\tau$ elements in time $T$.
    We have no guarantee on the output of the algorithm if $E$ is not empty and contains less than $\tau$ elements.

    To answer the query ``is $E_U$ empty'', we build a new auxiliary graph $H'$, and let $E'=\PP(H')$. We show that $E_U$ is empty iff $E'$ is empty. Moreover, if $E_U$ is not empty, then $E'$ contains at least $\tau$ elements.
    We then run the algorithm $\textbf{A}$ on the graph $H'$.
    If it outputs that $E'$ is empty, then we output that $E_U$ is empty as well, and if it outputs that $E'$ is not empty, then we output that $E_U$ is not empty.

    We explain how to construct the graph $H'$.
    For every $i$ let $U_i=U\cap V_i$. We define the vertex set $V'$ of $H'$ to be $V'=V_1'\sqcup V_2'\sqcup\ldots\sqcup V_k'$, where $V_i'$ is obtained by duplicating each vertex in $U_i$ for $\ceil*{\abs{V_i}/\abs{U_i}}$ times.
    For every pair of vertices $x',y'\in V'$, where $x'$ and $y'$ are duplicates of two vertices $x,y\in U$, we connect $x'$ and $y'$ if at least one of the following holds.
    \begin{enumerate}
        \item $x$ and $y$ are connected in $H$.
        \item $x=y$.
    \end{enumerate}
Constructing the graph $H'$ takes time $O(\abs{V}^2)$.
    Since $\PP$ is in the class $\CL$, we have that $E_U$ is empty if and only if $E'$ is empty.
    Moreover, if $E_U$ is not empty, then $E'$ contains at least $\prod_{i\in [k]}\abs{V_i}/\abs{U_i}=\gs(V)/\gs(U)$ elements.
    Therefore, for sets $U$ of measure $\gs(U)\leq \gs(V)/\tau$, we have that $E'$ is of size at least $\gs(V)/\gs(U)\geq \tau$, which means that if $E_U$ is not empty, then $E'$ contains at least $\tau$ elements.
    On the other hand if $E_U$ is empty, then $E'$ is empty as well.
    To see why, assume towards contradiction that $E'$ contains a an element $e=(u_1',u_2',\ldots,u_k')$, where $u_i'\in V_i'$ for every $i\in[k]$. Let $u_i$ be the ``original'' vertex in $U_i$ that $u_i'$ was duplicated from. Then $(u_1,u_2,\ldots,u_k)$ must be in $E_U$, since $\PP$ is in the class $\CL$, which is a contradiction.
\end{proof}

\bibliographystyle{alpha}

\addcontentsline{toc}{section}{Bibliography}
\bibliography{refs.bib}

\section{Proof of \Cref{lemma:simp B}} \label{sec:proof of B}
This section is dedicated to the proof of \Cref{lemma:simp B}, which follows the same lines as in \cite{CEW24}, and we repeat it here with only minor modifications.
\LemCoreSimpB* 

For this subsection, we use $q\triangleq p^k$.
We now prove the following lemma on the probability that the algorithm $\algCorezB{V,\Lambda}$ outputs a good approximation for the number of hyperedges in $G$, following the line of proof in \cite{CEW24}.
\begin{lemma}[Induction Hypothesis]\label{lemma2:ih}
    Given a hypergraph $G$ and $\eps$, we set $p=1/2$, and $\eps'=\neweps$.
Fix $\K$ where $\K=o(n)$.
Define the following events:
    \begin{align*}
        \Emia\triangleq\set {\algCorezB{V,\Lambda}=m_V(1\pm \eps')^{\DD} \pm \DD[\Lambda] \cdot(2\K\cdot\Lambda/\eps') }\;.
\end{align*}
We have that $\Pr{\Emia}\geq 1- \rPz$.
\end{lemma}
Next, we prove \Cref{lemma:simp B} using \Cref{lemma2:ih}.
\begin{proof}[Proof of \Cref{lemma:simp B} using \Cref{lemma2:ih}]
    Set $\K=\Qd/\DD$. We get that
    \begin{align*}
        \Pr{\Emia}\geq 1- \rPz=1-\frac{4k\DD^2}{q\Qd}
        \geq 1-\rPzB\;,
    \end{align*}
    which completes the proof.
\end{proof}
We prove \Cref{lemma2:ih} using induction on the depth of the recursion.
\paragraph*{Base Case: ($\DD=0$).}
In this case we simply return $m_V$.
The induction hypothesis holds trivially, as the algorithm always outputs the exact count of hyperedges in $G$.
\paragraph*{Induction Hypothesis:}
The induction hypothesis is that $\Pr{\Emia}\geq 1- \rPz$.

~\\We also define
\begin{align*}
    \Emib\triangleq \set{\algCorezB{U,\Lambda\cdot q}=m_U\apmp^{\DDF} \pm q\cdot\DDF\cdot (2\K\Lambda/\eps')}\;.
\end{align*}
For a subset of vertices $U\subseteq V$, we use the notation $m_U$ to denote the number of hyperedges in $G$ that are contained in $U$.
The statement $\Pr{\Emib}\geq 1-\frac{12\DDF}{\K q}$, is the induction hypothesis on the random subset of vertices $U$, which is defined in \Cref{line:5}.
To prove the step, we assume the induction hypothesis holds on $U$,
and show that the induction hypothesis also holds on $V$.

\paragraph*{Step:}
We first define some notation and events.
We have three subset of vertices: $U\subseteq V'\subseteq V$, where $V$ is the vertex set of $G$, $V'$ is the set of vertices without the $\Lambda$-heavy vertices, and $U$ is a uniformly random subset of $V'$, obtained by keeping each vertex from $V'$ independently with probability $p$.
Let $\htt_V$ denote the output of $\algCorezB{V,\Lambda}$, and let $\htt_U$ denote the output of $\algCorezB{U,\Lambda \cdot q}$.
Recall that we use $m_\Lambda$ to denote $m_V(V_\Lambda)$, i.e., the number of hyperedges (contained in $V$) that contain at least one vertex from $V_\Lambda$.
Let $\htt_\Lambda$ denote the output of $\Cialg(V_\Lambda, \Lambda)$.
Note that $\DD\geq \DDF+1$ (unless $\DD=0$, in which case $\DD=\DDF=0$).
\paragraph{Intuition.}
We compute two values, $\htt_\Lambda$ and $\htt_U$.
We then output $\htt_v=\htt_\Lambda + \htt_U/q$ as an approximation for $m_V$.
By the guarantees of the two procedures $\FindHalg$ and $\Cialg$,
we have that $\htt_\Lambda$ is a good approximation for $m_\Lambda$. What is left is to show that $\htt_U/q$ is a good approximation for $m_{V'}$. We split this into two parts. First, we show that $m_U/q$ is ``close'' to the value of $m_{V'}$ (See \Cref{cor42:e3}).
Next, we use the induction hypothesis, to show that $\htt_U$ is a good approximation of $m_U$. We need to show that the composition of these approximations also serves as a good approximation for $m_V$.

Let $\Emh$ denote the event that ${V_\Lambda}$ contains a superset of the $\Lambda$-heavy vertices, without any $\Lambda/\clog $-light vertices, and that $\htt_\Lambda=\apmp m_\Lambda$.
Note that $\Pr{\Emh}\geq \fpr{2}{4}$, which follows from the guarantees of the two procedures $\FindHalg$ and $\Cialg$.
Let $\Ems\triangleq\set{m_U/q=m_{V'}\apmp\pm \K \cdot \Lambda/\eps'}$. This denotes the event that the random variable $t_U$ is not too far from its expectation.
We prove that $\Pr{\Ems}\geq 1-\frac{6}{\K q }$.
Using those facts, we can complete the proof of the lemma by showing that
\begin{align}
    \set{\Emib\cap \Ems\cap \Emh}\subseteq \Emia\;. \label{eq:iterval art}
\end{align}
\begin{proof}[Proof of \Cref{eq:iterval art}]
\newcommand{\sm}{\sigma_{\times}}
    \newcommand{\sa}{\sigma_{+}}
Define three intervals $\II_V$, $\II_{V'}$ and $\II_U$ as follows:
    \begin{align*}
        \II_V    & \triangleq m_V(1\pm \eps')^{\DD} \pm 2\DD\K \cdot \Lambda/\eps'            \\
        \II_{V'} & \triangleq m_{V'}(1\pm \eps')^{\DD} \pm 2\DD\K \cdot \Lambda/\eps'         \\
        \II_U    & \triangleq m_U(1\pm \eps')^{\DDF} \pm 2\DDF\K \cdot (\Lambda\cdot q)/\eps'
    \end{align*}
    Recall that
    \begin{align*}
        \Emia & \triangleq \set{\htt_V=m_V(1\pm \eps')^{\DD} \pm 2\DD\K \cdot \Lambda/\eps' }              \\
        \Emib & \triangleq \set{\htt_U=m_U\apmp^{\DDF} \pm 2\DDF\cdot \K \cdot  (\Lambda\cdot q)/\eps'}\;,
    \end{align*}
    where the event $\Emia$ can also be written as $\set{\htt_V\in \II_V}$ and the event $\Emib$ can be written as $\set{\htt_U\in \II_U}$, which follows by their definitions.
Recall that $\haT_V=\htt_\Lambda+\htt_U/q$.
    We prove the following claim:
If $\htt_\Lambda$ falls inside the interval $t_\Lambda\apmp$, and $\htt_U/q$ falls inside the interval $\II_{V'}$, then $\haT_V$ falls inside the interval $\II_V$.
To prove this, we show that $t_\Lambda\apmp + \II_{V'}\subseteq \II_V$:
    \begin{align*}
        m_\Lambda(1\pm \eps') + \II_{V'} & = m_\Lambda(1\pm \eps') +m_{V'}(1 \pm \eps')^{\DD}    \pm \brak{2\DD \K\cdot \Lambda/\eps'}                  \\
                                         & \subseteq  (m_\Lambda + m_{V'}) (1\pm \eps')^{\DD} \;\quad \pm \brak{2\DD \K\cdot \Lambda/\eps'}             \\
                                         & =  m_V (1\pm \eps')^{\DD}                       \;\quad\quad\quad\quad \pm \brak{2\DD \K\cdot \Lambda/\eps'}
        = \II_V\;.
    \end{align*}

    In what follows, we show that each of the events $\set{\htt_\Lambda\in m_\Lambda\apmp}$, and $\set{\htt_U/q\in \II_{V'}}$ contains the event $\set{\Emib\cap \Ems\cap \Emh}$. In other words, we show that if the event $\set{\Emib\cap \Ems\cap \Emh}$ occurs, then so do the events $\set{\htt_\Lambda\in m_\Lambda\apmp}$ and $\set{\htt_U/q\in \II_{V'}}$.

    By definition, the event that $\set{\htt_\Lambda\in m_\Lambda\apmp}$ contains the event $\Emh$.
    We show that event that $\set{\htt_U/q\in \II_{V'}}$ contains the event $\set{\Emib\cap \Ems}$, and complete the proof.
    We assume that the event $\set{\Emib\cap \Ems\cap \Emh}$ occurs, and show that the event $\set{\htt_U/q\in \II_{V'}}$ occurs:
    \begin{align*}
        \dfrac{\htt_U}{q}
         & \overset{(i)}{\in} \dfrac{\II_U}{q}=
        \dfrac{m_U\apmp^{\DDF}\pm \brak{2\DDF\K q \Lambda/\eps'}}{q}                                                               \\
         & = \dfrac{m_U\apmp^{\DDF}}{q}                                                   & \pm\brak{\DDF\cdot(2\K \Lambda/\eps')} \\
         & \overset{(ii)}{=}\; \sbrak{m_{V'} \apmp\pm \K \cdot \Lambda/\eps'}\apmp^{\DDF} & \pm\brak{\DDF\cdot(2\K \Lambda/\eps')} \\
         & \subseteq m_{V'}\apmp^{\DD} \pm \brak{\K\Lambda/\eps' }\cdot \apmp^{\DDF}      & \pm\brak{\DDF\cdot(2\K \Lambda/\eps')} \\
         & \overset{(iii)}{\subseteq} m_{V'}\apmp^{\DD} \pm \brak{2\K\Lambda/\eps'}       & \pm\brak{\DDF\cdot(2\K \Lambda/\eps')} \\
         & = m_{V'}\apmp^{\DD} \pm \brak{2\K\Lambda/\eps'}(1+ \DDF)                                                                \\
         & = m_{V'}\apmp^{\DD} \pm\brak{\DD\cdot2\K \Lambda/\eps'} = \II_{V'}\;.
\end{align*}
    The $(i)$ inequality follows since we assumed the event ${\Emib}$ occurs.
    The $(ii)$ inequality follows since we assumed the event ${\Ems}$ occurs.
    The $(iii)$ transition is true because we chose sufficiently small $\eps'$ such that $(1\pm\eps')^{\DDF}$ is contained in both the interval $[-1/2,3/2]$ and the interval $[1-\eps/2,1+\eps/2]$.
    This completes the proof.
\end{proof}
\newcommand{\zza}{1-\frac{k}{2q\cdot \K}}
\newcommand{\zzb}{1-\frac{4k\cdot \DDF+1}{q\cdot \K}}

To complete the proof of \Cref{lemma2:ih},
we are left with proving that
\begin{align*}
    \Pr{\Emib\cap \Ems\cap \Emh}\geq 1- \rPz\;,
\end{align*}
which proves that $\set{\Emia\geq \rPz}$ and therefore completes the proof of \Cref{lemma2:ih}.
To show this, we need the following claim.
\begin{claim}\label{claim:order}
    Define two events as follows. $A_1=\set{\Ems\mid \Emh}, $ and $A_2=\set{\Emib\mid \Emh}$.
    Then, the following inequalities hold:
    \begin{enumerate}[label=(\arabic*)]
        \item $\Pr{A_1}\geq \zza$.
        \item $\Pr{A_2}\geq \zzb$.
        \item $\Pr{A_1\cap A_2}\geq 1- {\frac{4k\cdot \DD-3k + 1}{q\cdot  \K}}$.
    \end{enumerate}
\end{claim}
Proving this claim and using it to prove \Cref{lemma2:ih} mostly involves algebraic manipulations. The only non-trivial part is to show $(1)$, which we prove in \Cref{cor42:e3}.
We first prove \Cref{lemma2:ih} using \Cref{claim:order}.
\begin{proof}[Proof of \Cref{lemma2:ih} using \Cref{claim:order}]
    \begin{align*}
             & \Pr{\Emib\cap \Ems\cap \Emh}                             \\
        \geq & \Pr{\Emib\cap \Ems\mid \Emh}\cdot \Pr{\Emh}              \\
        =    & \Pr{A_1\cap A_2\mid \Emh}\cdot \Pr{\Emh}                 \\
        \geq & (1- \frac{4k\DD -3k + 1}{ q\cdot \K })\cdot (\fpr{2}{4})
        \geq  1- \rPz\;.
    \end{align*}
\end{proof}
Next, we prove that $(3)$ holds, given that $(1)$ and $(2)$ hold.
\begin{align*}
    \Pr{A_1\cap A_2}
    \geq 1-(\Pr{\overline{A_1}} + \Pr{\overline{A_2}})
    \geq 1- \frac{k + 4k\DDF + 1}{q\cdot \K }
    = 1- \frac{4k\DD -3k + 1}{ q\cdot \K }\;.
\end{align*}
The first inequality follows by union bound.
The first inequality follows by parts $(1)$ and $(2)$ of \Cref{claim:order}.
The last inequality follows as $\DD= \DDF+1$.
We are left with proving $(1)$ and $(2)$, which we prove in \Cref{claim:z00,cor42:e3} respectively.
\begin{claim}\label{claim:z00}
    $\Pr{A_2}\geq \zzb$.
\end{claim}
\begin{proof}
    \begin{align*}
        \Pr{A_2}=\Pr{\Emib \mid \Emh}
        =    & \Pr{\Emib \cap \Emh}/\Pr{\Emh}           \\
        \geq & \Pr{\Emib \cap \Emh}                     \\
        =    & 1-\Pr{\BEmib \cup \BEmh}                 \\
        \geq & 1-(\Pr{\BEmib} + \Pr{\BEmh})             \\
        \geq & 1- (\frac{4k\DDF}{\K q} + \frac{2}{n^4})
        \geq \zzb\;.
    \end{align*}
    The penultimate inequality follows as $\Pr{\Emib}\geq 1- \frac{4k\DDF}{\K q}$ by the induction hypotheses, and $\Pr{\Emh}\geq \fpr{2}{4}$.
    The last inequality follows from the fact that $\K =o(n)$.
\end{proof}
\newcommand{\tauH}{\tau_{H}}
\newcommand{\tHp}{t_{H}}
\newcommand{\tFp}{t_{F}}
\begin{claim}\label{cor42:e3}
    $\Pr{A_1}\geq 1-\frac{k\cdot}{2\K q}$.
\end{claim}
Recall that $A_1=\set{\Ems\mid \Emh}$, which is the event that the random variable $m_U$ is not too far from its expectation, assuming that the event $\Emh$ occurs - which is the event that $V'$ does not contain any $\Lambda$-heavy vertices.

To prove \Cref{cor42:e3} we will use also the following claim.

\begin{claim}\label{claim:variance}
    Let $G$ be a $k$-partite hypergraph with vertex sets $V_1\sqcup V_2\sqcup \ldots\sqcup V_k$, and a set of hyperedges $ E$, where every hyperedge $e\in E$ intersects every vertex set $V_i$ in exactly one vertex.
    We denote the number of vertices by $n$ and the number of hyperedges by $m$.
    Let $P=\set{p_1,p_2,\ldots,p_k}$ be some sampling vector with $p_i\in[0,1]$ for every $i\in[k]$.
    Let $q=\prod_{i=1}^k p_i$.
    Let $U$ be a random subset of vertices, where each vertex $v\in V_i$ is sampled independently with probability $p_i$.
    Let $X$ denote the number of hyperedges in $G[U]$.
    Then, $\Exp{X}=m\cdot q$ and $\Var{X}\leq k\cdot \Delta\cdot \Exp{X}$, where $\Delta$ is maximum degree of a vertex in $G$.
\end{claim}
\begin{proof}
    Computing the expectation of $X$ is immediate; $\Exp{X}=\sum_{e\in E}\Pr{e\in G[U]}=\sum_{e\in E}\prod_{i=1}^k p_i=m\cdot q$.
    Next, we compute $\Exp{X^2}$.
    \newcommand{\CM}{\mathcal{M}}
    Let $\CM$ denote the set of ordered pairs of hyperedges.
    We partition $\CM$ into two sets, $\CM_1$ and $\CM_2$, where $\CM_1$ contains the ordered pairs of hyperedges that are disjoint in vertices, and $\CM_2$ contains the ordered pairs of hyperedges that share at least one vertex.
    We have
    \begin{align*}
        \Exp{X^2}
        = \sum_{(e_1,e_2)\in\CM}\Pr{e_1, e_2\in G[U]}
    \end{align*}
    We partition this sum into two parts, one for the ordered pairs in $\CM_1$ and one for the ordered pairs in $\CM_2$.
    We have
    \begin{align*}
        \sum_{(e_1,e_2)\in\CM_1}\Pr{e_1, e_2\in G[U]}
        =q^2\cdot \abs{\CM_1}\leq q^2\cdot m^2=\Exp{X}^2\;.
    \end{align*}
    We are left with bounding the sum over the ordered pairs in $\CM_2$.
    We have
    \begin{align*}
        \sum_{(e_1,e_2)\in\CM_2}\Pr{e_1, e_2\in G[U]}
        \leq q\abs{\CM_2}\;.
    \end{align*}
Next, we demonstrate that the cardinality of $\CM_2$ is bounded by $k \cdot m \cdot \Delta$. We represent each pair of intersecting hyperedges from $\CM_2$ as a tuple within the space $[k] \times [m] \times [\Delta]$. To identify the first hyperedge in the pair, we use $m$ possibilities. Once the first hyperedge is established, we have $k$ choices to determine the first (in lexicographic order) vertex $u$ common to both hyperedges. Given vertex $u$, we then have $\deg(u)$ choices to identify the second hyperedge that includes $u$, where $\deg(u) \leq \Delta$. This completes the proof.
\end{proof}

We are now ready to prove \Cref{cor42:e3}.

\begin{proof}[Proof of \Cref{cor42:e3}]
    Recall that $\Ems\triangleq\set{m_U/q=m_{H}\apmp\pm \K \cdot \Lambda/\eps'}$. Assuming that $\Emh$ occurs, we have that $\Delta_{V'}$ which is
    the maximum degree over $V'$, is at most $\Lambda$.
    Let $\sigma=q(m_{V'}\eps + \K \Lambda/\eps)$.
    We can write $\set{A_1}$ as $\set{\abs{m_U/q-m_{V'}}\leq \sigma\mid \Emh}$:
    \begin{align*}
        \set{A_1}
         & =\set{m_U/q=m_{V'}\apmp\pm \K \cdot \Lambda/\eps'\mid \Emh}                  \\
         & =\set{\abs{m_U/q-m_{V'}}\leq m_{V'}\eps' \pm \K\cdot \Lambda/\eps'\mid \Emh} \\
         & =\set{\abs{m_U/q-m_{V'}}\leq \sigma\mid \Emh}\;.
    \end{align*}
    We apply Chebyshev's inequality on $m_U$ and get that
    \begin{align*}
        \Pr{\abs{m_U/q-m_{V'}} \leq \sigma}
        \leq \frac{\Var{m_U}}{\sigma^2}
        \leq \frac{k\cdot \Delta_{V'}\cdot m_{V'}\cdot q}{2q^{2}\cdot m_{V'}\cdot \K\cdot \Lambda}
        \leq \frac{k}{2q\K}\;.
\end{align*}
    The first inequality follows from Chebyshev's inequality, and the second inequality follows from the fact that $\Var{m_U}\leq k\cdot \Delta_{V'}\cdot m_{V'}$, as shown in \Cref{claim:variance}.
    The next inequality follows from the fact that $\Delta_{V'}\leq \Lambda$, and reordering the terms. This completes the proof.
\end{proof}

 \section{The Median Trick}
\label{app:median}
Here we prove a claim on the median of $r$ samples of a random variable $Y$,
which we use in the proof of \Cref{lemma:alg trunk}.
\begin{claim}\label{claim:med trick0}
    Let $Y$ be a random variable, and let $[a,b]$ be some interval, where $\Pr{Y\in [a,b]}\geq 2/3$.
    Define $Z$ to be the median of $r=\rr$ independent samples of $Y$. Then, $\Pr{Z\in [a,b]}\geq 1-\frac{1}{n^5}$.
\end{claim}
\begin{proof}Let $Y_i$ be the $i$-th sample. We define an indicator random variable $X_i$ to be $1$ if $Y_i\in[a,b]$ for $i\in [r]$.
    Let $X=\sum_{i\in[r]} X_i$.
    We claim that if more than half of the samples fall inside the interval $[a,b]$ then so is the median of all samples. In other words, the rest of the samples do not affect the median.
    Therefore, $\Pr{Z\in[a,b]}\geq \Pr{X>\frac{r}{2}}$.
    We lower bound the probability that $\Pr{X>\frac{r}{2}}$.
    Note that $X$ is the sum of $r$ independent $p$-Bernoulli random variables with $p\geq 2/3$. Then,
    \begin{align*}
        \Pr{X\leq\frac{r}{2}}
        \leq \Pr{\Bin{r}{\frac{2}{3}}\leq \frac{r}{2}}
        = \Pr{\Bin{r}{\frac{2}{3}}\leq \frac{2r}{3}\cdot(1-\frac{1}{4})}\;,
\end{align*}
    and by the following version of Chernoff's inequality:
    $\Pr{X\leq (1-\delta)\Exp{X}}\leq \exp\brak{-\frac{\delta^2\cdot \Exp{X}}{2}}$, we have
    \begin{align*}
        \Pr{\Bin{r}{\frac{2}{3}}\leq \frac{2r}{3}\cdot(1-\frac{1}{4})}
        \leq \exp\brak{-\frac{r}{48}}
        \leq \exp\brak{-8\log n}
        \leq n^{-5}\;.
    \end{align*}
\end{proof}

\section{Some Matrix Multiplication Tools}\label{sec:appendix}
Recall that we write $\MM{n^a,n^b,n^c}$ or $n^{\omega(a,b,c)}$ for the running time of the matrix multiplication algorithm that multiplies $n^a\times n^b$ by $n^b\times n^c$ matrices.
\begin{claim}\label{claim:a2}
    For any $p_1,p_2,p_3\in[0,1]$ we have
    $\MM{n p_1,n p_2,n p_3}\leq \MM{n,n ,n \cdot(p_1\cdot p_2\cdot p_3)}$.
\end{claim}

\begin{proof}[Proof of \Cref{claim:a2}]
    Let $p_i=n^{\ell_i}$ for $i=1,2,3$.
    We want to show that $\omega(1+\ell_1,1+\ell_2,1+\ell_3)\leq \omega(1,1,1+\ell_1+\ell_2+\ell_3)$. This follows from two applications of \Cref{thm:MM balance is slower} which states that for every $y,x,z\geq 0$, $x>y$ and every $\eps \in [0,(x-y)]$, $\omega(x-\eps,y+\eps,z)\leq \omega(x,y,z)$.

    First, we apply \Cref{thm:MM balance is slower} with $\eps=\ell_1$, $x=1+\ell_1+\ell_3, y=1,z=1+\ell_2$:
    \[\omega(1+\ell_1,1+\ell_2,1+\ell_3)=\omega(1+\ell_3, 1+\ell_1, 1+\ell_2,)\leq \omega(1+\ell_1+\ell_3,1,1+\ell_2)=\omega(1+\ell_1+\ell_3,1+\ell_2,1).\]

    We then apply  \Cref{thm:MM balance is slower} with $\eps=\ell_2$, $x=1+\ell_1+\ell_2+\ell_3, y=1,z=1$:
    \[\omega(1+\ell_1+\ell_3,1+\ell_2,1)\leq \omega(1+\ell_1+\ell_2+\ell_3,1,1)=\omega(1,1,1+\ell_1+\ell_2+\ell_3),\]
    and putting it all together we get $\omega(1+\ell_1,1+\ell_2,1+\ell_3)\leq \omega(1,1,1+\ell_1+\ell_2+\ell_3)$.

\end{proof}
\end{document}